\documentclass[12pt]{article}
\usepackage{mathrsfs}
\usepackage{amsfonts,amsmath,latexsym,amsthm}
\usepackage{graphicx} 
\usepackage{url, hyperref}
\usepackage{setspace}
\usepackage{natbib}
\bibpunct{(}{)}{;}{a}{}{,}
\usepackage{multirow}  
\usepackage{enumerate}  
\usepackage{pst-all}      
\usepackage{pst-poly}     
\usepackage{multido}      
\usepackage{bm}
\usepackage{arydshln,leftidx,mathtools}

\hypersetup{
    colorlinks,%
    citecolor=blue,%
    filecolor=black,%
    linkcolor=blue,%
    urlcolor=blue
}

\usepackage{setspace}

\newtheorem{theorem}{Theorem}
\newtheorem{lemma}{Lemma}
\newtheorem{proposition}{Proposition}

\usepackage{tikz}
\usetikzlibrary{decorations.pathreplacing}

\long\gdef\blind#1#2{\ifbld{\em \color{blue}#1}\else{#2}\fi}
\newif\ifbld \bldfalse

\title{Multiple Changepoint Detection with Partial Information on Changepoint Times}
\date{}
\author{  \blind{}{     
Yingbo Li
\footnote{Department of Mathematical Sciences, Clemson University, Clemson, SC 29634}, 
Robert Lund\footnotemark[\value{footnote}], and Anuradha Hewaarachchi\footnote{University of Kelaniya,
Kalenaiya, Sri Lanka.}
} 
}

\begin{document}
\maketitle


\begin{abstract}

This paper proposes a new minimum description length procedure to detect 
multiple changepoints in time series data when some times are a priori 
thought more likely to be changepoints.  This scenario arises with 
temperature time series homogenization pursuits, our focus here. Our 
Bayesian procedure constructs a natural prior distribution for 
the situation, and is shown to estimate the changepoint locations 
consistently, with an optimal convergence rate. Our methods 
substantially improve changepoint detection power when prior information 
is available. The methods are also tailored to bivariate data, allowing 
changes to occur in one or both component series.

\end{abstract}



\section{Introduction}\label{section:intro} 

Changepoints, also called structural breaks or breakpoints, are times in 
a sequential record where the data abruptly shift in some manner (mean, 
variance, autocovariance, quantile, etc.). The primary goal of a 
retrospective multiple changepoint analysis, the case considered here, 
is to estimate the number of changepoints and their location times. 
Various approaches have been developed for independent data; good recent 
references include \citet{Fryzlewicz_2014}, \citet{Pein_etal_2017}, and 
the review paper \citet{Niu_etal_2016} (and the references therein). 
When the data are correlated, such as the monthly temperature records 
studied here, this feature can greatly impede changepoint detection; in 
fact, mean shifts can often be misattributed to positive correlation 
\citep{Lund_etal_2007}.

One simple way to detect multiple changepoints is to combine an at most 
one changepoint (AMOC) technique (say a CUSUM or likelihood ratio test) 
with a binary segmentation procedure, e.g., \citet{Shao_Zhang_2010, 
Aue_Horvath_2013, Fryzlewicz_SubbaRao_2014}.  Wild binary segmentation 
techniques usually improve upon ordinary binary segmentation methods 
\citep{Fryzlewicz_2014}.  Since estimating the optimal multiple 
changepoint configuration can be formulated as a model selection 
problem, penalized likelihood methods such as BIC \citep{Yao_1988} and 
its modifications \citep{Zhang_Siegmund_2007, Zhang_Siegmund_2012}, and 
minimum description lengths (MDL) are also popular.  In this paper, an 
MDL technique is developed that takes into account prior information on 
the changepoint numbers and locations.  This scenario is shown to arise 
in the homogenization of temperature time series to account for gauge 
changes and station location moves.

The MDL principle \citep{Risanen_1989} from information theory has been 
successfully applied in statistical model selection problems 
\citep{Hansen_Yu_2001}.  MDL penalties are the sum of penalties (i.e., 
description lengths, or code lengths) of all unknown model parameters. 
In the multiple changepoint literature, the seminal work of 
\citet{Davis_etal_2006} develops an MDL penalty for piecewise 
autoregressive (AR) processes.  Here, the penalty is constructed by 
following certain automatic rules that assign different penalties to 
different parameter types: bounded integer parameters, unbounded integer 
parameters, and real-valued parameters. Since MDL penalties are not just 
simple multiples of the number of model parameters, they are believed 
superior to AIC and BIC penalties (a belief supported by simulations), 
and are shown consistent for changepoint estimation under infill 
asymptotics \citep{Davis_etal_2006, Davis_Yau_2013}. Following the 
automatic penalty rules, MDL methods have been extended to various time 
series structures, including GARCH processes \citep{Davis_etal_2008}, 
periodic ARs \citep{Lu_etal_2010}, autoregressive moving-averages 
\citep{Davis_Yau_2013}, and threshold ARs \citep{Yau_etal_2015}.

The main goal of this paper is to incorporate partial information on 
changepoint numbers and times into the MDL penalty, an aspect not 
readily handled by existing MDL methods.  Indeed, this will require us 
to revisit information theory.  The motivating example involves the 
climate homogenization \citep{Caussinus_Mestre_2004, 
Menne_WilliamsJr_2005} of monthly temperature records. Here, the aim is 
to detect abrupt mean shifts, which are often induced by artificial 
causes such as station relocations or gauge changes. Two types of {\it a 
priori} changepoint knowledge arise. First, metadata station history 
logs, which document the times of physical changes in the station, are 
sometimes available. Although metadata climate records are notoriously 
incomplete, and not all documented metadata times induce actual mean 
shifts in the series, climatologists believe that metadata times are 
more likely than non-metadata times to be changepoints.  Second, when 
multivariate records exist for the same station, changepoints may affect 
component records simultaneously. For example, with monthly maximum and 
minimum temperature averages (called Tmax and Tmin, respectively), 
moving a station to a drier location can both increase daytime highs and 
reduce nighttime lows. While changepoints in either Tmax or Tmin can 
occur by themselves, climatologists believe that it is more likely for 
changepoints to occur in both component series at the same time (these 
are called concurrent shifts).

While metadata is typically only used to verify climate changepoint 
conclusions in hindsight, Sections \ref{sec:simulation} and 
\ref{sec:Tuscaloosa} will show that use of metadata can improve 
detection power and time of estimation accuracy. This benefit is not 
limited to climatological pursuits; in other areas such as biology, 
economics, and engineering, domain expert knowledge is often available; 
e.g., knowledge from previous experiments on possible copy number 
variation locations, or the impact of certain political policy or regime 
changes on financial series.

Of course, Bayesian methods account for {\it a priori} knowledge via the 
construction of prior distributions. From a Bayesian model selection 
perspective, the optimal model (i.e., multiple changepoint 
configuration) is the one with the highest posterior probability 
\citep{Clyde_George_2004}. This maximum {\it a posteriori} (MAP) rule 
can be loosely viewed as a penalization method, where the posterior 
density is a penalized likelihood and the prior density is the penalty. 
Compared to frequentist approaches, one advantage of Bayesian posterior 
analysis is that it can also provide a measure of uncertainty for model 
parameters and changepoint locations. Bayesian approaches have been 
proposed for retrospective multiple changepoint detection --- see 
\citet{Barry_Hartigan_1993, Chib_1998, Fearnhead_2006, Giron_etal_2007, 
Zhang_Siegmund_2007, Giordani_Kohn_2008, Fearnhead_Vasileiou_2009, 
Hannart_Naveau_2012}. However, theoretical studies of large sample 
performance of Bayesian methods are in general lacking; while
\citet{Du_etal_2016} study asymptotic consistency of changepoint 
locations, they only consider independent data. 

More importantly, existing Bayesian changepoint approaches are typically 
derived under non-informative prior distributions; they rarely explicate 
how to incorporate real subjective prior knowledge. BIC-based 
changepoint detection methods cannot readily handle subjective prior 
information: from a Bayesian model selection perspective, BIC is a large 
sample approximation of the marginal likelihood. Thus, comparing models 
directly based on their BICs imposes an implicit assumption that the 
prior probabilities of the models are the same, which is not appropriate 
when one wants to incorporate metadata information.

The only exception to the above is \citet{Li_Lund_2015}, which accounts 
for metadata in a univariate precipitation time series.  That work was 
written for a climate audience and was largely void of statistical and 
technical detail.  This paper complements that work by dealing with the 
statistical and technical issues.  It has a different focus and content, 
aiming to develop a general MDL framework that can handle prior 
information on changepoint times in a wide range of changepoint 
problems.  For example, multivariate series, which involve the more 
challenging problem of borrowing information across component series, 
are pursued. In this sense, \citet{Li_Lund_2015} is a special case of 
the current paper. This paper also includes a thorough investigation of 
the asymptotic consistency of the proposed methods.

Changepoint detection for multivariate data has received significant 
attention in recent years, e.g., \citet{Cho_Fryzlewicz_2015, 
Kirch_etal_2015, Preuss_etal_2015, Ma_Yau_2016}. In 
\citet{Davis_etal_2006}, the automatic MDL is applied to multivariate AR 
series, where changepoints affect all component series. However, for 
many applications, a changepoint may not affect all component series. 
The automatic MDL does not directly accommodate this case, probably 
because it is unclear whether a change affecting all components should 
receive the same penalty as one that affects a subset of components. On 
the other hand, Bayesian approaches such as \citet{Zhang_Siegmund_2012} 
and \citet{Bardwell_Fearnhead_2017} can handle this problem, but only 
for independent data over time and components. Since these works are 
developed under non-informative prior distributions, they are not ready 
applicable to handle multivariate temperature homogenization, where 
concurrent changes in Tmax and Tmin should be encouraged.

In this paper, a new class of flexible MDL methods is proposed that 
incorporates domain experts' {\it a priori} knowledge for multiple 
changepoint detection, in both univariate and multivariate time series. 
Multiple changepoint configurations are reformulated as vectors of 
zero/one indicators, thus permitting natural construction of subjective 
prior distributions, with straightforward hyper-parameter elicitation. 
To account for correlation in time and across components, AR processes 
for univariate data, and vector autoregressive (VAR) processes for 
multivariate data are employed. Our MDL method is termed a Bayesian MDL 
(BMDL) because it can be viewed as an empirical Bayes model selection 
approach. While our main focus is to improve and generalize conventional 
MDL changepoint detection approaches, to the best of our knowledge, this 
paper is the first Bayesian multiple changepoint work to establish 
asymptotic consistency with correlated observations. Under infill 
asymptotics, the estimated changepoint locations are shown to converge 
in probability to their true values; moreover, estimators of the number 
of changepoints and model parameters such as regime means and AR 
coefficients are also consistent.

We choose to work within the MDL framework rather than extending 
BIC-based approaches due to the following considerations. First, the BIC 
approximation to the marginal likelihood is usually precise only up to 
an $O(1)$ error. Although it is asymptotically consistent for model 
selection, it often does not work well when the sample size is small or 
moderate \citep{Grunwald_2007}. Second and perhaps more importantly, in 
the changepoint detection literature, MDL penalties have been 
demonstrated to be more flexible and have better empirical performance 
than BIC penalties \citep{Davis_etal_2006}. Therefore, MDL methods to 
are exclusively pursued here.

The rest of this paper is organized as follows. Section 
\ref{section:MDL_review} briefly reviews MDL principles. Section 
\ref{section:BMDL} develops a BMDL penalty to detect mean shifts in 
univariate series. This work incorporates metadata, while allowing for a 
confounding seasonal mean cycle and AR errors.  Section 
\ref{section:bivariateBMDL} extends the BMDL to the multivariate 
setting, where Tmax and Tmin series are modeled jointly. Section 
\ref{sec:simulation} presents simulation examples. Section 
\ref{sec:Tuscaloosa} moves to an application to 114 years of monthly 
temperatures from Tuscaloosa, Alabama. Section \ref{section:asymptotics} 
studies the frequentist large sample performance of the univariate BMDL. 
Comments close the paper in Section \ref{sec:discussion}.  Technical 
results and proofs are delegated to an appendix.

\section{A Brief Review of MDL}\label{section:MDL_review} 

In information theory, a code length is the number of binary storage 
units required to transmit a random number or code. To reduce storage 
costs, one wants to assign shorter (longer) code lengths to common 
(rare) outcomes. Competing probability models can be compared by their 
code lengths; the true data generating distribution (i.e., the true 
model) should have the shortest expected code length. The MDL principle 
\citep{Risanen_1989} states that given the observed data, the model with 
the shortest code length is optimal.

For a discrete random variable $X$ with probability mass function 
$f(\cdot)$, \citet{Shannon_1948} states that the encoding with code 
length
\begin{equation}\label{eq:basic_MDL}
\mathcal{L}(X) = -\log_2 \{f(X)\}
\end{equation}
has the shortest expected code length. The existing MDL approach for 
multiple changepoint detection \citep{Davis_etal_2006} is developed 
under the automatic rules that the code length of a positive random 
integer $X$ bounded above by $N$ is $\log_2(N)$, and that of an 
unbounded positive random integer $X$ is $\log_2(X)$. The former rule 
implies a uniform distribution over the set $\{1, 2, \ldots, N \}$, 
which leads to the code length $\mathcal{L}(X) = -\log_2(1/N) = 
\log_2(N)$, while the latter implies an improper power law distribution with 
the probability mass function $f(X) \propto 1/X$.

For a continuous random variable, say $X \in \mathbb{R}^k$ with density 
function $f(\cdot)$, after discretizing each dimension into equal cells 
of size $\delta$ (often viewed as the machine precision), one can mimic 
the discrete case to obtain $\mathcal{L}(X) = -\log_2 \{f(X)\delta^k\} = 
-\log_2 f(X) - k\log_2(\delta)$. Because $k$ and $\delta$ do not vary 
with $X$, the term $-k\log_2(\delta)$ does not affect comparison between 
different outcomes of $X$ and is hence often omitted.  Thus, the MDL for 
a continuous variable can also be expressed as in \eqref{eq:basic_MDL}. 
In the rest of this paper, the natural logarithm is substituted for the 
base two logarithm --- this does not affect model comparisons since 
$\log_2(x)/\log(x)$ is constant in $x$.

Now suppose that a dataset $\mathbf{X} = (X_1, \ldots, X_N)'$, believed 
to be generated from a certain parametric model $\mathcal{M}$ with 
density $f(\mathbf{X} \mid \theta, \mathcal{M})$, is to be transmitted 
along with a possibly unknown parameter $\theta \in \Theta$. As reviewed 
in \citet{Hansen_Yu_2001}, two types of MDL approaches, the two-part MDL 
and the mixture MDL, are commonly used.

\subsection{Two-part MDLs}

The two-part MDL, also called the two-stage MDL, considers the 
transmission of $\mathbf{X}$ and $\theta$ in two steps.  If both the 
sender and receiver know $\theta$, the MDL of $\mathbf{X}$ is 
$\mathcal{L}(\mathbf{X} \mid \theta, \mathcal{M}) = -\log \{f(\mathbf{X} 
\mid \theta, \mathcal{M})\}$. Here, notations such as $\mathcal{L}(\cdot 
\mid \cdot)$ are analogous to the usual conditional distribution 
notations that emphasize dependence. Should $\theta$ also be unknown to 
the receiver, an additional cost of $\mathcal{L}(\theta \mid 
\mathcal{M})$ is incurred in transmitting it. Hence, the two-part MDL is
\[
\mathcal{L}(\mathbf{X}, \theta \mid \mathcal{M}) = 
\mathcal{L}(\mathbf{X} \mid \theta, \mathcal{M}) + \mathcal{L}(\theta 
\mid \mathcal{M}).
\]

Suppose that $\mathcal{L}(\mathbf{X}, \theta \mid \mathcal{M})$ is 
minimized at $\hat{\theta}$, an estimator of $\theta$ based on the data 
$\mathbf{X}$. If $\theta$ is a $k$-dimensional continuous parameter and 
$\hat{\theta}$ is a $\sqrt{N}$-consistent estimator, then one can set 
the machine precision to be $\delta = c/\sqrt{N}$, where $c$ is a 
positive constant. Under a uniform encoder $\pi(\theta \mid \mathcal{M}) 
\propto 1$, the code length needed to transmit $\theta$ (including 
$\hat{\theta}$) is hence $\mathcal{L}(\theta\mid \mathcal{M}) = -\log 
\{\pi(\theta\mid \mathcal{M})\} -k \log(c/\sqrt{N}) = k\log(N) / 2 - k 
\log(c)$, which does not depend on $\theta$. Hence, the maximum 
likelihood estimator (MLE) minimizes $\mathcal{L}(\mathbf{X}, \theta 
\mid \mathcal{M})$, and the two-part MDL coincides with the BIC 
\citep{Schwarz_1978}. In fact, $\hat{\theta}$ need not be the MLE; any 
$\sqrt{N}$-consistent estimator is justifiable. Again the constant term 
$k\log(c)$ can be dropped and the remaining code length 
$\mathcal{L}(\hat{\theta} \mid \mathcal{M}) = k \log(N)/2$ is adopted by 
\citet{Davis_etal_2006} as the automatic MDL rule for a $k$-dimensional 
continuous parameter.

If there exists a discrete set of candidate models, to account for model 
uncertainty, the two-part MDL can be modified to include an additional 
code length for the model $\mathcal{M}$, i.e.,
\begin{equation}
\label{eq:MDL_two-part3}
\mathcal{L}(\mathbf{X}, \hat{\theta}, \mathcal{M}) 
= \mathcal{L}(\mathbf{X} \mid \hat{\theta}, \mathcal{M}) + 
\mathcal{L}(\hat{\theta} \mid \mathcal{M}) + 
\mathcal{L}(\mathcal{M}),
\end{equation}
where $\hat{\theta}$ is model dependent, $\mathcal{L}(\mathcal{M}) = 
-\log\{\pi(\mathcal{M})\}$, and $\pi(\mathcal{M})$ is the prior 
distribution over the model space. The model with the smallest MDL in 
\eqref{eq:MDL_two-part3} is deemed optimal.

All existing automatic MDL methods for multiple changepoint detection 
are based on two-part MDLs. However, for a finite sample size $N$, the 
two-part MDL is problematic when the dimension of $\theta$ changes 
across models, as in the multiple changepoint case. Consider a setting 
of two competing models $\mathcal{M}_1$ and $\mathcal{M}_2$, whose 
parameters $\theta_j$ are $k_j$-dimensional continuous parameters, for 
$j = 1, 2$, and $k_1 \neq k_2$. Model $\mathcal{M}_1$ is favored if 
$\mathcal{L}(\mathbf{X}, \hat{\theta}_1, \mathcal{M}_1) - 
\mathcal{L}(\mathbf{X}, \hat{\theta}_2, \mathcal{M}_2)$ is negative; 
otherwise, model $\mathcal{M}_2$ is favored. Note that the code length 
difference for the parameters $\mathcal{L}(\hat{\theta}_1 \mid 
\mathcal{M}_1) - \mathcal{L}(\hat{\theta}_2 \mid \mathcal{M}_2)$ 
contains the term $(k_1 - k_2) \{\log(N) - 2\log(c)\}/2$. This term, and 
hence also $\mathcal{L}(\mathbf{X}, \hat{\theta}_1, \mathcal{M}_1) - 
\mathcal{L}(\mathbf{X}, \hat{\theta}_2, \mathcal{M}_2)$, could be either 
positive or negative depending on $N$ and the arbitrary constant $c$. 
One cannot judge either model superior without knowledge of $c$.  Of 
course, this issue does not conflict with the asymptotic consistency of 
BIC or automatic MDLs: as $N$ increases, $\log(N)$ dominates the 
constant $\log(c)$. Mixture MDLs, reviewed next, do not suffer from such 
a problem for a finite $N$.

\subsection{Mixture MDLs}

By \citet{Hansen_Yu_2001}, the mixture MDL is defined to be based on the 
marginal likelihood $f(\mathbf{X} \mid \mathcal{M})$:
\[
\mathcal{L}(\mathbf{X} \mid \mathcal{M}) = 
-\log \{f(\mathbf{X} \mid \mathcal{M})\}, \quad \text{where }
f(\mathbf{X} \mid \mathcal{M}) 
= \int_{\Theta} f(\mathbf{X} \mid \theta, \mathcal{M}) \pi(\theta \mid 
\mathcal{M})d\theta
\] 
averages the likelihood $f(\mathbf{X} \mid \theta, \mathcal{M})$ over 
$\theta$ under its prior density $\pi(\theta \mid \mathcal{M})$. If this 
prior distribution depends on an unknown hyper-parameter $\psi$, then a 
two-part MDL can be used to account for the additional cost needed to 
transmit $\psi$. In this case, the overall mixture MDL, for any 
$\sqrt{N}$-consistent estimator of $\psi$, is
\[
\mathcal{L}(\mathbf{X}, 
\hat{\psi} \mid \mathcal{M}) = - \log \left\{ \int_{\Theta} f(\mathbf{X} \mid 
\theta, \mathcal{M}) \pi(\theta \mid \hat{\psi}, \mathcal{M})d\theta \right\} 
+ \mathcal{L}(\hat{\psi}\mid \mathcal{M}).
\]

The mixture MDL for the model $\mathcal{M}$ is thus 
$\mathcal{L}(\mathbf{X}, \hat{\psi}, \mathcal{M}) = 
\mathcal{L}(\mathbf{X}, \hat{\psi} \mid \mathcal{M}) + \mathcal{L}( 
\mathcal{M})$, which is related to empirical Bayes (EB) approaches 
\citep{Carlin_Louis_2000}. If the prior probabilities of two models are 
the same, i.e., $\pi(\mathcal{M}_1) = \pi(\mathcal{M}_2)$, and the 
hyper-parameter $\psi$ is transmitted under the uniform encoder 
$\pi(\psi \mid \mathcal{M}_j) \propto 1$ for $j = 1,2$, then the 
difference of the two mixture MDLs, $\mathcal{L}(\mathbf{X}, 
\hat{\psi}_1, \mathcal{M}_1) - \mathcal{L} (\mathbf{X}, \hat{\psi}_2, 
\mathcal{M}_2)$, equals the logarithm of their Bayes factor 
$\text{BF}_{\mathcal{M}_2: \mathcal{M}_1}$ \citep{Kass_Raftery_1995}. 
Similarly, in EB settings, while the estimator $\hat{\psi}$ is often 
chosen to maximize the marginal likelihood $f(\mathbf{X} \mid \psi, 
\mathcal{M})$, other consistent estimators (moments for example) can be 
used.

\section{Bayesian Minimum Description Lengths for a Univariate Time Series}\label{section:BMDL} 

Consider a univariate time series $\mathbf{X}_{1:N}= (X_1, \ldots, 
X_N)'$ with a seasonal mean cycle with fundamental period $T$. For 
monthly data, $T = 12$. A model with autoregressive errors describing 
this situation is
\begin{equation} 
\label{eq:likelihood1} 
X_t  = s_{v(t)} + \mu_{r(t)} + \epsilon_t,  \quad
\epsilon_t  = \sum_{j = 1}^p \phi_j \epsilon_{t-j} + Z_t.
\end{equation} 
Here, $v(t) = t - T \lfloor (t-1)/T \rfloor \in \{ 1, 2, \ldots, T\}$ is 
the season corresponding to time $t$, where $\lfloor x \rfloor$ is the 
largest integer less than or equal to $x$. The seasonal means 
$\mathbf{s} = (s_1, \ldots, s_T)'$ are unknown. The errors $\{ \epsilon_ 
t \}_{t=1}^N$ are a causal zero mean AR process. Here, we assume that 
the AR order $p$ is known; if unsure, picking a slightly larger value 
for $p$ is advised. The AR coefficients $\boldsymbol{\phi} = (\phi_1, 
\ldots, \phi_p)'$ and the white noise variance $\text{Var}(Z_t) = 
\sigma^2$ are assumed unknown. For likelihood computations, following 
\citet{Davis_etal_2006}, white noises are assumed iid normal. This can 
be justified as a quasi-likelihood approach; furthermore, in climate 
applications, monthly averaged temperatures are approximately normally 
distributed \citep{Wilks_2011}.

Suppose a multiple changepoint configuration (i.e., a model) contains 
$m$ changepoints at the times $\tau_1 < \cdots < \tau_m \leq N$. These 
times partition the observations $\{ 1, \ldots, N \}$ into $m+1$ 
distinct regimes (segments), where the series' overall mean (neglecting 
its seasonal component), $\mu_{r(t)}$, changes across regimes. To avoid 
trite work with edge effects of the autoregression, we assume that no 
changepoints occur during the first $p$ observations. For notation, set 
$\tau_0=1$ and $\tau_{m+1}=N+1$. The regime indicator $r(t)$ in 
\eqref{eq:likelihood1} satisfies $r(t)=r$ when $\tau_{r-1} \leq t < 
\tau_r$. To ensure identifiability, $\mu_1$ is set to zero; hence, 
$E(X_t)=s_{v(t)}$ when $t$ lies in the first regime. The other regime 
means $\boldsymbol{\mu} = (\mu_2, \ldots, \mu_{m+1})'$ are unknown.

Following \citet{Li_Lund_2015}, the multiple changepoint configuration 
$(m; \boldsymbol{\tau})$ is reformulated as an $(N-p)$-dimensional 
vector of zero/one indicators: $\boldsymbol{\eta} = (\eta_{p+1}, \ldots, 
\eta_N)'$. Here, $\eta_t=1$ indicates that time $t$ is a changepoint in 
this model; $\eta_t=0$ means that time $t$ is not a changepoint. The 
total number of changepoints in model $\boldsymbol\eta$ is thus 
$m=\sum_{t=p+1}^N \eta_t$.

Our idea is to apply the mixture MDL to the continuous parameter 
$\boldsymbol{\mu}$, whose dimension varies across models, and use the 
two-part MDL for the parameters $\mathbf{s}, \boldsymbol\sigma^2, 
\boldsymbol\phi$, and the model $\boldsymbol\eta$. In the rest of this 
section, subsection \ref{subsection:prior} introduces our priors on 
$\boldsymbol{\eta}$ and $\boldsymbol{\mu}$, subsection 
\ref{subsection:bmdl_univariate} derives the BMDL formula 
\eqref{eq:BMDL_univariate}, and subsection 
\ref{subsection:computation} concludes with computational strategies.
 Asymptotic studies are included in section \ref{section:asymptotics}.

\subsection{Prior specifications}
\label{subsection:prior}

Our prior distribution for the changepoint model $\boldsymbol{\eta}$ 
assumes that, in the absence of metadata, each time $t$ has an equal 
probability $\rho$ of being a changepoint, independently of all other 
times, i.e., \begin{equation}\label{eq:eta_prior_bernoulli} \eta_t 
\stackrel{\text{iid}}{\sim} \text{Bernoulli} (\rho), \quad t = p+1, 
\ldots, N. \end{equation} This independent Bernoulli prior has been used 
in previous Bayesian multiple changepoint detection works 
\citep{Chernoff_Zacks_1964, Yao_1984, Barry_Hartigan_1993}.  From a 
hidden Markov perspective, this prior is equivalent to $\tau_r \mid 
\tau_{r-1} \sim \text{Geometric}(\rho)$ for $r = 1, \ldots, m$ 
\citep{Fearnhead_Vasileiou_2009}, and thus is a special case of the 
negative Binomial prior \citep{Hannart_Naveau_2012}. The uniform prior 
$\pi(\boldsymbol\eta) \propto 1$ adopted in \citet{Du_etal_2016} is a 
special case of the Bernoulli prior with $\rho = 0.5$.  For applications 
where knowledge beyond metadata is unavailable, an iid prior on $\{ 
\eta_t \}$ seems reasonable. In other applications, 
$\pi(\boldsymbol{\eta})$ is allowed to have different success 
probabilities in different regimes \citep{Chib_1998}; correlation across 
different changepoint times can also be achieved using Ising priors 
\citep{Li_Zhang_2010}.

To account for uncertainty in the success probability $\rho$, a 
hyper-prior is placed on it. \citet{Barry_Hartigan_1993} let $\rho$ have 
a uniform prior on the interval $(0, \rho_0)$, where $\rho_0 < 1$. For 
additional flexibility, we use the Beta distribution
\begin{equation}\label{eq:eta_prior_beta}
\rho \sim \text{Beta}(a, b), 
\end{equation} 
where $a,b > 0$ are fixed hyper-parameters. The Beta-Binomial 
hierarchical priors in \eqref{eq:eta_prior_bernoulli} and 
\eqref{eq:eta_prior_beta} are widely used in Bayesian model selection 
\citep{Scott_Berger_2010}, and have been adopted to detect changepoints 
\citep{Giordani_Kohn_2008, Li_Lund_2015}. Due to conjugacy, the marginal 
prior density of $\boldsymbol{\eta}$ has the following closed form, with 
$\beta(\cdot, \cdot)$ denoting the Beta function:
\begin{equation}
\label{eq:prior_eta} 
\pi(\boldsymbol\eta) = \int_0^1 \pi(\rho) \prod_{t = p + 1}^N \pi(\eta_t 
\mid \rho)  d\rho = \frac{\beta(a + m, b + N-p - m)}
{\beta(a, b)}.
\end{equation}

Note that here, the Beta-Binomial density in \eqref{eq:prior_eta} depends 
on $\boldsymbol\eta$ through $m$, the total number of changepoints in 
the multiple changepoint model $\boldsymbol\eta$. In common changepoint 
detection problems, changepoints are usually relatively sparse ($m \ll 
N$). Suppose our prior belief on $\rho$ reflects this sparsity 
assumption, say, $E(\rho) = a/(a+b) \leq 1/2$, i.e., $a\leq b$.  Then 
\eqref{eq:prior_eta} decreases as $m$ increases until $m$ reaches a 
relatively large value (at least $(N - p)/2$). Thus, the Beta-Binomial 
prior can be viewed as a prior preference on smaller models, or 
equivalently, a penalty on the number of changepoints.

For hyper-parameter choices, an objective Bayesian option 
\citep{Giron_etal_2007} is $a=b=1$. In this case, $\pi(\boldsymbol\eta) 
= \left\{{N-p \choose m} (N-p+1)\right\}^{-1}$, which implies that 
marginally, the number of changepoints $m$ has a uniform prior on the 
set $\{0, 1, \ldots, N-p\}$, and all models containing the same number 
of changepoints have the same prior probabilities. The Beta-Binomial 
prior can be tuned to accommodate subjective knowledge from domain 
experts. For temperature homogenization, \citet{Mitchell_1953} estimates 
an average of six station relocations and gauge changes per century in 
United States temperature series; this long-term rate is 0.005 
changepoints per month and can be produced with $a=1$ and $b=199$; with 
these parameters, $E(\rho)=a/(a+b)=0.005$.

This prior is now modified to accommodate metadata. Suppose that during 
the times $\{p+1, \ldots, N\}$, there are $N^{(2)}$ documented times 
(times listed in the metadata) and $N^{(1)}=N-p-N^{(2)}$ undocumented 
times. For notation, all quantities superscripted with $(1)$ refer to 
undocumented times; quantities superscripted with $(2)$ refer to 
documented times. Following \citet{Li_Lund_2015}, we posit that the 
undocumented times have a Beta-Binomial$(a, b^{(1)})$ prior, and 
independently, the documented times have a Beta-Binomial$(a, b^{(2)})$ 
prior. To make the metadata times more likely to induce true mean 
shifts, we impose $b^{(1)} > b^{(2)}$ so that
\[
E\left(\rho^{(1)}\right) = \frac{a}{a + b^{(1)}} < \frac{a}{a + b^{(2)}} = E\left(\rho^{(2)}\right).
\]  

For monthly data, default values are $a=1, b^{(1)}=239$, and 
$b^{(2)}=47$, making $E(\rho^{(1)}) = 0.0042$, i.e., an average of one 
changepoint about every 20 years for non-metadata times, and 
$E(\rho^{(2)})=0.0208$, i.e., on average, one changepoint in every 4 
years for metadata times. In other words, {\it a priori}, a documented 
time is roughly five times more likely to be a changepoint than an 
undocumented time. For different problems, one may need to modify 
$b^{(1)}$ and $b^{(2)}$ to reflect specific domain knowledge. Our 
previous paper \citep{Li_Lund_2015} gives a detailed sensitivity 
analysis on the choice of Beta-Binomial hyper-parameters. It suggests 
that changepoint detection results are relatively stable under a range 
of $E(\rho^{(2)})/E(\rho^{(1)})$ values. For applications that lack any 
subjective information, the non-informative Beta-Binomial$(1,1)$ prior 
can serve as a default choice.  In this paper, this prior is referred to 
as ``oBMDL'', with ``o'' standing for objective. Empirical comparison 
will be provided in the univariate simulation examples in Section 
\ref{sec:simulation_univariate}.

Following \eqref{eq:prior_eta} and writing Beta integrals via their 
Gamma function representations, a changepoint configuration 
$\boldsymbol{\eta}$ with $m^{(2)}$ documented changepoints and $m^{(1)}$ 
undocumented changepoints ($m=m^{(1)} + m^{(2)}$) has a marginal prior 
density (up to a normalizing constant)
\[
\pi(\boldsymbol{\eta}) \propto \prod_{k=1}^2
\Gamma\left(a + m^{(k)}\right) \Gamma\left(b^{(k)} + N^{(k)} - m^{(k)}\right).
\]

For a changepoint model with $m > 0$ changepoints, priors for the 
$m$-dimensional regime means $\boldsymbol{\mu}$ are posited to have 
independent normal prior distributions:
\begin{equation}\label{eq:mu_prior_normal}
\boldsymbol{\mu} \mid \sigma^2, \boldsymbol{\eta} 
\sim \text{N}(\mathbf{0}, \nu \sigma^2 \mathbf{I}_{m}).
\end{equation}
Here, $\nu$ is a pre-specified non-negative parameter that is relatively 
large (making the variances of the regime means large multiples of the 
white noise variances). Similar to the sensitivity analysis in 
\citet{Du_etal_2016}, our experience suggests that model selection 
results are stable under a wide range of $\nu$ values. Our default takes 
$\nu = 5$.

In fact, $\pi(\boldsymbol{\mu})$ can be any zero mean continuous 
distribution. For example, if mean shifts are expected to be large, 
heavy-tailed distributions such as the Student-$t$ may be preferable. 
When $\boldsymbol{\mu}$ cannot be tractably integrated out, inferences 
can be based on Laplace approximations or posterior sampling with a 
reversible-jump MCMCs \citep{Green_1995}. Due to conjugacy under 
Gaussian likelihoods, the normal prior leads to closed form marginal 
likelihoods. Hence, for computational ease in the rest of this paper, 
the normal regime mean priors in \eqref{eq:mu_prior_normal} are used.

\subsection{The BMDL expression}
\label{subsection:bmdl_univariate}

To derive the BMDL expression in \eqref{eq:BMDL_univariate}, the data 
likelihood is first obtained.  This is then integrated over 
$\boldsymbol\mu$ to obtain the mixture MDL; finally, two-part MDLs are 
obtained for the rest of the parameters.

Given a changepoint model $\boldsymbol{\eta}$, the sampling distribution 
\eqref{eq:likelihood1} has the regression representation
\begin{equation}
\label{eq:likelihood3}
\mathbf{X}_{1:N} = \mathbf{A}_{1:N} \mathbf{s} + 
\mathbf{D}_{1:N}\boldsymbol\mu + \boldsymbol\epsilon_{1:N},
\end{equation}
with $\mathbf{A}_{1:N}\in \mathbb{R}^{N \times T}$ and 
$\mathbf{D}_{1:N} \in \mathbb{R}^{N \times m}$ as seasonal and regime 
indicator matrices, respectively:
\begin{align*}
&\left[ \mathbf{A}_{1:N} \right]_{t,v} = 
	\mathbf{1}( \text{time } t \text{ is in season } v ), ~~~ v = 1,  \ldots, T,\\
&\left[ \mathbf{D}_{1:N} \right]_{t,r-1} = 
	\mathbf{1}( \text{time } t \text{ is in regime } r ), ~~~ r = 2, \ldots, m + 1,
\end{align*}
where $\mathbf{1}(A)$ denotes the indicator of the event $A$. In 
\eqref{eq:likelihood3}, the subscript $1:N$, or in general ${t_1:t_2}$, 
signifies that only rows $t_1$ through $t_2$ are used in the quantities. 
The normal white noises $\{ Z_t \}$ in the AR process imply the 
distributional result $\boldsymbol\epsilon_{(p+1):N} - \sum_{j = 1}^p 
\phi_j \boldsymbol\epsilon_{{(p+1-j):(N-j)}} \sim \text{N}(\mathbf{0}, 
\sigma^2 \mathbf{I}_{N-p})$, where $\mathbf{I}_k$ denotes the $k \times 
k$ identity matrix. Now define
\begin{align}
\label{eq:X_tilde}
\widetilde{\mathbf{X}} 
&= \mathbf{X}_{(p+1):N} 
	- \sum_{j = 1}^p \phi_j \mathbf{X}_{(p+1-j):(N-j)}, \\ \label{eq:A_D_tilde}
\widetilde{\mathbf{A}} 
&= \mathbf{A}_{(p+1):N} 
	- \sum_{j = 1}^p \phi_j \mathbf{A}_{(p+1-j):(N-j)}, \quad
\widetilde{\mathbf{D}} 
= \mathbf{D}_{(p+1):N} 
	- \sum_{j = 1}^p \phi_j \mathbf{D}_{(p+1-j):(N-j)},
\end{align}
and observe that 
\begin{equation}
\label{eq:likelihood5}
\widetilde{\mathbf{X}} - \widetilde{\mathbf{A}} \mathbf{s} 
- \widetilde{\mathbf{D}} \boldsymbol\mu \sim
\text{N}(\mathbf{0}, \sigma^2 \mathbf{I}_{N-p}).
\end{equation}
Note that all terms superscripted with $\sim$ depend on the unknown AR 
parameter $\boldsymbol{\phi}$. To avoid AR edge effects, a likelihood 
conditional on the initial observations $\mathbf{X}_{1:p}$ is used. In 
the change of variable computations, the Jacobian $|\partial 
(\widetilde{\mathbf{X}} - \widetilde{\mathbf{A}} \mathbf{s} - 
\widetilde{\mathbf{D}} \boldsymbol\mu) / \partial \mathbf{X}_{(p + 
1):N}| = 1$ and the likelihood has the multivariate normal form
\begin{equation*}
f\left(\mathbf{X}_{(p+1):N} \mid 
\boldsymbol\mu, \mathbf{s}, \sigma^2, \boldsymbol{\phi}, \boldsymbol{\eta}\right)
= \left(2\pi\sigma^2\right)^{-\frac{N-p}{2}} e^{-\frac{1}{2\sigma^2}
(\widetilde{\mathbf{X}} - \widetilde{\mathbf{A}} \mathbf{s} - 
\widetilde{\mathbf{D}}\boldsymbol\mu)'
(\widetilde{\mathbf{X}} - \widetilde{\mathbf{A}} \mathbf{s} - 
\widetilde{\mathbf{D}}\boldsymbol\mu)}.
\end{equation*}
Innovation forms of the likelihood \citep{Brockwell_Davis_1991} can be 
used if one wants a moving-average or long-memory component in $\{ 
\epsilon_t \}$.

We now obtain a BMDL for the changepoint model $\boldsymbol{\eta}$. If 
$m > 0$, we first use the mixture MDL on $\boldsymbol\mu$. The marginal 
likelihood, after integrating $\boldsymbol\mu$ out, has the closed form
\begin{align*} 
&	f (\mathbf{X}_{(p+1):N} \mid 
	\mathbf{s}, \sigma^2, \boldsymbol\phi, \boldsymbol\eta)
= 	\int_{\mathbb{R}^{m}} f\left(\mathbf{X}_{(p+1):N} \mid 
	\boldsymbol\mu, \mathbf{s}, \sigma^2, \boldsymbol\phi, \boldsymbol\eta
	\right)
	\pi(\boldsymbol\mu \mid \sigma^2, \boldsymbol\eta) d\boldsymbol\mu\\ 
=~&	(2\pi\sigma^2)^{-\frac{N-p}{2}} \nu^{-\frac{m}{2}} 
\left|\widetilde{\mathbf{D}}' 
\widetilde{\mathbf{D}}+\frac{\mathbf{I}_{m}}{\nu}\right|^{-\frac{1}{2}}
e^{-\frac{1}{2\sigma^2} (\widetilde{\mathbf{X}} - 
\widetilde{\mathbf{A}}\mathbf{s})' 
\widetilde{\mathbf{B}}
(\widetilde{\mathbf{X}} - \widetilde{\mathbf{A}}\mathbf{s})
},
\end{align*}
where the notation has
\begin{equation}
\label{eq:B_tilde}
\widetilde{\mathbf{B}} = \mathbf{I}_{N-p} - \widetilde{\mathbf{D}}
\left(\widetilde{\mathbf{D}}' 
\widetilde{\mathbf{D}}+\frac{\mathbf{I}_{m}}{\nu}\right)^{-1}
\widetilde{\mathbf{D}}'.
\end{equation}
If the parameters $\mathbf{s}$, $\sigma^2$, and $\boldsymbol{\phi}$ are 
known, the mixture MDL is simply $\mathcal{L}(\mathbf{X}_{(p+1):N} \mid 
\mathbf{s}, \sigma^2, \boldsymbol\phi, \boldsymbol\eta) = -\log \{f 
(\mathbf{X}_{(p+1):N} \mid \mathbf{s}, \sigma^2, \boldsymbol{\phi}, 
\boldsymbol{\eta})\}$.

Under a given changepoint model $\boldsymbol\eta$, the two-part MDL is used 
to quantify the cost of transmitting the parameters $\mathbf{s}$, $\sigma^2$, 
and $\boldsymbol{\phi}$. The 
optimal $\mathbf{s}$ and $\sigma^2$ that minimize the mixture MDL have 
closed forms:
\begin{align}
\label{eq:s_hat}
\hat{\mathbf{s}} 
& 	= \arg\min_{\mathbf{s}} 
	\mathcal{L}(\mathbf{X}_{(p+1):N} \mid \mathbf{s}, \sigma^2, 
	\boldsymbol\phi, \boldsymbol\eta)
	= (\widetilde{\mathbf{A}}'\widetilde{\mathbf{B}}
	\widetilde{\mathbf{A}})^{-1}
	(\widetilde{\mathbf{A}}'\widetilde{\mathbf{B}}\widetilde{\mathbf{X}}),
\\ \label{eq:sigmasq_hat}
\hat{\sigma}^2 
& 	= \arg\min_{\sigma^2} 
	\mathcal{L}(\mathbf{X}_{(p+1):N} \mid \hat{\mathbf{s}}, \sigma^2, 
	\boldsymbol\phi, \boldsymbol\eta)
	= \frac{1}{N-p} \widetilde{\mathbf{X}}' \left\{
 	\widetilde{\mathbf{B}}
	- \widetilde{\mathbf{B}} \widetilde{\mathbf{A}}
	\left(\widetilde{\mathbf{A}}' 
\widetilde{\mathbf{B}}\widetilde{\mathbf{A}}\right)^{-1}
	\widetilde{\mathbf{A}}' \widetilde{\mathbf{B}}
	\right\}\widetilde{\mathbf{X}}.
\end{align}
These estimators depend on $\boldsymbol{\phi}$; however, the 
$\boldsymbol{\phi}$ that minimizes $\mathcal{L}(\mathbf{X}_{(p+1):N} 
\mid \hat{\mathbf{s}}, \hat{\sigma}^2, \boldsymbol{\phi}, 
\boldsymbol{\eta})$ is intractable.  In general, likelihood estimators 
for autoregressive models do not have closed forms. Hence, simple 
Yule-Walker moment estimators, which are asymptotically most efficient 
and $\sqrt{N}$-consistent under the true changepoint model, are used. 
There is actually little difference between moment and likelihood 
estimators for autoregressions \citep{Brockwell_Davis_1991}.

In the linear model \eqref{eq:likelihood3}, the ordinary least squares 
residuals are
\begin{equation}
\label{eq:Y}
\boldsymbol\epsilon_{1:N}^{\text{ols}} = (\mathbf{I}_N - 
\mathcal{P}_{[\mathbf{A}_{1:N}|\mathbf{D}_{1:N}]})\mathbf{X}_{1:N},
\end{equation}
where $[\mathbf{A}_{1:N}|\mathbf{D}_{1:N}]$ denotes the block matrix formed 
by $\mathbf{A}_{1:N}$ and $\mathbf{D}_{1:N}$, and 
$\mathcal{P}_{[\mathbf{A}_{1:N}|\mathbf{D}_{1:N}]}$ is the 
orthogonal projection matrix onto its column space.  The sample 
autocovariance of the residuals are 
$\hat{\gamma}(h) = N^{-1} \sum_{t = h + 1}^N \epsilon_t^{\text{ols}} 
\epsilon_{t-h}^{\text{ols}}$, at lag $h = 0, 1, \ldots, p$. The Yule-Walker estimator of 
$\boldsymbol{\phi}$ is $\hat{\boldsymbol{\phi}} = 
\hat{\boldsymbol{\Gamma}}_p^{-1}\hat{\boldsymbol{\gamma}}_p$, where 
$\hat{\boldsymbol{\gamma}}_p = (\hat{\gamma}(1), \ldots, 
\hat{\gamma}(p))'$ and $\hat{\boldsymbol{\Gamma}}_p$ is a $p \times p$ 
matrix whose $(i,j)$th entry is $\hat{\gamma}(|i-j|)$.  This matrix is 
invertible whenever the data are non-constant 
\citep{Brockwell_Davis_1991}. Next, the Yule-Walker estimator 
$\hat{\boldsymbol{\phi}}$ is substituted for $\boldsymbol{\phi}$ in 
$\widetilde{\mathbf{X}}$, $\widetilde{\mathbf{A}}$, 
$\widetilde{\mathbf{D}}$, $\widetilde{\mathbf{B}}$, and 
$\hat{\sigma}^2$.  The resulting quantities are denoted by 
$\widehat{\mathbf{X}}$, $\widehat{\mathbf{A}}$, $\widehat{\mathbf{D}}$, 
$\widehat{\mathbf{B}}$, and $\hat{\sigma}^2$, respectively.  In 
particular, $\widehat{\mathbf{X}}$ contains estimated one-step-ahead 
prediction residuals (innovations).

By \eqref{eq:MDL_two-part3}, the BMDL for transmitting the data $\mathbf{X}_{(p+1):N}$, 
the model $\boldsymbol{\eta}$, and its parameters is (up to a constant)
\begin{align} \nonumber
\text{BMDL}(\boldsymbol\eta)  
&	= \mathcal{L}(\mathbf{X}_{(p+1):N} \mid \hat{\mathbf{s}}, \hat{\sigma}^2, 
\hat{\boldsymbol\phi}, \boldsymbol\eta) 
+ \mathcal{L}(\hat{\mathbf{s}}, \hat{\sigma}^2,\hat{\boldsymbol\phi} \mid \boldsymbol\eta)
+ \mathcal{L}(\boldsymbol\eta)\\ \label{eq:BMDL_sum}
&	= -\log \left\{f(\mathbf{X}_{(p+1):N} \mid \hat{\mathbf{s}}, \hat{\sigma}^2, 
\hat{\boldsymbol\phi}, \boldsymbol\eta)\right\} -\log \left\{\pi(\boldsymbol\eta)\right\}.
\end{align}
The second equality holds because under a uniform encoder 
$\pi(\mathbf{s}, \sigma^2, \boldsymbol{\phi}) \propto 1$, the two-part 
MDL $\mathcal{L}(\hat{\mathbf{s}}, \hat{\sigma}^2,\hat{\boldsymbol\phi} 
\mid \boldsymbol\eta) = (T + 1 + p)\log(N-p)/2$ is constant across 
models and hence can be omitted. Therefore, for a model with $m > 0$ 
changepoints, its BMDL is (up to a constant)
\begin{align}
\label{eq:BMDL_univariate}
\text{BMDL}(\boldsymbol\eta) = ~&
	 \frac{N-p}{2}\log \left( \hat{\sigma}^2 \right) + \frac{m}{2}\log(\nu) 
	+ \frac{1}{2}\log\left( \left|\widehat{\mathbf{D}}' 
	\widehat{\mathbf{D}}+\frac{\mathbf{I}_{m}}{\nu}\right| \right)\\ \nonumber
&	-\sum_{k=1}^2\log\left\{\Gamma\left(a + m^{(k)}\right) 
\Gamma\left(b^{(k)} + N^{(k)} - m^{(k)}\right) \right\}.
\end{align}

For a model with no changepoints ($m=0$), denoted by 
$\boldsymbol{\eta}_{\o}$, the above procedure needs modification. Since 
$\boldsymbol{\eta}_{\o}$ does not involve $\boldsymbol{\mu}$, the 
mixture MDL step can be skipped. As $\mathbf{D}$ has no columns, 
$\widetilde{\mathbf{B}}$ in \eqref{eq:B_tilde} is reduced to 
$\mathbf{I}_{N-p}$, and hence \eqref{eq:sigmasq_hat} still holds.  With 
the convention that the determinant of a $0 \times 0$ matrix is unity, 
$\log\left( \left|\widehat{\mathbf{D}}' \widehat{\mathbf{D}}+ 
\mathbf{I}_{m}/ \nu\right| \right) = 0$. Therefore, 
\eqref{eq:BMDL_univariate} also holds for $\boldsymbol{\eta}_{\o}$. This 
resolves the issue of evaluating $\log(m)$ at $m = 0$ with some existing 
MDL methods.

\subsection{BMDL optimization}
\label{subsection:computation} 

The optimal changepoint model $\hat{\boldsymbol{\eta}}$ is selected as 
the one with the smallest BMDL score. However, exhaustively searching 
the changepoint configuration space is formidable since the total number 
of admissible models, $2^{N-p}$, is extremely large. To overcome this, 
genetic algorithms are used as optimization tools in 
\citet{Davis_etal_2006} and \citet{Lu_etal_2010}.  Genetic algorithms 
efficiently explore the model space, only evaluating the penalized 
likelihood at a relatively small number of promising models.

The following connection to empirical Bayes (EB) methods allow us to 
borrow MCMC model search algorithms that are commonly used in Bayesian 
model selection. The BMDL under model $\boldsymbol\eta$ represented in 
\eqref{eq:BMDL_sum} is equivalent to the negative logarithm of an EB 
estimator of the posterior probability of $\boldsymbol{\eta}$:
\[
p_{\text{EB}}(\boldsymbol\eta \mid \mathbf{X}_{(p+1):N})
\propto~	 \pi(\boldsymbol\eta)
	 \int_{\mathbb{R}^{m}} f\left(\mathbf{X}_{(p+1):N} \mid 
	\boldsymbol\mu, \hat{\mathbf{s}}, \hat{\sigma}^2, \hat{\boldsymbol\phi}, 
	\boldsymbol\eta \right)
	\pi(\boldsymbol\mu \mid  \hat{\sigma}^2, \boldsymbol\eta) 
        d\boldsymbol\mu.
\]

As our BMDL formula \eqref{eq:BMDL_univariate} is tractable, Bayesian 
stochastic model search algorithms can be used; see 
\citet{Garcia-Donato_Martinez-Beneito_2013} and the references therein.
Here, we modify the Metropolis-Hastings algorithm in 
\citet{George_McCulloch_1997} by intertwining two types of proposals: a 
component-wise flipping at a random location and a simple random 
swapping between a changepoint and a non-changepoint. This algorithm is 
described in detail in \citet{Li_Lund_2015} and can be implemented by the 
R package {\tt BayesMDL} (\url{https://github.com/yingboli/BayesMDL}).

\section{Extensions to Multivariate Time Series}\label{section:bivariateBMDL} 

Mimicking the univariate setup, this section develops a BMDL for 
multivariate time series. While the details are illustrated for 
bivariate series, similar extensions apply to multivariate series of 
more than two components. The BMDL penalty constructed here allows 
changepoints to occur in one or both component series. Furthermore, it 
can accommodate domain experts' knowledge that encourage concurrent 
changes, i.e., changes affecting both series at the same time.

In temperature homogenization, to model Tmax and Tmin series jointly, 
both series are concatenated via $\mathbf{X}_{1:N} = 
(\mathbf{X}_{1:N,1}', \mathbf{X}_{1:N,2}')'$ $\in \mathbb{R}^{2N}$, 
where $\mathbf{X}_{1:N,i} = (X_{1,i}, \ldots, X_{N,i})'$ is the record 
for Tmax ($i=1$) or Tmin ($i=2$). Again, each time in $\{ p+1, \ldots, N 
\}$ is allowed to be a changepoint in either the Tmax or Tmin series, or 
both. A multiple changepoint configuration is denoted by 
$\boldsymbol{\eta} = (\boldsymbol{\eta}_1', \boldsymbol{\eta}_2')'$, 
where $\boldsymbol{\eta}_i = (\eta_{p+1,i}, \ldots, \eta_{N,i})' \in 
\{0, 1\}^{N-p}$ is defined as in the univariate case. Given a bivariate 
changepoint model $\boldsymbol{\eta}$, series $i$ has $m_i = 
\sum_{t=p+1}^N \eta_{t,i}$ changepoints.  As in the univariate case, the 
seasonal means are denoted by $\mathbf{s}_i = (s_{1,i}, \ldots, 
s_{T,i})' \in \mathbb{R}^T$; regime means are denoted by 
$\boldsymbol{\mu}_i = (\mu_{2,i}, \ldots, \mu_{m_i+1,i})' \in 
\mathbb{R}^{m_i}$.  The seasonal and regime indicator matrices 
$\mathbf{A}_{1:N,i} \in \mathbb{R}^{N \times T}$ and $\mathbf{D}_{1:N, 
i} \in \mathbb{R}^{N \times m_i}$ are constructed analogously to their 
univariate counterparts.

The regression representation \eqref{eq:likelihood3} holds for the 
bivariate case, with $\mathbf{s} = (\mathbf{s}_1', \mathbf{s}_2')'$, 
$\boldsymbol{\mu} = (\boldsymbol{\mu}_1', \boldsymbol{\mu}_2')'$, 
$\boldsymbol{\epsilon}_{1:N} = (\boldsymbol\epsilon_{1:N,1}', 
\boldsymbol\epsilon_{1:N,2}')'$ denoting the concatenated seasonal 
means, regime means, and regression errors, respectively.  The seasonal 
indicator matrix has the block diagonal form 
$\mathbf{A}_{1:N} = \text{diag}\left(\mathbf{A}_{1:N, 1}, \mathbf{A}_{1:N, 2}\right)$, 
and similarly the regime indicator matrix $\mathbf{D}_{1:N} = 
\text{diag}\left(\mathbf{D}_{1:N, 1}, \mathbf{D}_{1:N, 2} \right)$. Note 
that the seasonal indicators for Tmax and Tmin coincide, i.e., 
$\mathbf{A}_{1:N, 1} = \mathbf{A}_{1:N, 2}$, while $\mathbf{D}_{1:N,1}$ 
and $\mathbf{D}_{1:N,2}$ differ unless all changepoints are concurrent.

As Tmax and Tmin temperature series tend to fluctuate about the seasonal 
mean in tandem (positive correlation), the errors $\{ 
\boldsymbol{\epsilon}_t = (\epsilon_{t, 1}, \epsilon_{t, 2})' \}$ need 
to be correlated across components.  For this, a vector autoregressive 
model (VAR) of order $p$ is employed:
\[
\boldsymbol\epsilon_t = \sum_{j=1}^p 
\boldsymbol\Phi_j \boldsymbol\epsilon_{t-j} 
+ \mathbf{Z}_t, 
\quad 
\mbox{Cov}(\mathbf{Z}_t) =  \boldsymbol\Sigma,
\]
where $\boldsymbol{\Phi}_1, \ldots, \boldsymbol{\Phi}_p$ are $2 \times 
2$ VAR coefficient matrices.  The VAR model allows for correlation in 
time and between components.

As \eqref{eq:likelihood5} holds after replacing 
$\sigma^2\mathbf{I}_{N-p}$ with $\boldsymbol{\Sigma} \otimes 
\mathbf{I}_{N-p}$, the likelihood of $\mathbf{X}_{(p+1):N}$, conditional 
on the initial observations $\mathbf{X}_{1:p}$, is (up to a 
multiplicative constant)
\[
	f(\mathbf{X}_{(p+1):N} \mid \mathbf{s}, \boldsymbol\mu, \boldsymbol\Sigma, 
	\boldsymbol{\Phi}_{1:p}, \boldsymbol\eta)
\propto  
\left| \boldsymbol\Sigma \right|^{-\frac{N-p}{2}}
	e^{
	-\frac{1}{2} (\widetilde{\mathbf{X}} 
	- \widetilde{\mathbf{A}} \mathbf{s} - \widetilde{\mathbf{D}} \boldsymbol\mu)'
	(\boldsymbol\Sigma^{-1} \otimes \mathbf{I}_{N-p}) (\widetilde{\mathbf{X}} 
	- \widetilde{\mathbf{A}}\mathbf{s}- \widetilde{\mathbf{D}} \boldsymbol\mu) 
	}.
\]
Here, $\otimes$ denotes a Kronecker product and the terms 
$\widetilde{\mathbf{X}}, \widetilde{\mathbf{A}}, \widetilde{\mathbf{D}}$ 
are modified by replacing $\phi_j$ with $\boldsymbol{\Phi}_j \otimes 
\mathbf{I}_{N-p}$ in \eqref{eq:X_tilde} and \eqref{eq:A_D_tilde}, for $j 
= 1, \ldots, p$.


\subsection{Prior specifications} For $t = p + 1, \ldots, N$, the 
indicator $\boldsymbol{\eta}_{t} = (\eta_{t,1}, \eta_{t,2})'$ takes 
values in one of the four categories: $(1,1)'$, mean shifts in both Tmax 
and Tmin; $(1,0)'$, a mean shift in Tmax but not in Tmin; $(0,1)'$, a mean 
shift in Tmin but not in Tmax; and $(0,0)'$, no mean shifts. As a natural 
extension of the Beta-Binomial prior, a Dirichlet-Multinomial prior is 
put on $\boldsymbol{\eta}_{t}$:
\[
\boldsymbol\eta_{t} \mid \boldsymbol\rho
\stackrel{\text{iid}}{\sim} \text{Multinomial}(1; \boldsymbol{\rho}), 
\quad 
\boldsymbol{\rho} \sim \text{Dirichlet}(\boldsymbol{\alpha}),
\]
where $\boldsymbol{\rho} = (\rho_1, \ldots, \rho_4)'$ are the 
probabilities of the four categories satisfying $\sum_{\ell=1}^4 
\rho_\ell = 1$, and $\boldsymbol{\alpha} = (\alpha_1, \ldots, 
\alpha_4)'$ are the Dirichlet parameters with $\alpha_\ell > 0$ for each 
$\ell = 1, \ldots, 4$. Suppose that the changepoint configuration 
$\boldsymbol{\eta}$ has $m_\ell$ times in category $\ell$.  Due to 
Dirichlet-multinomial conjugacy, the marginal prior of 
$\boldsymbol{\eta}$ has a closed form after integrating out 
$\boldsymbol{\rho}^{(1)}$ and $\boldsymbol{\rho}^{(2)}$:
\[
\pi(\boldsymbol\eta) \propto
 	\prod_{k = 1}^2\prod_{\ell=1}^4 
	\Gamma\left(\alpha_\ell^{(k)} + m_\ell^{(k)}\right).
\]
Again, superscripts $(1)$ and $(2)$ refer to non-metadata and metadata
related terms, respectively.

The choice of the hyper-parameter $\boldsymbol{\alpha}$ should reflect 
our belief that concurrent changepoints are more likely to occur than 
when the component series are independent. The ratios between the prior 
expectations satisfy $E(\rho_1) : E(\rho_2) : E(\rho_3) : E(\rho_4) = 
\alpha_1 : \alpha_2 : \alpha_3 : \alpha_4$. If changepoints in the Tmax 
and Tmin series at time $t$ are independent events, then $\rho_1 = 
P(\eta_{t,1} = 1, \eta_{t,2} = 1) = P(\eta_{t,1} = 1) P(\eta_{t,2} = 1) 
= (\rho_1 + \rho_2)(\rho_1 + \rho_3)$. To encourage concurrent shifts, 
$\boldsymbol{\alpha}$ is hence chosen such that
\[ 
E(\rho_1) = 
\frac{\alpha_1}{\sum_{\ell=1}^4 \alpha_{\ell}} > \frac{\alpha_1 + 
\alpha_2} {\sum_{\ell=1}^4 \alpha_{\ell}} ~ \frac{\alpha_1 + 
\alpha_3}{\sum_{\ell=1}^4 \alpha_{\ell}} = E(\rho_1 +\rho_2) E(\rho_1 + 
\rho_3). 
\] 
In addition, the prior probability of not obtaining a changepoint at a 
time is set to its counterpart in the univariate case, i.e., $\alpha_4/ 
\sum_{\ell=1}^4 \alpha_{\ell} = b / (a+b)$. After consulting 
climatologists, default hyper-parameters are set to 
$\boldsymbol{\alpha}^{(1)} = \left(3/7, 2/7, 2/7, 239\right)'$ and 
$\boldsymbol{\alpha}^{(2)} = \left(3/7, 2/7, 2/7, 47\right)'$ for 
monthly data.

To obtain the mixture MDL in a closed form, for a bivariate model with 
$m = m_1+m_2 > 0$ changepoints, the regime means $\boldsymbol{\mu}$ 
again are taken to have independent normal priors
\[
\boldsymbol\mu \mid \boldsymbol\Sigma, \boldsymbol\eta
	\sim \text{N}(\mathbf{0}, \boldsymbol\Omega), \quad
\boldsymbol\Omega = \nu ~ \text{diag}
	\left(\underbrace{\sigma_1^2, \ldots, \sigma_1^2}_{m_1}, 
	\underbrace{\sigma_2^2, \ldots, \sigma_2^2}_{m_2}\right),
\]
where $\sigma_1^2$ and $\sigma_2^2$ are the diagonal entries 
of the white noise covariance $\boldsymbol{\Sigma}$.

\subsection{The bivariate BMDL}

For a model $\boldsymbol\eta$ with $m > 0$, the marginal 
likelihood, after integrating $\boldsymbol{\mu}$ out, has a closed form:
\begin{align*}
	&f(\mathbf{X}_{(p+1):N} \mid \mathbf{s}, \boldsymbol\Sigma, 
	\boldsymbol\Phi_{1:p}, \boldsymbol\eta)\\
\propto &
	\left| \boldsymbol\Sigma \right|^{-\frac{N-p}{2}} \left| \boldsymbol\Omega \right|^{-\frac{1}{2}}
	\left| \widetilde{\mathbf{D}}'  (\boldsymbol\Sigma^{-1}\otimes \mathbf{I}_{N-p})
	\widetilde{\mathbf{D}} + \boldsymbol\Omega^{-1} \right|^{-\frac{1}{2}}	
	 e^{-\frac{1}{2} (\widetilde{\mathbf{X}} - 
	\widetilde{\mathbf{A}}\mathbf{s})' 
	\widetilde{\mathbf{B}}
	(\widetilde{\mathbf{X}} - \widetilde{\mathbf{A}}\mathbf{s})},
\end{align*}
where $\widetilde{\mathbf{B}}$ is modified to
\[
\widetilde{\mathbf{B}} = 
	(\boldsymbol\Sigma^{-1}\otimes \mathbf{I}_{N-p}) \times 
	\left[\mathbf{I}_{2(N-p)} - \widetilde{\mathbf{D}}
	\left\{\widetilde{\mathbf{D}}' (\boldsymbol\Sigma^{-1}\otimes \mathbf{I}_{N-p})
	\widetilde{\mathbf{D}}+\boldsymbol\Omega^{-1}\right\}^{-1}
	\widetilde{\mathbf{D}}'(\boldsymbol\Sigma^{-1}\otimes \mathbf{I}_{N-p})
	\right].
\]
The maximum marginal likelihood estimator $\tilde{\mathbf{s}}$ is 
unaltered from \eqref{eq:s_hat}. However, after plugging 
$\hat{\mathbf{s}}$ back into the likelihood, the maximum likelihood 
estimators of $\boldsymbol{\Sigma}$ and $\boldsymbol{\Phi}_1, \ldots, 
\boldsymbol{\Phi}_p$ do not have closed forms. Again, Yule-Walker 
estimators are used.

To find Yule-Walker estimators for the time series regression 
\eqref{eq:likelihood3}, generalized least squares residuals of the mean 
fit, denoted by $\boldsymbol\epsilon_{1:N}^\text{gls} = 
((\boldsymbol\epsilon_{1:N, 1}^{\text{gls}})', 
(\boldsymbol\epsilon_{1:N, 2}^\text{gls})')' \in \mathbb{R}^{2N}$, are 
computed via
\[ 
\boldsymbol\epsilon_{1:N}^\text{gls} = \left[\mathbf{I}_{2N} - \mathbf{G}
	\left\{ \mathbf{G}'\left(\hat{\boldsymbol{\Gamma}}^\text{ols}(0)^{-1} \otimes 
	\mathbf{I}_N\right)\mathbf{G} \right\}^{-1}
	\mathbf{G}'\left(\hat{\boldsymbol{\Gamma}}^\text{ols}(0)^{-1} \otimes 
	\mathbf{I}_N\right)\right] 
	\mathbf{X}_{1:N},
\]
where 
\[
\mathbf{G} = \left[ \begin{array}{cccc} 
					\mathbf{A}_{1:N, 1} 	& \mathbf{D}_{1:N,1} 	& \mathbf{0}				& \mathbf{0}				\\
					\mathbf{0}			& \mathbf{0}				& \mathbf{A}_{1:N,2} 	& \mathbf{D}_{1:N,2} 	\\
					\end{array} \right].
\]
Here, $\hat{\boldsymbol{\Gamma}}^\text{ols}(0) = N^{-1}\sum_{t=1}^N 
\boldsymbol\epsilon_t^{\text{ols}} 
(\boldsymbol\epsilon_t^{\text{ols}})^\prime$ is a $2 \times 2$ 
covariance matrix of the ordinary (unweighted) least squares residuals 
$\boldsymbol\epsilon_t^{\text{ols}} = (\epsilon_{t, 1}^{\text{ols}}, 
\epsilon_{t, 2}^{\text{ols}})^\prime$, where $\epsilon_{t, 
1}^{\text{ols}}$ and $\epsilon_{t, 2}^{\text{ols}}$ are computed 
analogously to \eqref{eq:Y} with the design matrices 
$[\mathbf{A}_{1:N,1}| \mathbf{D}_{1:N,1}]$ and $[\mathbf{A}_{1:N,2} | 
\mathbf{D}_{1:N,2}]$, respectively. The sample autocovariances at lag $h 
= 0, 1, \ldots, p$ of the generalized least squares residuals 
$\boldsymbol\epsilon_t^\text{gls} = (\epsilon_{t, 1}^\text{gls}, 
\epsilon_{t, 2}^\text{gls})^\prime, t =1, \ldots, N$ are computed as 
$\hat{\boldsymbol{\Gamma}}(h) = N^{-1}\sum_{t = h+1}^N 
\boldsymbol\epsilon_t^\text{gls} 
(\boldsymbol\epsilon_{t-h}^{\text{gls}})^\prime$. The Yule-Walker 
estimators thus obey
\[ 
\left( 
\hat{\boldsymbol{\Phi}}_1, \ldots, \hat{\boldsymbol{\Phi}}_p \right) = 
\left( \widehat{\boldsymbol\Gamma}(1), \ldots,
		\widehat{\boldsymbol\Gamma}(p) \right)
		\left[
		\begin{array}{cccc}
		\widehat{\boldsymbol\Gamma}(0)	& \widehat{\boldsymbol\Gamma}(1)
		& \cdots							& \widehat{\boldsymbol\Gamma}(p-1)\\
		\widehat{\boldsymbol\Gamma}(1)'	& \widehat{\boldsymbol\Gamma}(0)
		& \cdots							& \widehat{\boldsymbol\Gamma}(p-2)\\
		\vdots							& \vdots
		& \ddots							& \vdots								\\
		\widehat{\boldsymbol\Gamma}(p-1)'	& \widehat{\boldsymbol\Gamma}(p-2)'
		& \cdots							& \widehat{\boldsymbol\Gamma}(0)\\
	\end{array}
	\right]^{-1}
\]
and $\widehat{\boldsymbol{\Sigma}} = \hat{\boldsymbol{\Gamma}}(0) - 
\sum_{j = 1}^p \hat{\boldsymbol{\Phi}}_j \hat{\boldsymbol{\Gamma}}(j)'$.

After plugging $\widehat{\boldsymbol\Sigma}$ and 
$\widehat{\boldsymbol\Phi}_1, \ldots, \widehat{\boldsymbol\Phi}_p$ back 
into the marginal likelihood, the terms $\widetilde{\mathbf{X}}, 
\widetilde{\mathbf{A}}, \widetilde{\mathbf{D}}, \widetilde{\mathbf{B}}$, 
and $\boldsymbol\Omega$, which depend on $\boldsymbol\Sigma$ and 
$\boldsymbol\Phi_1, \cdots, \boldsymbol\Phi_p$, are denoted by 
$\widehat{\mathbf{X}}, \widehat{\mathbf{A}}, \widehat{\mathbf{D}}, 
\widehat{\mathbf{B}}, \widehat{\boldsymbol\Omega}$, respectively. Hence, 
the Bayesian MDL for $\boldsymbol{\eta}$ is (up to a constant)
\begin{align*}
&\text{BMDL}(\boldsymbol\eta)\\
=~& \frac{N-p}{2}\log \left( \left| \widehat{\boldsymbol\Sigma} \right| \right) 
	+ \frac{1}{2}\sum_{i=1}^2 m_i \log (\nu\hat{\sigma}_i^2)  
	+ \frac{1}{2}\log \left( \left| \widehat{\mathbf{D}}'  
	(\widehat{\boldsymbol\Sigma}^{-1}\otimes \mathbf{I}_{N-p})
	\widehat{\mathbf{D}} + \widehat{\boldsymbol\Omega}^{-1} \right|\right)\\
&	+ \frac{1}{2}\widehat{\mathbf{X}}' \left\{
 	\widehat{\mathbf{B}}
	- \widehat{\mathbf{B}} \widehat{\mathbf{A}}
	\left(\widehat{\mathbf{A}}' 
	\widehat{\mathbf{B}}\widehat{\mathbf{A}}\right)^{-1}
	\widehat{\mathbf{A}}' \widehat{\mathbf{B}}
	\right\}\widehat{\mathbf{X}} - \sum_{k = 1}^2\sum_{\ell=1}^4 
	\log \left\{ \Gamma\left(\alpha_{\ell}^{(k)} + m_{\ell}^{(k)}\right) \right\}.
\end{align*}
Under the null model $\boldsymbol\eta_{\o}$, since $\widehat{\mathbf{B}} 
= \widehat{\boldsymbol\Sigma}^{-1}\otimes \mathbf{I}_{N-p}$, with the 
convention that the determinant of a $0 \times 0$ matrix is unity, the 
above BMDL still holds.

\section{Simulation Studies}\label{sec:simulation} 

This section studies changepoint detection performance under finite 
samples via simulation. Our simulation parameters are selected to 
roughly resemble the bivariate Tuscaloosa data, which will be studied in 
Section \ref{sec:Tuscaloosa}. Specifically, the bivariate error series 
$\{ \boldsymbol\epsilon_t \}$ is chosen to follow a zero mean Gaussian 
VAR model with $p=3$.  The VAR parameters are taken as
\[
\boldsymbol\Phi_1 = \left(
  \begin{array}{cc}
	0.2 & 0.02 \\ 
	0.02 & 0.2 
  \end{array} \right),
\boldsymbol\Phi_2 = \left(
  \begin{array}{cc}
	0.1 & 0.01 \\ 
	0.01 & 0.1 
  \end{array} \right),
\boldsymbol\Phi_3 = \left(
  \begin{array}{cc}
	0.05 & 0.005 \\ 
	0.005 & 0.05 
  \end{array} \right),
\]
and 
\[
\boldsymbol\Sigma = \left(
  \begin{array}{cc}
	9 & 2 \\ 
	2 & 9 
 \end{array} \right).
\]
In each of 1000 independent runs, 50 year monthly Tmax and Tmin series 
($N = 600$) are simulated with $m = 3$ changepoints in each series.  For 
the Tmax series, mean shifts are placed at the times $150, 300$, and 
$450$.  The regime means have form $\boldsymbol{\mu}_1 = (0, \Delta, 2 
\Delta, 3 \Delta)'$ where $\Delta > 0$ will be varied.  For the Tmin 
series, mean shifts are placed at times $150, 300$, and $375$.  The 
regime means are $\boldsymbol{\mu}_2 = (0, -\Delta, \Delta, 0)'$. Here, 
Tmax has monotonic ``up, up, up'' shifts of equal shift magnitudes; Tmin 
shifts in a ``down, up, down'' fashion and the second shift is twice as 
large as the other two shifts. The shifts at times 150 and 300 are 
concurrent in both series.

Seasonal means are set to $\mathbf{s} = (0, 3, 10, 18, 26, 33, 36, 36, 
31, 20, 8, 2)'$ in both series. Seasonal mean parameters are not 
critical, but the $\Delta$ parameter controlling the mean shift size is. 
Our detection powers will be reported under different signal to noise 
ratios, measured by $\kappa = \Delta / \sigma$. Our study examines 
$\kappa \in \{ 1, 1.5, 2 \}$, with $\sigma = 3$ agreeing with the 
diagonal elements of $\boldsymbol\Sigma$. For metadata, a record 
containing four documented changes at the times $75, 150, 250$, and 
$550$ is posited. Among the documented times, only time 150 is a true 
changepoint.

A simulated series with $\kappa = 1.5$ is shown in Figure 
\ref{fg:simulation_sample1}.  Figure \ref{fg:simulation_sample2} in the 
Appendix shows the same series after subtraction of sample monthly 
means.

\begin{figure}
\centering
\includegraphics[width = 0.6\textwidth, angle = 270]
{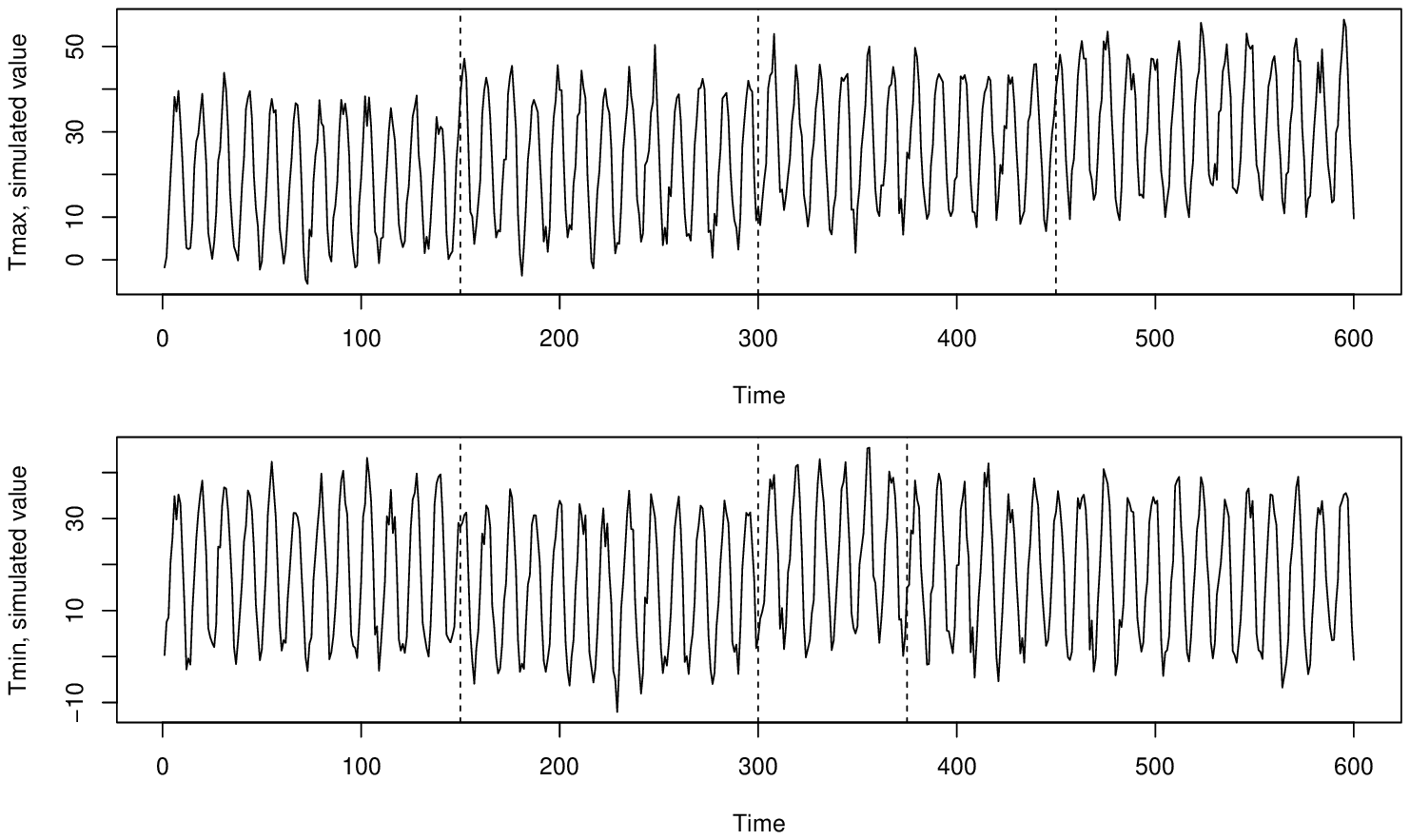}
\caption{\label{fg:simulation_sample1}
A simulated dataset with the signal to noise ratio $\kappa = 1.5$, 
which has three changepoints in Tmax (top panel) 
and three changepoints in Tmin (bottom panel). Vertical dashed lines 
demarcate the true changepoint times.}
\end{figure}

\subsection{Univariate simulations}\label{sec:simulation_univariate}

First, the Tmax and Tmin series are analyzed separately, each fitted by 
univariate BMDL methods with default parameters, once with the 
fictitious metadata and once without metadata.  We also compare various 
methods without metadata, including a BMDL under the objective Bayes 
parameters $a=b=1$ (denoted by oBMDL), the automatic MDL 
 (denoted by MDL), and the BIC.
The MDL obtained 
when the automatic code length rules in \citet{Davis_etal_2006} are 
applied to our multiple mean shift problem:
\begin{equation}
\label{eq:MDL_univariate}
\text{MDL}(\boldsymbol\eta) = 
	 \frac{N-p}{2}\log \left( \hat{\sigma}^2_{\nu = \infty} \right) 
	 + \frac{1}{2}\sum_{r=2}^{m+1}\log(N_r) + \log(m + 1) + (m+1) \log(N-p).
\end{equation}
The first term in \eqref{eq:MDL_univariate} approximates the negative 
logarithm of the maximum likelihood, where the Yule-Walker estimator of 
$\sigma^2$ is used, which equals \eqref{eq:sigmasq_hat} with $\nu = 
\infty$ after $\boldsymbol\phi$ is replaced by $\hat{\boldsymbol\phi}$. 
This estimator is denoted by $\hat{\sigma}^2_{\nu = \infty}$ here. The 
other terms in \eqref{eq:MDL_univariate} are the two-part MDLs for the 
regime means $\mu_2, \ldots, \mu_{m+1}$, the number of changepoints $m$ 
(the original penalty of $\log(m)$ is undefined for the null model with 
$m=0$; the ad-hoc fix to this simply uses $m+1$ in the logarithm), and 
the regime lengths $N_1, \ldots, N_{m+1}$, respectively. The two-part 
MDLs of the global parameters $\mathbf{s}$, $\sigma^2$, and 
$\boldsymbol\phi$ are constants and hence omitted. An MDL for the AR 
order $p$ is not needed as $p$ is tacitly assumed known. 
BIC, up to a constant, is
\[
\text{BIC} (\boldsymbol\eta) = \frac{N-p}{2}\log \left( 
\hat{\sigma}^2_{\nu = \infty} \right) + m \log(N-p).
\] 

In each fit, an MCMC chain of 100,000 iterations is generated.  The optimal 
multiple changepoint model is taken as the one that minimizes the objective 
function.

\begin{table}
    \caption{\label{tb:simulation_Tmax_uni}
    Univariate results for Tmax, aggregated from 1000 simulated 
datasets. The detection rates for the documented change when metadata is 
used are in bold.}
  \centering
\fbox{%
\begin{tabular}{ c  c l | c c c | c | c }
\multirow{2}{*}{$\kappa$}	&  \multirow{2}{*}{Metadata}	& \multirow{2}{*}{Method}	& \multicolumn{3}{c|}{True positive detection $(\%)$} 	&Average false positive   	&  \multirow{2}{*}{$\hat{m}$ $(se)$}	 \\ 
 			& 			& 		 	& $t = 150$ & $t = 300$ & $t = 450$ 					& detection	$(\%)$ 			&   \\ 
  \hline
  & yes & BMDL & {\bf 58.8} & 16.8 & 14.5 & 0.29 & 2.65 (0.56) \\ 
   & no & BMDL & 15.4 & 16.3 & 16.4 & 0.36 & 2.61 (0.61) \\ 
  1.0 & no & oBMDL & 14.4 & 16.9 & 16.1 & 0.37 & 2.68 (0.59) \\ 
   & no & MDL & 14.9 & 17.2 & 16.2 & 0.36 & 2.64 (0.62) \\ 
   & no & BIC & 17.0 & 17.4 & 18.3 & 0.43 & 3.07 (0.54) \\ 
  \hline
   & yes & BMDL & {\bf 75.7} & 41.7 & 37.9 & 0.25 & 3.02 (0.13) \\ 
   & no & BMDL & 36.3 & 40.8 & 37.1 & 0.31 & 3.02 (0.13) \\ 
  1.5 & no & oBMDL & 36.5 & 41.3 & 37.2 & 0.31 & 3.03 (0.17) \\ 
   & no & MDL & 37.6 & 41.3 & 37.0 & 0.31 & 3.02 (0.15) \\ 
   & no & BIC & 37.0 & 40.2 & 36.3 & 0.33 & 3.12 (0.38) \\ 
  \hline
   & yes & BMDL & {\bf 84.1} & 59.3 & 57.6 & 0.17 & 3.02 (0.14) \\ 
   & no & BMDL & 54.2 & 59.7 & 57.2 & 0.22 & 3.02 (0.15) \\ 
  2.0 & no & oBMDL & 54.4 & 59.4 & 57.3 & 0.22 & 3.03 (0.18) \\ 
   & no & MDL & 54.7 & 59.4 & 58.0 & 0.22 & 3.02 (0.16) \\ 
   & no & BIC & 53.4 & 59.1 & 56.9 & 0.24 & 3.11 (0.36) \\ 
\end{tabular}
 }
\end{table}

For Tmax series, Table \ref{tb:simulation_Tmax_uni} reports empirical 
detection percentages, including true positive rates at the exact times 
of changepoints and average false positive rates at non-changepoint 
times, along with estimated number of changepoints $\hat{m}$ and its 
standard error. When metadata is ignored, since the three shifts are of 
equal size $\Delta$, their detection rates are similar. False detection 
rates are very low; even when $\kappa = 1$, on average, a 
non-changepoint is flagged 0.43\% of the time or less.

Among different methods without metadata, detection rates of true 
changepoints are similar, while BIC flags slightly more false positives 
than MDL-based methods (BMDL, oBMDL, and MDL). When $\kappa = 1$, the 
number of changepoints $m=3$ is underestimated by the MDL-based methods 
and better estimated by BIC penalties; when $\kappa = 1.5$ and $2$, $m$ 
is better estimated by the MDL-based methods, and overestimated by BIC. 
Overall, BIC tends to favor models with more changepoints than the 
MDL-based methods. As suggested by Theorem \ref{prop:BMDL_MDL} below, 
the BMDL performs similarly to the automatic MDL.

Interestingly, without metadata, the BMDL under the default parameters 
$a = 1$ and $b = 239$ and the objective choices $a=b=1$ perform 
similarly. Figure \ref{fg:BMDL_penalties} in the Appendix reveals that 
as functions of $m$, the code lengths $\mathcal{L}(\boldsymbol\eta) = 
-\log\{\pi(\boldsymbol\eta)\}$ under BMDL and oBMDL have similar shapes, 
with a nearly constant difference over the region where $m$ is small. In 
this case, if knowledge of changepoint frequency is not available, a 
BMDL can still be used with objective parameters.


Metadata use substantially increases detection power for the BMDL. In 
Figure \ref{fg:simulation_Tmax}, the true documented change at time 150 
is detected $75.7\%$ of the time when metadata is used, more than twice 
as high ($36.3\%$) when metadata is eschewed. Moreover, times near the 
changepoint at time 150 are less likely to be erroneously flagged as 
changepoints.  Our prior belief that metadata times are more likely to 
be changepoints is especially important when the mean shift is small: 
when $\kappa = 1$, using metadata increases the detection rate of the 
time 150 changepoint from $15.4\%$ to $58.8\%$. For false positives, 
Figure \ref{fg:simulation_Tmax} shows that using metadata does not 
increase false detection rates at the documented times 75, 250, and 550 
(where no shifts occur).  This suggests that the prior distribution does 
not ``overwhelm" the data. Table \ref{tb:simulation_Tmax_uni} shows that 
average false positive rates even drop after using metadata.

\begin{figure}
\centering
\includegraphics[width = 0.6\textwidth, angle = 270]
{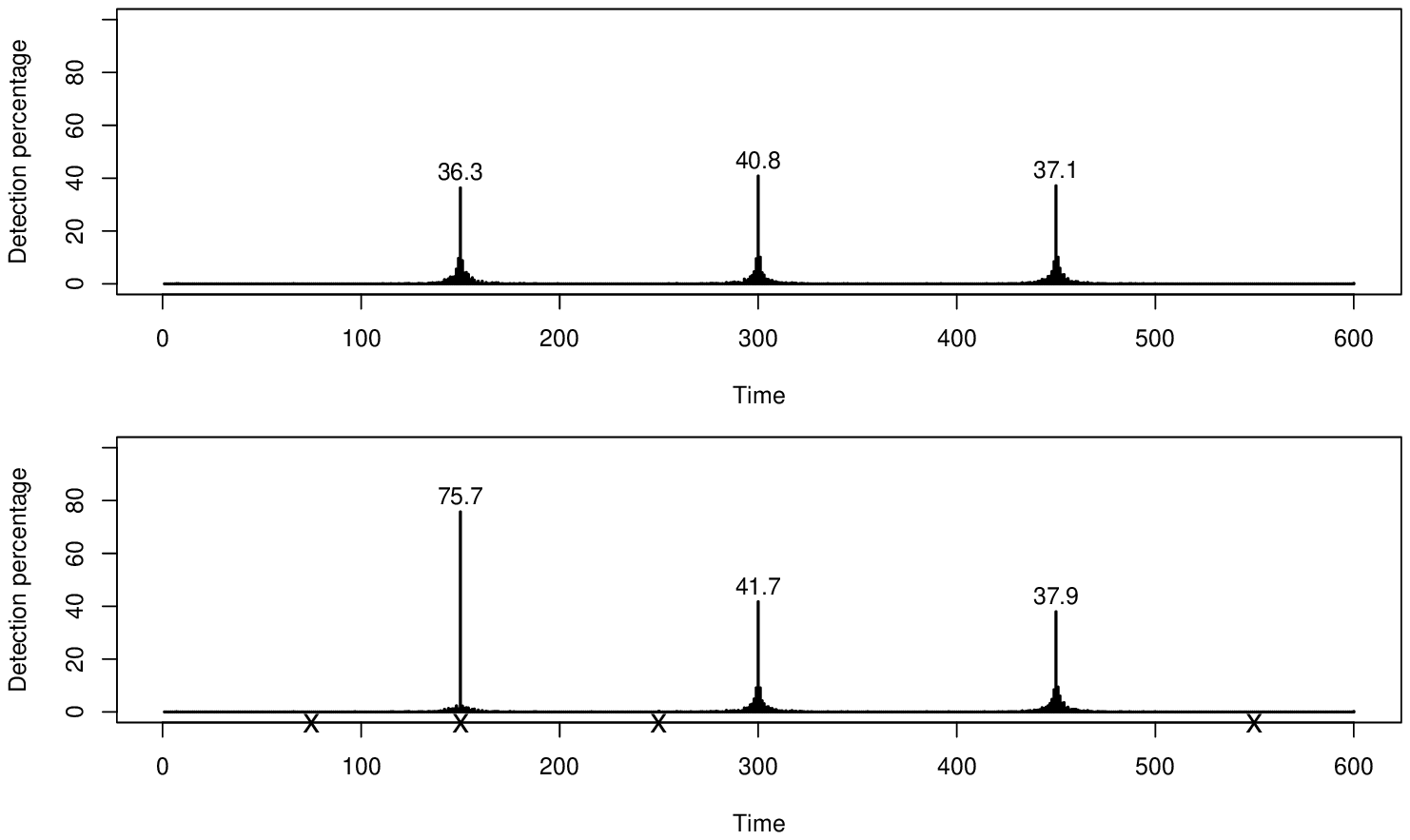}
\caption{\label{fg:simulation_Tmax}
Detection times and percentages of changepoints in Tmax series 
using univariate BMDL.  The top panel ignores the four metadata times; 
the bottom panel uses the metadata (metadata times are marked as crosses 
on the axis). Numerical percentages on the graphic are for detection at 
``their exact times". The results are aggregated from 1000 independent 
simulated datasets with $\kappa = 1.5$.}
\end{figure}

\begin{table}
\caption{\label{tb:simulation_Tmin_uni}
Univariate results for Tmin, aggregated from 1000 simulated 
datasets. Detection rates for the documented change when metadata is 
used are in bold.}
\centering
\fbox{
\begin{tabular}{c  c l | c c c | c | c}
\multirow{2}{*}{$\kappa$}	&  \multirow{2}{*}{Metadata}	& \multirow{2}{*}{Method}	& \multicolumn{3}{c|}{True positive detection $(\%)$} 	&Average false positive   	&  \multirow{2}{*}{$\hat{m}$  $(se)$}\\ 
 			& 			& 		 	& $t = 150$ & $t = 300$ & $t = 375$ 					& detection	$(\%)$ 			&   \\ 
  \hline
   & yes & BMDL & {\bf 62.0} & 53.5 & 14.3 & 0.23 & 2.69 (0.77) \\ 
   & no & BMDL & 18.0 & 52.4 & 14.1 & 0.30 & 2.63 (0.86) \\ 
  1.0 & no & oBMDL & 18.7 & 54.9 & 14.6 & 0.31 & 2.76 (0.71) \\ 
   & no & MDL & 17.4 & 50.5 & 13.6 & 0.28 & 2.50 (0.99) \\ 
   & no & BIC & 19.5 & 55.0 & 15.8 & 0.36 & 3.07 (0.52) \\ 
  \hline
   & yes & BMDL & {\bf 77.3} & 84.4 & 38.2 & 0.17 & 3.01 (0.15) \\ 
   & no & BMDL & 37.4 & 84.7 & 39.5 & 0.24 & 3.02 (0.17) \\ 
  1.5 & no & oBMDL & 37.5 & 84.3 & 38.9 & 0.24 & 3.03 (0.20) \\ 
   & no & MDL & 37.2 & 84.3 & 38.6 & 0.24 & 3.01 (0.15) \\ 
   & no & BIC & 36.5 & 83.3 & 38.0 & 0.26 & 3.13 (0.44) \\ 
  \hline
   & yes & BMDL & {\bf 85.2} & 95.4 & 56.1 & 0.11 & 3.01 (0.13) \\ 
   & no & BMDL & 58.2 & 95.4 & 56.4 & 0.15 & 3.02 (0.13) \\ 
  2.0 & no & oBMDL & 58.2 & 95.2 & 56.5 & 0.16 & 3.03 (0.18) \\ 
   & no & MDL & 58.0 & 95.5 & 56.9 & 0.15 & 3.01 (0.12) \\ 
   & no & BIC & 57.7 & 95.5 & 55.7 & 0.17 & 3.12 (0.43) \\ 
\end{tabular}
}
\end{table}

For Tmin series, the non-monotonic shift aspect (down, up, down) that 
troubles some at most one change (AMOC) binary segmentation approaches 
\citep{Li_Lund_2012} is well handled by all multiple changepoint 
detection methods examined. Table \ref{tb:simulation_Tmin_uni} shows 
that when metadata is ignored, the larger shift at time 300 is more 
easily detected than the two smaller shifts at times 150 and 375. When 
metadata is used, the detection rate of the time 150 shift becomes 
comparable to the detection rate of time 300 shift, which is twice as 
large in size, but is not a metadata time. False positive rates are 
uniformly low, and $m$ is well-estimated by MDL-based methods when 
$\kappa$ is not too small. Again, without metadata, the MDL-based 
methods are similar, while BIC tends to favor models with larger $m$.

\subsection{Bivariate simulations}

Since the BMDL is flexible enough to handle non-concurrent shifts for 
bivariate series, we now apply it to Tmax and Tmin series jointly. Each 
bivariate series is fitted by an MCMC chain of 50,000 iterations, once 
without metadata, and once with metadata. Metadata impacts are similar 
to the univariate case, increasing detection of true mean shifts at 
metadata times and also slightly decreasing average false positive rates 
(see Tables \ref{tb:simulation_Tmax_bi} and 
\ref{tb:simulation_Tmin_bi}). Figure \ref{fg:simulation_txtn} shows 
bivariate detection rates with metadata when $\kappa = 1.5$. For the 
non-concurrent shift times at 375 and 450, detection rates for the 
component series are very different; in most runs, concurrent shifts are 
not erroneously signaled.

\begin{figure}
\centering
\includegraphics[width = 0.6\textwidth, angle = 270]
{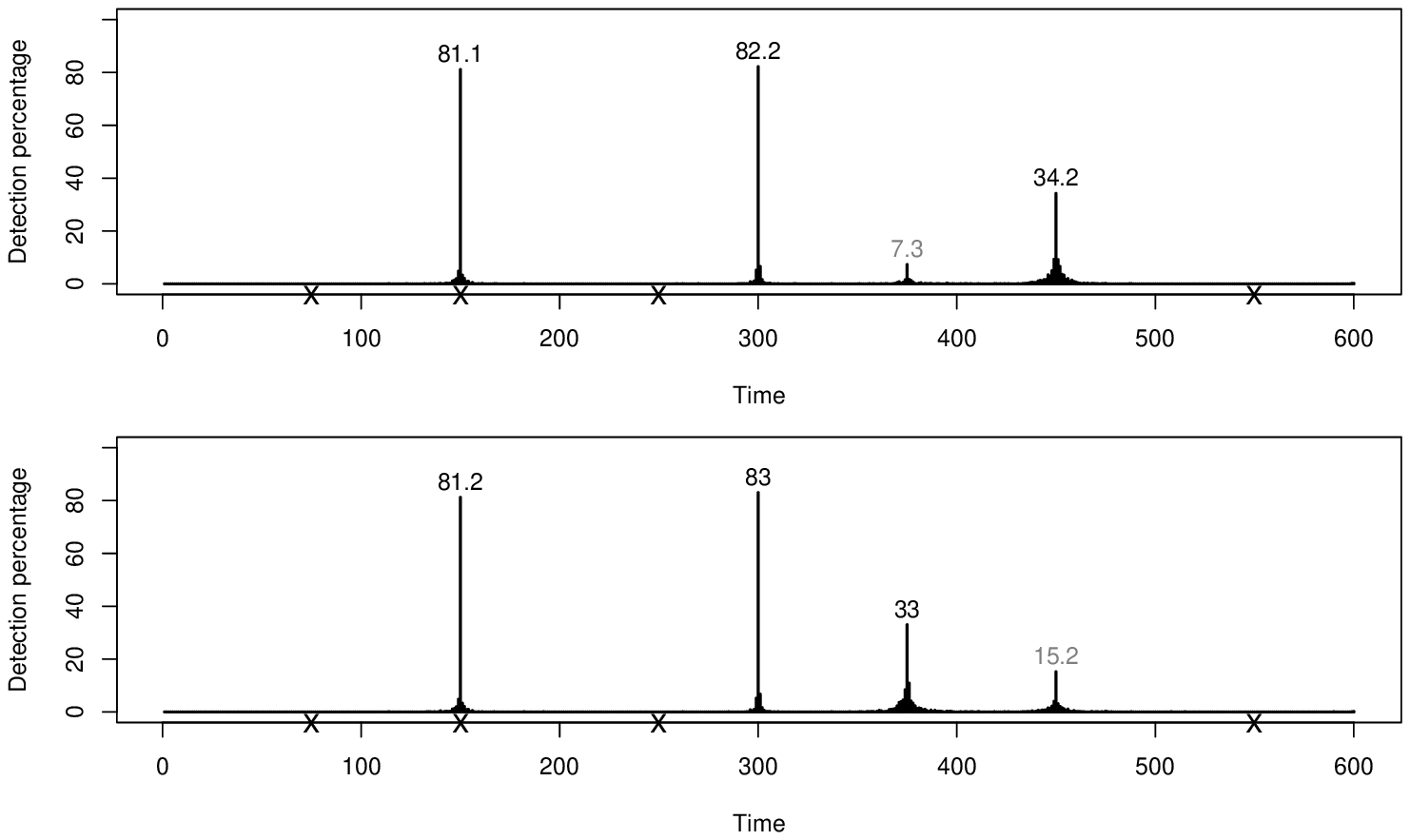}
\caption{\label{fg:simulation_txtn}
Detection percentages of Tmax (top panel) and Tmin (bottom 
panel) using bivariate BMDL methods with metadata (metadata times are 
marked as crosses on the axis). Numerical percentages on the graphic are 
for detection at ``their exact times". The results are aggregated from 
1000 independent simulated datasets with $\kappa = 1.5$.}
\end{figure}

\begin{table}
\caption{\label{tb:simulation_Tmax_bi}
Bivariate results for Tmax by BMDL, aggregated from 1000 simulated datasets.}
\centering
\fbox{
\begin{tabular}{c  c  | c c c | c c | c}
\multirow{2}{*}{$\kappa$}	&  \multirow{2}{*}{Metadata}	& \multicolumn{3}{c|}{True positive detection $(\%)$} 	&\multicolumn{2}{c|}{False positive   detection	$(\%)$}	&  \multirow{2}{*}{$\hat{m}$  $(se)$}	 \\ 
 			& 			& $t = 150$ & $t = 300$ & $t = 450$ 			& $t=375$	& average  			&   \\ 
  \hline
\multirow{2}{*}{1.0} & yes & 60.7 & 54.5 & 11.5 & 6.8 & 0.31  & 3.12 (0.45) \\ 
   & no & 36.5 & 55.2 & 11.4 & 8.3 & 0.36 & 3.19 (0.48) \\ 
  \hline
  \multirow{2}{*}{1.5} & yes & 81.1 & 82.2 & 34.2 & 7.3 & 0.20 & 3.18 (0.43) \\ 
   & no & 66.7 & 82.9 & 33.9 & 10.8 & 0.24 & 3.29 (0.47) \\ 
  \hline
  \multirow{2}{*}{2.0} & yes & 92.1 & 93.5 & 55.9 & 3.7 & 0.11 & 3.07 (0.28) \\ 
   & no & 84.7 & 94.8 & 55.6 & 6.2 & 0.13 & 3.13 (0.35) \\ 
\end{tabular}
}
\end{table}

\begin{table}
\caption{\label{tb:simulation_Tmin_bi}
Bivariate results for Tmin by BMDL, aggregated from 1000 
simulated datasets.}
\centering
\fbox{
\begin{tabular}{c  c  | c c c | c c | c}
\multirow{2}{*}{$\kappa$}	&  \multirow{2}{*}{Metadata}	& \multicolumn{3}{c|}{True positive detection $(\%)$} 	&\multicolumn{2}{c|}{False positive   detection	$(\%)$}	&  \multirow{2}{*}{$\hat{m}$  $(se)$}	 \\ 
 			& 			& $t = 150$ & $t = 300$ & $t = 375$ 			& $t=450$	& average  			&   \\ 
  \hline
\multirow{2}{*}{1.0} & yes & 60.1 & 54.9 & 9.5 & 8.7 & 0.31 & 3.10 (0.57) \\ 
   & no & 36.2 & 55.3 & 10.2 & 9.6 & 0.36 & 3.17 (0.55) \\ 
  \hline
  \multirow{2}{*}{1.5} & yes & 81.2 & 83.0 & 33.0 & 15.2 & 0.24 & 3.38 (0.54) \\ 
   & no & 66.4 & 83.4 & 34.2 & 21.3 & 0.30 & 3.61 (0.54) \\ 
  \hline
  \multirow{2}{*}{2.0} & yes & 92.0 & 94.8 & 57.8 & 16.2 & 0.14 & 3.28 (0.46) \\ 
   & no & 84.8 & 95.1 & 54.9 & 32.1 & 0.21 & 3.59 (0.53) \\ 
\end{tabular}
}
\end{table}

While concurrent shifts are not always the case, they are believed to be 
more likely in our parameter elicitation settings. Compared to the 
univariate BMDL, the bivariate method enhances detection power of 
concurrent changepoints. When $\kappa = 1.5$, at time 150, where Tmax 
(Tmin) shifts $\Delta$ ($-\Delta$), the bivariate BMDL increases the 
univariate detection rates from both series from about $77\%$ to above 
$81\%$.  At time 300, where Tmax (Tmin) shifts by $\Delta$ ($2\Delta$), 
the detection rate increases from $41.1\%$ to $82.2\%$ for Tmax, while 
it remains roughly the same for Tmin. Tables \ref{tb:simulation_Tmax_bi} and 
\ref{tb:simulation_Tmin_bi} show that detection power gains under the 
bivariate approach are greater for small signals $\kappa = 1$, without 
metadata. An interesting phenomenon is observed: bivariate BMDL improves 
univariate methods more when the concurrent shifts move the series in 
opposite directions than in the same direction (detection rates at time 
300 do not increase for Tmin). Because Tmax and Tmin are positively 
correlated series, concurrent shifts in the same direction may be 
misattributed to positively correlated errors; this cannot happen for 
shifts in opposite directions.

Overall, while bivariate detection does not induce more false positives, 
it tends to flag more false positives at locations where the mean in the 
other series shifts. Figure \ref{fg:simulation_txtn} shows that at time 
375, a changepoint time in Tmin but not in Tmax, a false detection rate 
of $7.3\%$ for Tmax is obtained.  At time 450, a changepoint in Tmax but 
not Tmin, a false detection rate of $15.2\%$ is obtained for Tmin. These 
false positive rates slightly degrade inferences at nearby changepoints; 
for example, at time 450 for Tmax and time 375 for Tmin, detection rates 
are $34.2\%$ and $33.0\%$, respectively, slightly lower than the 
$37.9\%$ and $38.2\%$ reported in the univariate case.  Finally, Tables 
\ref{tb:simulation_Tmax_bi} and \ref{tb:simulation_Tmin_bi} show that 
the bivariate approach tends to overestimate $m$, which differs from the 
univariate case.

\section{The Tuscaloosa Data}
\label{sec:Tuscaloosa} 

Monthly Tmax and Tmin series from Tuscaloosa, Alabama (the target 
station) over the 114 year period January, 1901 -- December, 2014 are 
plotted in Figure \ref{fg:Tuscaloosa1}.  \cite{Lu_etal_2010} study 
annually averaged values of this series from 1901-2000. The Tuscaloosa 
metadata lists station relocations in November 1921, March 1939, June 
1956, and May 1987; November 1956 and May 1987 are listed as instrument 
change times.

\begin{figure}
\centering
\includegraphics[width = 0.6\textwidth, angle = 270]{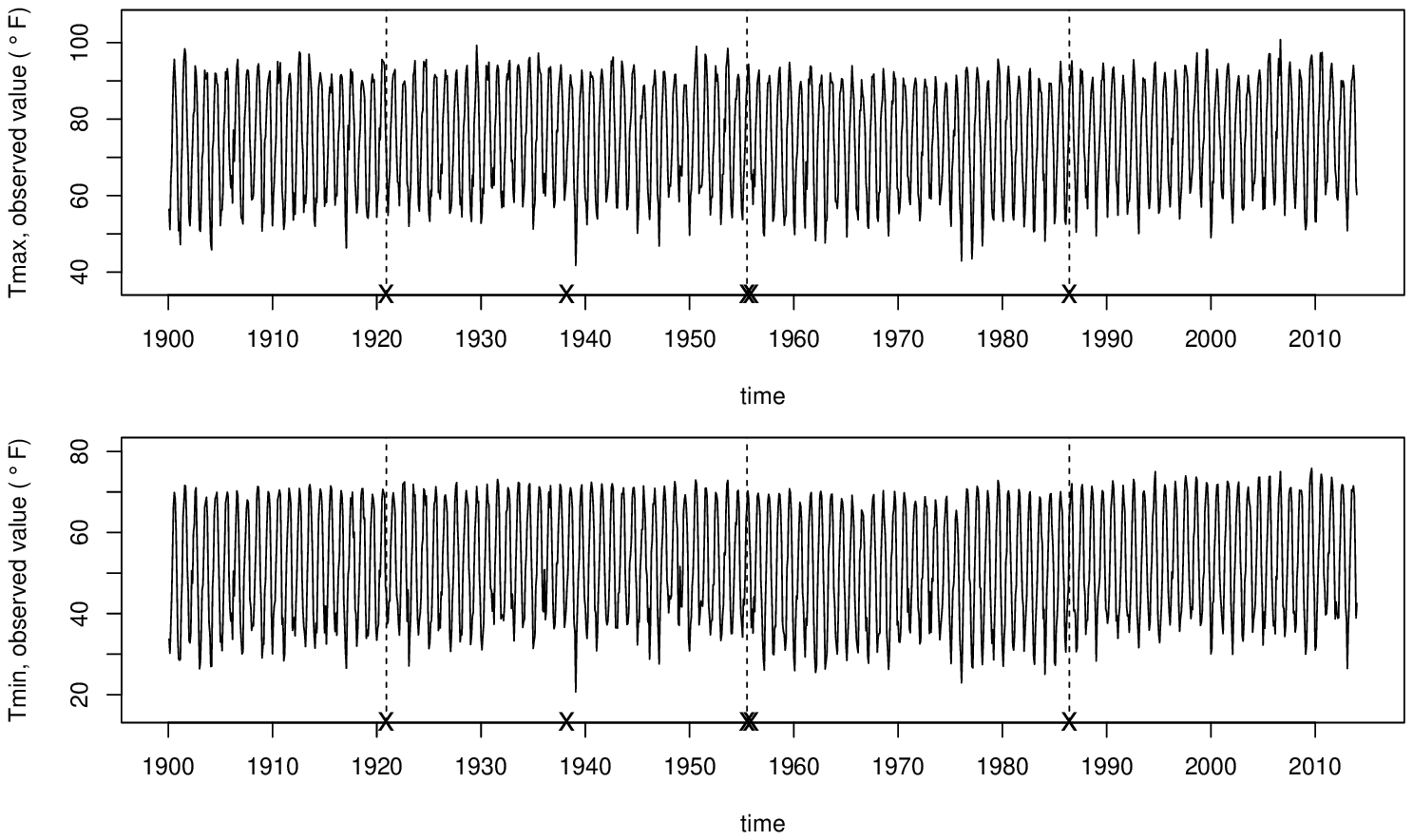}
\caption{\label{fg:Tuscaloosa1}
Tuscaloosa monthly Tmax (top panel) and Tmin (bottom panel) series. Metadata times 
are marked with crosses on the axis. Vertical dashed lines show estimated changepoint times 
from bivariate BMDL with metadata.}
\end{figure}

In this section, the Tmax and Tmin series will be analyzed from both 
univariate and bivariate perspectives via the penalization methods of 
Section \ref{sec:simulation}. All parameters are set to default values; 
the AR order $p = 2$ is judged as appropriate: by Figure 
\ref{fg:Tuscaloosa_acf} in the Appendix, almost all sample 
autocorrelations of residuals fitted with $p=2$ lie inside pointwise 
$95\%$ confidence bands.

To ensure convergence in the MCMC search algorithm, for each fit, 50 
Markov chains are generated from different starting points, each 
containing 1,000,000 (univariate) or 100,000 (bivariate) iterations. 
Among all changepoint models visited by the 50 Markov chains, the one 
with the smallest BMDL is reported as the optimal model.

\subsection{Univariate fits}

The top half of Table \ref{tb:Tuscaloosa_analysis} displays estimated 
changepoints for the univariate fits. When metadata is ignored, all 
methods (BMDL, oBMDL, MDL, and BIC) estimate the same optimal 
changepoint configuration: Tmax has two estimated changepoints and Tmin 
has three; of these, only January 1990 is a concurrent change. Another 
changepoint is approximately concurrent: March 1957 for Tmax and July 
1957 for Tmin.  The 1918 changepoint flagged for Tmin is close to the 
station relocation in November 1921; the station relocation in June 1956 
and the equipment change in November 1956 are near the two estimated 
changepoints in 1957.  The metadata time in May 1987 is about three 
years from the concurrent changepoints flagged in January 1990. Of 
course, when metadata is ignored, estimated changepoint times may not 
coincide (exactly) with metadata times.

\begin{table}
\caption{\label{tb:Tuscaloosa_analysis}
Estimated changepoints for the Tuscaloosa data.}
\centering
\fbox{
\begin{tabular}{c | c | l}
Metadata	 		& 	Series	& Estimated changepoints \\ 
   \hline
\multicolumn{3}{c}{Univariate}\\
   \hline
\multirow{2}{*}{yes}	&   Tmax & 1956 Nov, 1987 May \\ 
	 				&   Tmin & 1921 Nov, 1956 Jun, 1987 May \\ 
  \hline
 \multirow{2}{*}{no}	& 	Tmax 	& 1957 Mar, 1990 Jan \\ 
	 				& 	Tmin 	& 1918 Feb, 1957 Jul, 1990 Jan \\ 
  \hline
\multicolumn{3}{c}{Bivariate}\\
   \hline
\multirow{2}{*}{yes}	&   Tmax & 1921 Nov, 1956 Jun, 1987 May \\ 
	 				&   Tmin & 1921 Nov, 1956 Jun, 1987 May \\ 
  \hline
 \multirow{2}{*}{no}	& 	Tmax 	& 1918 Feb, 1957 Jul, 1988 Jul  \\ 
	 				& 	Tmin 	& 1918 Feb, 1957 Jul, 1988 Jul \\ 
\end{tabular}
}
\end{table}

Repeating the above analysis with metadata, two changepoints are found 
in Tmax and three in Tmin. All estimated changepoint times now coincide 
with metadata times. Only the May 1987 changepoint is concurrent. 
Between Tmax and Tmin, the two estimated changepoints in 1956 (i.e., the 
two metadata times in 1956) are just a few months apart. As parameter 
estimates are similar with or without metadata, only estimates for the 
optimal changepoint model with metadata are reported.  For Tmax, 
estimated regime means are (one standard error is in parentheses) 
$\hat{\mu}_2 = -1.50~(0.24)$ and $\hat{\mu}_3 = 0.66~(0.25)$ (recall 
that $\mu_1 = 0$); estimated AR(2) coefficients are $\hat{\phi}_1 = 
0.21, \hat{\phi}_2 = 0.05$, and $\hat{\sigma}^2 = 11.59$. For Tmin, the 
estimated parameters are $\hat{\mu}_2 = 1.76~(0.21), \hat{\mu}_3 = 
-1.06~(0.22), \hat{\mu}_4 = 2.35~(0.24), \hat{\phi}_1 = 0.18, 
\hat{\phi}_2 = 0.05$, and $\hat{\sigma}^2 = 10.81$. The concurrent May 
1987 changepoint shifts both series to warmer regimes.

\subsection{Bivariate fits}

Both Tmax and Tmin series are now analyzed in tandem with our methods. 
Three changepoints are detected in both series, with or without 
metadata, and all are concurrent (see the bottom half of Table 
\ref{tb:Tuscaloosa_analysis}).  Figure \ref{fg:Tuscaloosa1} illustrates 
the optimal bivariate BMDL changepoint configuration. When metadata is 
used, all estimated changepoint times migrate to metadata times. 
Comparing to the univariate results, the bivariate approach yields the 
same changepoint configuration for Tmin; for Tmax, a new changepoint in 
November 1921 is flagged and the November 1956 changepoint moves to June 
1956, thus becoming a concurrent change.  For this changepoint 
configuration, the estimated VAR parameters are
\[
\widehat{\boldsymbol{\Phi}}_1 = 
\left(
\begin{array}{cc}
	0.21 & -0.01 \\ 
	-0.02 & 0.20 
\end{array} 
\right),
~
\widehat{\boldsymbol{\Phi}}_2 = 
\left(
\begin{array}{cc}
	0.06 & -0.02 \\ 
       -0.04 & 0.08 
  \end{array} \right), 
~
\widehat{\boldsymbol{\Sigma}} = 
\left(
\begin{array}{cc}
	11.56 & 8.13 \\ 
	8.13 & 10.81 
\end{array} 
\right).
\]

In temperature homogenization problems, the goal is often to detect (and 
then adjust for) ``artificial" changes.  Naturally occurring climate 
shifts should be left in the record if possible.  Because of this, 
analyses often consider target minus reference series, where a reference 
series is a record from a nearby station that shares similar weather 
with the target station. A changepoint detection analysis using 
bivariate BMDL is performed on target minus reference data, and is 
included in the Appendix Section \ref{subsec:target_minus_reference}.

\section{Asymptotic Properties of the Univariate BMDL}
\label{section:asymptotics}

Infill asymptotics, which assume regime lengths tend to infinity with 
the sample size $N$, have been widely adopted to study consistency of 
multiple changepoint detection procedures \citep{Davis_etal_2006, 
Davis_Yau_2013, Du_etal_2016}. Under infill asymptotics, a relative 
changepoint configuration with $m$ changepoints is denoted by 
$\boldsymbol{\lambda} = (\lambda_1, \ldots, \lambda_{m})'$, where $0 < 
\lambda_1 < \cdots < \lambda_m < 1$. Here, time is scaled to $[0,1]$ by 
mapping time $t$ to $t/N$. For the edges, set $\lambda_0=0$ and 
$\lambda_{m+1}=1$.  For a given $N$, the $r$th changepoint location 
$\tau_r$ can be recovered from $\boldsymbol{\lambda}$ via $\tau_r = 
\lfloor \lambda_r N \rfloor$. The length of the $r$th regime, $N_r = 
\lfloor \lambda_r N \rfloor - \lfloor \lambda_{r-1} N \rfloor$, 
satisfies $\lim_{N\rightarrow \infty} N_r / N = \lambda_r - 
\lambda_{r-1}$, for $r = 1, \ldots, m+1$. For any 
$\boldsymbol{\lambda}$, no changepoints occur in time $\{1, \ldots, p 
\}$ when $N$ is large.

Suppose that the true relative changepoint configuration is 
$\boldsymbol{\lambda}^0 = (\lambda_1^0, \ldots, \lambda_{m^0}^0)'$, 
where true parameter values are superscripted with zero. Our goal is to 
identify $\boldsymbol{\lambda}^0$ over many candidate models. In fact, 
for a (fixed) large integer $M$, all relative changepoint configurations 
in
\[
\boldsymbol{\Lambda} = \{ 
\boldsymbol{\lambda}: 0 \leq m \leq M, 
\min_{r = 1, 2, \ldots, m + 1} \lambda_r - \lambda_{r-1} \geq d \}
\]
are considered, where $d$ is a small positive constant, smaller than 
$\lambda_r^0 - \lambda_{r-1}^0$ for all $r =1, \ldots, m^0 + 1$. We 
assume that $m^0 \leq M$ such that $\boldsymbol{\lambda}^0 \in 
\boldsymbol{\Lambda}$ and $M \leq 1/d$.

Under the same assumptions, the automatic MDL for piece-wise AR 
processes \citep{Davis_etal_2006} has been shown to consistently 
estimate relative changepoint locations and model parameters 
\citep{Davis_Yau_2013}. The following two theorems show that the BMDL 
\eqref{eq:BMDL_univariate} also achieve the same large sample 
consistency.

\begin{theorem}[Consistency of changepoint configuration]\label{thm:lambda_convergence}
Given the observed time series of length $N$, denote the estimated 
relative changepoint model as
\begin{equation}
\label{eq:lambda_hat}
\hat{\boldsymbol\lambda}_N = \arg \min_{\boldsymbol\lambda \in \boldsymbol\Lambda}
~ \text{BMDL}(\boldsymbol\lambda),
\end{equation}
with $\hat{m}_N = |\hat{\boldsymbol\lambda}_N|$ changepoints. 
Then as $N \rightarrow \infty$,
\begin{equation}\label{eq:lambda_convergence}
\hat{m}_N \stackrel{P}{\longrightarrow} m^0 \quad \text{and} \quad 
\hat{\boldsymbol\lambda}_N \stackrel{P}{\longrightarrow} \boldsymbol\lambda^0.
\end{equation}
Furthermore, the convergence rate for each $r = 1, \ldots, m^0$ is
\begin{equation}
\label{eq:lambda_convergence_rate}
\left| \hat{\lambda}_r - \lambda^0_r \right| = O_P\left(\frac{1}{N}\right).
\end{equation}
\end{theorem}

\begin{theorem}[Consistency of parameter estimation]
\label{thm:parameter_convergence} 
Suppose that under the true model $\boldsymbol\lambda^0$, the true model 
parameters are $\boldsymbol\mu^0, \mathbf{s}^0, (\sigma^{2})^0$, and 
$\boldsymbol\phi^0$. Under the estimated relative changepoint model 
$\hat{\boldsymbol\lambda}_N$ in \eqref{eq:lambda_hat}, the BMDL 
estimator for $\boldsymbol\phi$, denoted by $\hat{\boldsymbol\phi}_N$, 
is given by the Yule-Walker estimator described in Section 
\ref{subsection:bmdl_univariate}; the BMDL estimator for $\mathbf{s}$ 
and $\sigma^2$, denoted by $\hat{\mathbf{s}}_N$ and 
$\hat{\sigma}^{2}_N$, are given by \eqref{eq:s_hat} and 
\eqref{eq:sigmasq_hat} after replacing all terms containing 
$\boldsymbol\phi$ by $\hat{\boldsymbol\phi}_N$, respectively; the BMDL 
estimator for $\boldsymbol\mu$ is taken as its conditional posterior 
mean
\begin{equation}\label{eq:mu_hat}
\hat{\boldsymbol\mu}_N = 
E\left(\boldsymbol\mu \mid \hat{\mathbf{s}}_N, \hat{\sigma}^{2}_N, \hat{\boldsymbol\lambda}_N, 
\hat{\boldsymbol\lambda}_N, \mathbf{X}_{1:N}\right)
= \left(\widehat{\mathbf{D}}' 
\widehat{\mathbf{D}}+\frac{\mathbf{I}_{m}}{\nu}\right)^{-1}
\widehat{\mathbf{D}}' \left( \widehat{\mathbf{X}} - \widehat{\mathbf{A}}\hat{\mathbf{s}}_N \right).
\end{equation}
Then as $N\rightarrow \infty$, all estimators converge to their true values
in probability, i.e., 
\begin{equation}\label{eq:parameter_convergence}
\hat{\boldsymbol\mu}_N \stackrel{P}{\longrightarrow} \boldsymbol\mu^0, \quad 
\hat{\mathbf{s}}_N \stackrel{P}{\longrightarrow} \mathbf{s}^0, \quad 
\hat{\sigma}^{2}_N \stackrel{P}{\longrightarrow} (\sigma^{2})^0, \quad 
\hat{\boldsymbol\phi}_N \stackrel{P}{\longrightarrow} \boldsymbol\phi^0.
\end{equation}
\end{theorem}

Proofs of Theorem \ref{thm:lambda_convergence} and 
\ref{thm:parameter_convergence} are given in the Appendix Section 
\ref{PROOFthm:lambda_convergence} and 
\ref{PROOFthm:parameter_convergence}, respectively. The convergence rate 
$O_P(1/N)$ in \eqref{eq:lambda_convergence_rate} is viewed as the 
optimal rate in the multiple changepoint detection literature 
\citep{Niu_etal_2016}. From a Bayesian model selection perspective, a 
model selection criterion is consistent if the ratio of posterior 
probabilities between the true model $\boldsymbol{\lambda}^0$ and any 
other model $\boldsymbol{\lambda} \in \boldsymbol{\Lambda}$ tends to 
infinity \citep{Clyde_George_2004}. This is equivalent to the BMDL 
difference $\text{BMDL}\left(\boldsymbol\lambda\right) 
-\text{BMDL}\left(\boldsymbol\lambda^0\right) \longrightarrow \infty$, 
which is shown to hold in Proposition \ref{prop:pairwise_BMDL_OpN} and 
\ref{prop:pairwise_BMDL_OplogN} in the Appendix.
 
To better understand our BMDL penalty, we compare it to the MDL 
\eqref{eq:MDL_univariate}.
Under a given 
relative changepoint model $\boldsymbol\lambda$, 
\eqref{eq:MDL_univariate} increases linearly with $N$. The following 
theorem states that the difference between the BMDL in 
\eqref{eq:BMDL_univariate} and the automatic MDL in 
\eqref{eq:MDL_univariate} is asymptotically bounded.

\begin{theorem}
\label{prop:BMDL_MDL}
For any relative changepoint model 
$\boldsymbol{\lambda} \in \boldsymbol{\Lambda}$, as $N \rightarrow \infty$, 
up to an additive constant,
\[
\text{BMDL}(\boldsymbol\lambda) -  \text{MDL}(\boldsymbol\lambda) = O_P(1).
\]
\end{theorem}

A proof of Theorem \ref{prop:BMDL_MDL} is obtained by comparing the 
large sample performance of the corresponding terms in 
\eqref{eq:BMDL_univariate} and \eqref{eq:MDL_univariate} via order 
estimates derived in the Appendix. In the BMDL expression 
\eqref{eq:BMDL_univariate}, all but the last term arise from the mixture 
MDL. The term $(N-p)\log \left( \hat{\sigma}^2 \right)/2$ measures the 
model's goodness-of-fit. By Lemma \ref{lemma:sigmasq1} in the Appendix, 
$\hat{\sigma}^2 =\hat{\sigma}^2_{\nu = \infty} + O_P(1/N)$; hence, the 
difference between the first terms in \eqref{eq:BMDL_univariate} and 
\eqref{eq:MDL_univariate} obeys
\[
\frac{N-p}{2}\log \left( \hat{\sigma}^2 \right) - \frac{N-p}{2}\log 
\left( \hat{\sigma}^2_{\nu = \infty} \right) = O_P(1).
\]
In \eqref{eq:BMDL_univariate}, the second term is $O_P(1)$, while
the third term, by Lemma \ref{lemma:det} in the Appendix, satisfies
\[
\frac{1}{2}
\log\left( \left|\widehat{\mathbf{D}}' 
	\widehat{\mathbf{D}}+\frac{\mathbf{I}_{m}}{\nu}\right| \right)
= \frac{1}{2}\sum_{r=2}^{m+1}\log(N_r) + O_P\left(1\right),
\]
which interestingly suggests that the mixture MDL in 
\eqref{eq:BMDL_univariate} contains a built in penalty on 
$\boldsymbol\mu$ that performs similarly to the two-part MDL penalty on 
$\boldsymbol\mu$ in \eqref{eq:MDL_univariate}. The last term in 
\eqref{eq:BMDL_univariate} is the penalty on the changepoint 
configuration $\boldsymbol{\lambda}$. With or without metadata, Lemma 
\ref{lemma:Gamma_pairdiff} in the Appendix suggests that this term is 
asymptotically $m\log(N) + O_P(1)$, which only differs from the last 
term in \eqref{eq:MDL_univariate} by $O_P(1)$ plus a constant.

An implication of Theorem \ref{prop:BMDL_MDL} is that the model 
selection consistency results in Theorem \ref{thm:lambda_convergence} 
also hold for the automatic MDL \eqref{eq:MDL_univariate}, which gives 
alternate confirmation of the asymptotic results in 
\citet{Davis_etal_2006} and \citet{Davis_Yau_2013}. In addition, without 
metadata, the BMDL \eqref{eq:BMDL_univariate} and the automatic MDL 
\eqref{eq:MDL_univariate} perform similarly for large samples. Section 
\ref{sec:simulation} confirms this result via simulation examples, also 
demonstrating that when metadata is available and incorporated, the BMDL 
significantly increases changepoint detection power and precision
under finite samples.

\section{Discussion}\label{sec:discussion} 

This paper developed a flexible MDL-based multiple changepoint detection 
approach to accommodate {\it a priori} information on changepoint times via 
prior distributional specifications. Motivated by climate homogenization 
problems, our Bayesian MDL (BMDL) method incorporates subjective 
knowledge such as metadata in mean shift detection for univariate 
autoregressive processes with seasonal means, and then extended these 
ideas to bivariate VAR settings while encouraging concurrent changes in 
the component series. Both theoretical and simulation studies show that 
without metadata, our BMDL performs similarly to the state-of-art 
automatic MDL method; with metadata, the BMDL's detection power 
significantly improves under finite samples. Our BMDL has several 
practical advantages, including simple parameter elicitation, asymptotic 
consistency, and efficient MCMC computation.

The approach can be extended to accommodate more flexible time series 
structures, including periodic autoregressions \citep{Hewa_etal_2017}, 
moving-averages, and multivariate data with more than two series.  The 
methods could also be tailored to categorical data.  For count data, the 
likelihood could be Poisson-based.  With a conjugate Gamma prior on 
means, the resulting marginal likelihoods will again have closed forms. 
There is no technical difficulty in allowing a background linear trend, 
or even piecewise linear trends. This said, linear trends can be 
mistaken for multiple mean shifts should trends be present and ignored 
in the analysis \citep{Li_Lund_2015}. In addition, with straightforward 
modification, the BMDL can handle changes in variances or autocovariances.

Non-MCMC stochastic search methods could also be used. Genetic 
algorithms, popular in multiple changepoint MDL analyses, are also 
capable of minimizing the BMDL. Pre-screening methods such as 
\citet{Chan_etal_2014, Yau_Zhao_2016} can speed up model search 
algorithms. In simple settings when no global parameters exist (i.e., 
independent observations, no seasonal cycle, error variance known), 
dynamic programming based techniques such as the PELT 
\citep{Killick_etal_2012} can further accelerate computational speed.

\subsection*{Supplementary Materials}

{\bf Appendix}: includes more theoretical results and theorem proofs in 
Section \ref{sec:appendix_proofs}, and additional simulation and data 
examples in Section \ref{sec:appendix_examples}.


\subsection*{Acknowledgement}

The authors thank Matthew Menne, Jared Rennie, Claude Williams Jr., and 
Bin Yu for helpful discussions. The climate application was posed at 
SAMSI's 2014 climate homogeneity summit in Boulder, Colorado.  Robert 
Lund and Anuradha Hewaarachchi thank NSF Grant DMS 1407480 for partial 
support. Clemson University is acknowledged for generous allotment of 
computation time on its Palmetto cluster.  The comments of two referees 
substantially improved this manuscript.

{\small
\setstretch{0.85}
\bibliographystyle{asa}
\bibliography{ChangepointTxTn}

\begin{thebibliography}{58}
\newcommand{\enquote}[1]{``#1''}
\expandafter\ifx\csname natexlab\endcsname\relax\def\natexlab#1{#1}\fi

\bibitem[{Aue and Horv\'{a}th(2013)}]{Aue_Horvath_2013}
Aue, A. and Horv\'{a}th, L. (2013), \enquote{Structural Breaks in Time Series,}
  \textit{Journal of Time Series Analysis}, 34, 1--16.

\bibitem[{Bardwell and Fearnhead(2017)}]{Bardwell_Fearnhead_2017}
Bardwell, L. and Fearnhead, P. (2017), \enquote{Bayesian Detection of Abnormal
  Segments in Multiple Time Series,} \textit{Bayesian Analysis}.

\bibitem[{Barry and Hartigan(1993)}]{Barry_Hartigan_1993}
Barry, D. and Hartigan, J.~A. (1993), \enquote{A {B}ayesian Analysis for Change
  Point Problems,} \textit{Journal of the American Statistical Association},
  88, 309--319.

\bibitem[{Billingsley(1995)}]{Billingsley_1995}
Billingsley, P. (1995), \textit{Probability and Measure}, John Wiley \& Sons,
  3rd ed.

\bibitem[{Brockwell and Davis(1991)}]{Brockwell_Davis_1991}
Brockwell, P.~J. and Davis, R.~A. (1991), \textit{Time Series: Theory and
  Methods}, Springer-Verlag, 2nd ed.

\bibitem[{Carlin and Louis(2000)}]{Carlin_Louis_2000}
Carlin, B.~P. and Louis, T.~A. (2000), \textit{Bayes and Empirical {B}ayes
  Methods for Data Analysis}, Chapman \& Hall/CRC Boca Raton.

\bibitem[{Caussinus and Mestre(2004)}]{Caussinus_Mestre_2004}
Caussinus, H. and Mestre, O. (2004), \enquote{Detection and Correction of
  Artificial Shifts in Climate Series,} \textit{Journal of the Royal
  Statistical Society: Series C (Applied Statistics)}, 53, 405--425.

\bibitem[{Chan et~al.(2014)Chan, Yau, and Zhang}]{Chan_etal_2014}
Chan, N.~H., Yau, C.~Y., and Zhang, R.-M. (2014), \enquote{Group {LASSO} for
  Structural Break Time Series,} \textit{Journal of the American Statistical
  Association}, 109, 590--599.

\bibitem[{Chernoff and Zacks(1964)}]{Chernoff_Zacks_1964}
Chernoff, H. and Zacks, S. (1964), \enquote{Estimating the Current Mean of a
  Normal Distribution which is Subjected to Changes in Time,} \textit{The
  Annals of Mathematical Statistics}, 35, 999--1018.

\bibitem[{Chib(1998)}]{Chib_1998}
Chib, S. (1998), \enquote{Estimation and Comparison of Multiple Change-point
  Models,} \textit{Journal of Econometrics}, 86, 221--241.

\bibitem[{Cho and Fryzlewicz(2015)}]{Cho_Fryzlewicz_2015}
Cho, H. and Fryzlewicz, P. (2015), \enquote{Multiple-change-point Detection for
  High Dimensional Time Series via Sparsified Binary Segamentation,}
  \textit{Journal of the Royal Statistical Society: Series B (Statistical
  Methodology)}, 77, 475--507.

\bibitem[{Christensen(2002)}]{Christensen_2002}
Christensen, R. (2002), \textit{Plane Answers to Complex Questions: The Theory
  of Linear Models}, Springer.

\bibitem[{Clyde and George(2004)}]{Clyde_George_2004}
Clyde, M.~A. and George, E.~I. (2004), \enquote{Model Uncertainty,}
  \textit{Statistical Science}, 19, 81--94.

\bibitem[{Davis et~al.(2006)Davis, Lee, and Rodriguez-Yam}]{Davis_etal_2006}
Davis, R.~A., Lee, T. C.~M., and Rodriguez-Yam, G.~A. (2006),
  \enquote{Structural Break Estimation for Nonstationary Time Series Models,}
  \textit{Journal of the American Statistical Association}, 101, 223--239.

\bibitem[{Davis et~al.(2008)Davis, Lee, and Rodriguez-Yam}]{Davis_etal_2008}
--- (2008), \enquote{Break Detection for a Class of Nonlinear Time Series
  Models,} \textit{Journal of Time Series Analysis}, 29, 834--867.

\bibitem[{Davis and Yau(2013)}]{Davis_Yau_2013}
Davis, R.~A. and Yau, C.~Y. (2013), \enquote{Consistency of Minimum Description
  Length Model Selection for Piecewise Stationary Time Series Models,}
  \textit{Electronic Journal of Statistics}, 7, 381--411.

\bibitem[{Du et~al.(2016)Du, Kao, and Kou}]{Du_etal_2016}
Du, C., Kao, C.-L.~M., and Kou, S.~C. (2016), \enquote{Stepwise Signal
  Extraction via Marginal Likelihood,} \textit{Journal of the American
  Statistical Association}, 111, 314--330.

\bibitem[{Fearnhead(2006)}]{Fearnhead_2006}
Fearnhead, P. (2006), \enquote{Exact and Efficient {B}ayesian Inference for
  Multiple Changepoint Problems,} \textit{Statistical Computing}, 16, 203--213.

\bibitem[{Fearnhead and Vasileiou(2009)}]{Fearnhead_Vasileiou_2009}
Fearnhead, P. and Vasileiou, D. (2009), \enquote{Bayesian Analysis of
  Isochores,} \textit{Journal of the American Statistical Association}, 104,
  132--141.

\bibitem[{Fryzlewicz(2014)}]{Fryzlewicz_2014}
Fryzlewicz, P. (2014), \enquote{Wild Binary Segmentation for Multiple
  Change-Point Detection,} \textit{Annals of Statistics}, 42, 2243--2281.

\bibitem[{Fryzlewicz and Subba~Rao(2014)}]{Fryzlewicz_SubbaRao_2014}
Fryzlewicz, P. and Subba~Rao, S. (2014), \enquote{Multiple-Change-Point
  Detection for Auto-Regressive Conditional Heteroscedastic Processes,}
  \textit{Journal of the Royal Statistical Society: Series B (Statistical
  Methodology)}, 76, 903--924.

\bibitem[{Garc{\'\i}a-Donato and
  Mart{\'\i}nez-Beneito(2013)}]{Garcia-Donato_Martinez-Beneito_2013}
Garc{\'\i}a-Donato, G. and Mart{\'\i}nez-Beneito, M.~A. (2013), \enquote{On
  Sampling Strategies in {B}ayesian Variable Selection Problems with Large
  Model Spaces,} \textit{Journal of the American Statistical Association}, 108,
  340--352.

\bibitem[{George and McCulloch(1997)}]{George_McCulloch_1997}
George, E.~I. and McCulloch, R.~E. (1997), \enquote{Approaches for {B}ayesian
  Variable Selection,} \textit{Statistics Sinica}, 7, 339--373.

\bibitem[{Giordani and Kohn(2008)}]{Giordani_Kohn_2008}
Giordani, P. and Kohn, R. (2008), \enquote{Efficient {B}ayesian Inference for
  Multiple Change-Point and Mixture Innovation Models,} \textit{Journal of
  Business and Economic Statistics}, 26, 66--77.

\bibitem[{Gir{\'o}n et~al.(2007)Gir{\'o}n, Moreno, and
  Casella}]{Giron_etal_2007}
Gir{\'o}n, J., Moreno, E., and Casella, G. (2007), \enquote{Objective
  {B}ayesian Analysis of Multiple Changepoints for Linear Models,}
  \textit{Bayesian Statistics 8}.

\bibitem[{Green(1995)}]{Green_1995}
Green, Peter, J. (1995), \enquote{Reversible Jump {M}arkov Chain {M}onte
  {C}arlo Computation and {B}ayesian Model Determination,} \textit{Biometrika},
  82, 711--732.

\bibitem[{Gr{\"u}nwald(2007)}]{Grunwald_2007}
Gr{\"u}nwald, P.~D. (2007), \textit{The Minimum Description Length Principle},
  The MIT Press.

\bibitem[{Hannart and Naveau(2012)}]{Hannart_Naveau_2012}
Hannart, A. and Naveau, P. (2012), \enquote{An Improved {B}ayesian Information
  Criterion for Multiple Change-point Models,} \textit{Technometrics}, 54,
  256--268.

\bibitem[{Hansen and Yu(2001)}]{Hansen_Yu_2001}
Hansen, M.~H. and Yu, B. (2001), \enquote{Model Selection and the Principle of
  Minimum Description Length,} \textit{Journal of the American Statistical
  Association}, 96, 746--774.

\bibitem[{Harville(2008)}]{Harville_2008}
Harville, D.~A. (2008), \textit{Matrix Algebra From a Statistician's
  Perspective}, Springer-Verlag.

\bibitem[{Hewaarachchi et~al.(2017)Hewaarachchi, Li, Lund, and
  Rennie}]{Hewa_etal_2017}
Hewaarachchi, A., Li, Y., Lund, R., and Rennie, J. (2017),
  \enquote{Homogenization of Daily Temperature Data,} \textit{Journal of
  Climate}, 30, 985--999.

\bibitem[{Kass and Raftery(1995)}]{Kass_Raftery_1995}
Kass, R.~E. and Raftery, A.~E. (1995), \enquote{Bayes Factors,} \textit{Journal
  of the American Statistical Association}, 90, 773--795.

\bibitem[{Killick et~al.(2012)Killick, Fearnhead, and
  Eckley}]{Killick_etal_2012}
Killick, R., Fearnhead, P., and Eckley, I.~A. (2012), \enquote{Optimal
  Detection of Changepoints With a Linear Computational Cost,} \textit{Journal
  of the American Statistical Association}, 107, 1590--1598.

\bibitem[{Kirch et~al.(2015)Kirch, Muhsal, and Ombao}]{Kirch_etal_2015}
Kirch, C., Muhsal, B., and Ombao, H. (2015), \enquote{Detection of Changes in
  Multivariate Time Series with Application to {EEG} Data,} \textit{Journal of
  the American Statistical Association}, 110, 1197--1216.

\bibitem[{Li and Zhang(2010)}]{Li_Zhang_2010}
Li, F. and Zhang, N.~R. (2010), \enquote{Bayesian Variable Selection in
  Structured High-Dimensional Covariate Spaces with Applications in Genomics,}
  \textit{Journal of the American Statistical Association}, 105, 1202--1214.

\bibitem[{Li and Lund(2012)}]{Li_Lund_2012}
Li, S. and Lund, R. (2012), \enquote{Multiple Changepoint Detection via Genetic
  Algorithms,} \textit{Journal of Climate}, 25, 674--686.

\bibitem[{Li and Lund(2015)}]{Li_Lund_2015}
Li, Y. and Lund, R. (2015), \enquote{Multiple Changepoint Detection Using
  Metadata,} \textit{Journal of Climate}, 28, 4199--4216.

\bibitem[{Liu et~al.(2016)Liu, Shao, Lund, and Woody}]{Liu_etal_2016}
Liu, G., Shao, Q., Lund, R., and Woody, J. (2016), \enquote{Testing for
  Seasonal Means in Time Series Data,} \textit{Environmetrics}, 27, 198--211.

\bibitem[{Lu et~al.(2010)Lu, Lund, and Lee}]{Lu_etal_2010}
Lu, Q., Lund, R., and Lee, T. C.~M. (2010), \enquote{An {MDL} Approach to the
  Climate Segmentation Problem,} \textit{The Annals of Applied Statistics}, 4,
  299--319.

\bibitem[{Lund et~al.(2007)Lund, Wang, Reeves, Lu, Gallagher, and
  Feng}]{Lund_etal_2007}
Lund, R., Wang, X., Reeves, J., Lu, Q., Gallagher, C., and Feng, Y. (2007),
  \enquote{Changepoint Detection in Periodic and Autocorrelated Time Series,}
  \textit{Journal of Climate}, 20, 5178--5190.

\bibitem[{Ma and Yau(2016)}]{Ma_Yau_2016}
Ma, T.~F. and Yau, C.~Y. (2016), \enquote{A Pairwise Likelihood-based Approach
  for Changepoint Detection in Multivariate Time Series Models,}
  \textit{Biometrika}, 103, 409--421.

\bibitem[{Menne and Williams~Jr(2005)}]{Menne_WilliamsJr_2005}
Menne, M.~J. and Williams~Jr, C.~N. (2005), \enquote{Detection of Undocumented
  Changepoints Using Multiple Test Statistics and Composite Reference Series,}
  \textit{Journal of Climate}, 18, 4271--4286.

\bibitem[{Mitchell(1953)}]{Mitchell_1953}
Mitchell, J.~M. (1953), \enquote{On the Causes of Instrumentally Observed
  Secular Temperature Trends,} \textit{Journal of Meteorology}, 10, 244--261.

\bibitem[{Niu et~al.(2016)Niu, Hao, and Zhang}]{Niu_etal_2016}
Niu, Y.~S., Hao, N., and Zhang, H. (2016), \enquote{Multiple Change-Point
  Detection: A Selective Overview,} \textit{Statistical Science}, 31, 611--623.

\bibitem[{Pein et~al.(2017)Pein, Sieling, and Munk}]{Pein_etal_2017}
Pein, F., Sieling, H., and Munk, A. (2017), \enquote{Heterogeneous Change Point
  Inference,} \textit{Journal of the Royal Statistical Society: Series B
  (Statistical Methodology)}, 79, 1207--1227.

\bibitem[{Preuss et~al.(2015)Preuss, Puchstein, and Dette}]{Preuss_etal_2015}
Preuss, P., Puchstein, R., and Dette, H. (2015), \enquote{Detection of Multiple
  Structural Breaks in Multivariate Time Series,} \textit{Journal of the
  American Statistical Association}, 110, 654--668.

\bibitem[{Risanen(1989)}]{Risanen_1989}
Risanen, J. (1989), \textit{Stochastic Complexity in Statistical Inquiry}, vol.
  511, World Scientific, Singapore.

\bibitem[{Schwarz(1978)}]{Schwarz_1978}
Schwarz, G. (1978), \enquote{Estimating the Dimension of a Model,} \textit{The
  Annals of Statistics}, 6, 461--464.

\bibitem[{Scott and Berger(2010)}]{Scott_Berger_2010}
Scott, J. and Berger, J. (2010), \enquote{{B}ayes and Empirical-{B}ayes
  Multiplicity Adjustment in the Variable-selection Problem,} \textit{The
  Annals of Statistics}, 38, 2587--2619.

\bibitem[{Shannon(1948)}]{Shannon_1948}
Shannon, C.~E. (1948), \enquote{A Mathematical Theory of Communication,}
  \textit{Bell System Technical Journal}, 27, 623.

\bibitem[{Shao and Zhang(2010)}]{Shao_Zhang_2010}
Shao, X. and Zhang, X. (2010), \enquote{Testing for Change Points in Time
  Series,} \textit{Journal of the American Statistical Association}, 105,
  1228--1240.

\bibitem[{Wilks(2011)}]{Wilks_2011}
Wilks, D.~S. (2011), \textit{Statistical Methods in the Atmospheric Sciences},
  Academic Press.

\bibitem[{Yao(1984)}]{Yao_1984}
Yao, Y.-C. (1984), \enquote{Estimation of a Noisy Discrete-Time Step Function:
  {B}ayes and Empirical {B}ayes Approaches,} \textit{The Annals of Statistics},
  12, 1434--1447.

\bibitem[{Yao(1988)}]{Yao_1988}
--- (1988), \enquote{Estimating the Number of Change-Points via {S}chwarz'
  Criterion,} \textit{Statistics \& Probability Letters}, 6, 181--189.

\bibitem[{Yau et~al.(2015)Yau, Tang, and Lee}]{Yau_etal_2015}
Yau, C.~Y., Tang, C.~M., and Lee, T. C.~M. (2015), \enquote{Estimation of
  Multiple-Regime Threshold Autoregressive Models with Structural Breaks,}
  \textit{Journal of the American Statistical Association}, 110, 1175--1186.

\bibitem[{Yau and Zhao(2016)}]{Yau_Zhao_2016}
Yau, C.~Y. and Zhao, Z. (2016), \enquote{Inference for Multiple Change Points
  in Time Series via Likelihood Ratio Scan Statistics,} \textit{Journal of the
  Royal Statistical Society: Series B (Statistical Methodology)}, 78, 895--916.

\bibitem[{Zhang and Siegmund(2007)}]{Zhang_Siegmund_2007}
Zhang, N.~R. and Siegmund, D.~O. (2007), \enquote{A Modified {B}ayes
  Information Criterion with Applications to the Analysis of Comparative
  Genomic Hybridization Data,} \textit{Biometrics}, 63, 22--32.

\bibitem[{Zhang and Siegmund(2012)}]{Zhang_Siegmund_2012}
--- (2012), \enquote{Model Selection for High-Dimensional, Multi-Sequence
  Change-Point Problems,} \textit{Statistica Sinica}, 1507--1538.

\end{thebibliography}
}

\newpage 

\begin{center}
{\Large
Appendix for ``Multiple Changepoint Detection with Partial Information on Changepoint Times''
}
\end{center}

\renewcommand{\thesection}{\Alph{section}}
\setcounter{section}{0}
\section{Theoretical Results and Proofs}\label{sec:appendix_proofs}

In this Appendix, the asymptotic limits of the Yule-Walker estimator 
$\hat{\boldsymbol\phi}$ and white noise variance $\hat{\sigma}^2$ under 
a given changepoint model $\boldsymbol\lambda$ are investigated in 
Sections \ref{subsection:YW_phi} and \ref{subsection:YW_sigmasq}, 
respectively.  In Section \ref{sec:asym_pairwise_BMDL}, the BMDL 
difference between the true model $\boldsymbol\lambda^0$ and other 
models is studied, showing that $\boldsymbol\lambda^0$ achieves the 
smallest BMDL in the limit. Last, the proofs of Theorem 
\ref{thm:lambda_convergence} and Theorem \ref{thm:parameter_convergence} 
are given in Sections \ref{PROOFthm:lambda_convergence} and 
\ref{PROOFthm:parameter_convergence}, respectively.


\subsection{Asymptotic behavior of the Yule-Walker estimator of the 
autoregression coefficients $\hat{\boldsymbol\phi}$} 
\label{subsection:YW_phi} For a sample size $N$, the observations obey 
the true changepoint model $\boldsymbol\lambda^0$ in 
\eqref{eq:likelihood3}: 
\[ 
\mathbf{X}= \mathbf{A} \mathbf{s} + 
\mathbf{D}^0\boldsymbol\mu^0 + \boldsymbol\epsilon. 
\] 
Here, $\boldsymbol\epsilon$ is a zero-mean causal AR$(p)$ series. When 
there is no ambiguity, we simplify the notations $\boldsymbol\mu^0, 
\mathbf{s}^0, (\sigma^{2})^0, \boldsymbol\phi^0$ to $\boldsymbol\mu, 
\mathbf{s}, \sigma^2, \boldsymbol\phi$, respectively, and omit 
subscripts such as $1:N$ on the data vector and other quantities.

For any relative changepoint model $\boldsymbol\lambda$, suppose that 
$\boldsymbol\eta$ is the corresponding changepoint configuration under 
the sample size $N$. From \eqref{eq:Y}, the ordinary least squares 
residual vector is
\begin{equation}
\label{eq:Y2}
\boldsymbol\epsilon^{\text{ols}}  =  
	(\mathbf{I}_N - \mathcal{P}_{[\mathbf{A}|\mathbf{D}]})\mathbf{X} 
	       = (\mathbf{I}_N - \mathcal{P}_{[\mathbf{A}|\mathbf{D}]})
	(\mathbf{A} \mathbf{s} + 
	\mathbf{D}^0\boldsymbol\mu+ \boldsymbol\epsilon)
		=(\mathbf{I}_N - \mathcal{P}_{[\mathbf{A}|\mathbf{D}]})
	(\mathbf{D}^0\boldsymbol\mu + \boldsymbol\epsilon). 
\end{equation} 
Here, $[\mathbf{A}|\mathbf{D}]$ is the block matrix formed by 
$\mathbf{A}$ and $\mathbf{D}$and $\mathcal{P}_{\mathbf{A}}$ is the 
orthogonal projection onto the columns of the matrix $\mathbf{A}$. The 
regime indicator matrix $\mathbf{D}$ depends on $\boldsymbol\lambda$ and 
may not equal $\mathbf{D}^0$.

\begin{lemma}\label{lemma:Y}
For each relative changepoint configuration $\boldsymbol\lambda \in 
\boldsymbol\Lambda$ and $t \in \{ 1, \ldots, N \}$, when $N$ is large, 
each entry of $\boldsymbol\epsilon^{\text{ols}}$ can be expressed as
\begin{equation}
\label{eq:delta_t_W_t}
\epsilon_t^{\text{ols}}=\delta_t + W_t, \quad \text{where} \quad
\delta_t = \mu_{r^0(t)} - \bar{\mu}_{r(t)}
\quad \mbox{and} \quad
W_t = \epsilon_t - \bar{\epsilon}_{r(t)} - \bar{\epsilon}_{v(t)} + \bar{\epsilon}.
\end{equation}
Here, the functions $r^0(t)$ and $r(t)$ are the regimes that time 
$t$ is in under the models $\boldsymbol\lambda^0$ and 
$\boldsymbol\lambda$, respectively. In regime $\ell$ of the changepoint 
configuration $\boldsymbol\lambda$, $\bar{\mu}_{\ell} = N_{\ell}^{-1} 
\sum_{t \in {\cal R}_\ell} \mu_t$ is the average of the true mean 
parameters, $N_\ell$ is the number of time points in this regime, and 
${\cal R}_\ell$ is the set of all time points in this regime.  Likewise, 
$\bar{\epsilon}_\ell$ is the average of errors in regime $\ell$, 
$\bar{\epsilon}_{v}$ is the average of errors during season $v$, and 
$\bar{\epsilon}$ is the average of all errors. \end{lemma}

\begin{proof} Because of \eqref{eq:Y2}, our main objective is to study 
the projection residual $\mathbf{I}_N - 
\mathcal{P}_{[\mathbf{A}|\mathbf{D}]}$ under large $N$.  Since the two 
column spaces spanned by $(\mathbf{I}_N - \mathcal{P}_{\mathbf{D}}) 
\mathbf{A}$ and $\mathbf{D}$ are perpendicular, Theorem B.45 in 
\citet[pp.\ 411]{Christensen_2002} gives $\mathcal{P}_{[(\mathbf{I}_N - 
\mathcal{P}_{\mathbf{D}}) \mathbf{A}|\mathbf{D}]} = 
\mathcal{P}_{(\mathbf{I}_N - \mathcal{P}_{\mathbf{D}}) \mathbf{A}} + 
\mathcal{P}_{\mathbf{D}}$.  Projection properties give
\begin{equation}
\label{eq:Q_AD}
\mathbf{I}_N - \mathcal{P}_{[\mathbf{A}|\mathbf{D}]} 
= \mathbf{I}_N - \mathcal{P}_{[(\mathbf{I}_N - \mathcal{P}_{\mathbf{D}})\mathbf{A}|\mathbf{D}]}
= \mathbf{I}_N - \mathcal{P}_{(\mathbf{I}_N - \mathcal{P}_{\mathbf{D}}) \mathbf{A}}
- \mathcal{P}_{\mathbf{D}}.
\end{equation}
The term $\mathcal{P}_{(\mathbf{I}_N - \mathcal{P}_{\mathbf{D}}) 
\mathbf{A}}$ can be expanded as
\begin{equation}
\label{eq:P_QD_A}
\mathcal{P}_{(\mathbf{I}_N - \mathcal{P}_{\mathbf{D}}) \mathbf{A}}
= 	(\mathbf{I}_N - \mathcal{P}_{\mathbf{D}}) \mathbf{A}
	\left\{\mathbf{A}' (\mathbf{I}_N - \mathcal{P}_{\mathbf{D}}) \mathbf{A}\right\}^{-1}
	\mathbf{A}' (\mathbf{I}_N - \mathcal{P}_{\mathbf{D}}).
\end{equation}
For any $n \in \mathbb{N}$, let $\mathbf{0}_n$ be the $n$-dimensional 
vector containing all zero entries, $\mathbf{1}_n$ be the 
$n$-dimensional vector whose entries are all unity, and $\mathbf{J}_n$ 
be the $n \times n$ matrix whose entries are all unity, i.e., 
$\mathbf{J}_n = \mathbf{1}_n \mathbf{1}_n'$.

For $v \in \{ 1, \ldots, T \}$, suppose there are $k(v,\ell)$ time 
points in regime $\ell$ that are also in season $v$. Since $N_\ell$ 
increases linearly with $N$, so does $k(v,\ell)$. Moreover, when $N$ is 
large, inside each regime, the seasonal counts $k(v,\ell)$ are equal 
except for edge effects, i.e., $k(v,\ell)/N_\ell \approx 1/T$ for all 
seasons $v$.  To avoid trite work, we will ignore these edge effects in 
the ensuing calculations.  Proceeding under this simplification, the 
$v$th column in $\mathbf{A}$, denoted by $\mathbf{A}_v$, under the 
projection $\mathcal{P}_{\mathbf{D}}$, becomes
\begin{equation}
\label{eq:PD_Av}
\mathcal{P}_{\mathbf{D}}\mathbf{A}_v
= \left(\mathbf{0}_{N_1}', \frac{k(v, 2)}{N_2}\mathbf{1}_{N_2}', \ldots,
\frac{k(v, m+1)}{N_{m+1}}\mathbf{1}_{N_{m+1}}' \right)'
= \left(\mathbf{0}_{N_1}', \frac{1}{T}\mathbf{1}_{N - N_1}' \right)'.
\end{equation}

We can now obtain an expression for $\mathbf{A}' (\mathbf{I}_N - 
\mathcal{P}_{\mathbf{D}}) \mathbf{A}$. To do this, note that for $u,w 
\in \{1, 2, \ldots, T\}$,
\[
[\mathbf{A}' (\mathbf{I}_N - \mathcal{P}_{\mathbf{D}}) \mathbf{A}]_{u,w}
= 	\mathbf{A}_u' \mathbf{A}_w 
	-(\mathcal{P}_{\mathbf{D}}\mathbf{A}_u)'
	(\mathcal{P}_{\mathbf{D}}\mathbf{A}_w)
=	 \begin{cases}
 	\frac{N}{T^2}(T - (1-\lambda_1)),	& \text{if } u = w,\\
 	-\frac{N}{T^2}(1-\lambda_1),		& \text{if } u \neq w,
	\end{cases}
\]
and it follows that $\mathbf{A}' (\mathbf{I}_N - 
\mathcal{P}_{\mathbf{D}}) \mathbf{A} = NT^{-2}\{T \mathbf{I}_T - (1 - 
\lambda_1)\mathbf{J}_T\}$. The inverse of this matrix can be verified as
\[
\left\{\mathbf{A}' (\mathbf{I}_N - \mathcal{P}_{\mathbf{D}}) \mathbf{A}\right\}^{-1}
= \frac{1}{N}\left(T \mathbf{I}_T + 
\frac{1-\lambda_1}{\lambda_1}\mathbf{J}_T\right).
\]
Plugging this inverse into \eqref{eq:P_QD_A} and denoting 
$\mathcal{Q}_{\mathbf{D}} = \mathbf{I}_N - \mathcal{P}_{\mathbf{D}}$ 
produce
\begin{align}
\label{eq:P_QD_A2}
\mathcal{P}_{( \mathbf{I}_N - \mathcal{P}_{\mathbf{D}}) \mathbf{A}}
&=	 \frac{1}{N} (\mathcal{Q}_{\mathbf{D}} \mathbf{A})
	\left(T \mathbf{I}_T + 
	\frac{1-\lambda_1}{\lambda_1}\mathbf{J}_T\right)
	(\mathcal{Q}_{\mathbf{D}} \mathbf{A})'\\ \nonumber
&=	\frac{T}{N} (\mathcal{Q}_{\mathbf{D}} \mathbf{A})
	(\mathcal{Q}_{\mathbf{D}} \mathbf{A})'
+       \frac{1-\lambda_1}{N\lambda_1} 
	(\mathcal{Q}_{\mathbf{D}} \mathbf{A} \mathbf{1}_T)
	(\mathcal{Q}_{\mathbf{D}} \mathbf{A} \mathbf{1}_T)'.
\end{align}
For simplicity, we assume that regime $\ell$ starts with season one, 
ends with season $T$, and contains $n_\ell$ full cycles. Using $n = N/T= 
\sum_{r=1}^{m+1} n_r$ and \eqref{eq:PD_Av} gives
\begin{equation*}
\mathcal{Q}_{\mathbf{D}} \mathbf{A} = \left(
	\begin{array}{c}
	\mathbf{1}_{n_1} \otimes \mathbf{I}_T \\ \hdashline
	\mathbf{1}_{n - n_1} \otimes \left(\mathbf{I}_T - \frac{1}{T}\mathbf{J}_T\right)  
	\end{array}
\right), \quad
\mathcal{Q}_{\mathbf{D}}\mathbf{A}\mathbf{1}_T
= \left(
	\begin{array}{c}
	\mathbf{1}_{N_1} \\ \hdashline
	\mathbf{0}_{N - N_1} 
	\end{array}
\right).
\end{equation*}
Hence, quadratic forms of these matrices are
\[
	(\mathcal{Q}_{\mathbf{D}} \mathbf{A})
	(\mathcal{Q}_{\mathbf{D}} \mathbf{A})'
=	\left(
	\begin{array}{c:c}
	\mathbf{J}_{n_1} \otimes \mathbf{I}_T
	& \mathbf{J}_{n_1\times (n - n_1)} 
	\otimes \left(\mathbf{I}_T - \frac{1}{T}\mathbf{J}_T\right) \\ \hdashline
\mathbf{J}_{(n - n_1) \times n_1} 
	\otimes \left(\mathbf{I}_T - \frac{1}{T}\mathbf{J}_T\right) 
	& \mathbf{J}_{n - n_1} 
	\otimes \left(\mathbf{I}_T - \frac{1}{T}\mathbf{J}_T\right) 
	\end{array}
\right),
\]
and 
\[
	(\mathcal{Q}_{\mathbf{D}} \mathbf{A} \mathbf{1}_T)
	(\mathcal{Q}_{\mathbf{D}} \mathbf{A} \mathbf{1}_T)'
 = \left(
	\begin{array}{c:c}
	\mathbf{J}_{N_1} 	&	\mathbf{0}\\ \hdashline
	\mathbf{0}	&	\mathbf{0} 
	\end{array}
\right).
\]
Plugging these into \eqref{eq:P_QD_A2} produces
\[
\mathcal{P}_{(\mathbf{I}_N - \mathcal{P}_{\mathbf{D}}) \mathbf{A}}
= \frac{1}{N_1} \left(
	\begin{array}{c:c}
	\mathbf{J}_{N_1} 	&	\mathbf{0}\\ \hdashline
	\mathbf{0}	&	\mathbf{0} 
	\end{array}
\right) + \frac{T}{N} \mathbf{J}_{n} \otimes \mathbf{I}_T
-  \frac{1}{N}\mathbf{J}_{N}.
\]
Since $\mathcal{P}_{\mathbf{D}}$ is block-diagonal of form
\[
\mathcal{P}_{\mathbf{D}} = \text{diag}\left(\mathbf{0}_{N_1 \times N_1},
\frac{\mathbf{J}_{N_2}}{N_2}, \ldots, \frac{\mathbf{J}_{N_{m+1}}}{N_{m+1}}\right),
\] 
we have
\[
\mathbf{I}_N - \mathcal{P}_{[\mathbf{A}|\mathbf{D}]} 
= \mathbf{I}_N - \text{diag}\left(\frac{\mathbf{J}_{N_1}}{N_1}, 
	\frac{\mathbf{J}_{N_2}}{N_2}, \ldots, \frac{\mathbf{J}_{N_{m+1}}}{N_{m+1}}\right) 
	- \frac{T}{N}\mathbf{J}_{n} \otimes \mathbf{I}_T
	+  \frac{1}{N}\mathbf{J}_{N}.
\]
Therefore, for $t \in \{ 1, 2, \ldots, N \}$, the $t$th entries of the vectors in \eqref{eq:Y2} are
\[
W_t = [(\mathbf{I}_N - \mathcal{P}_{[\mathbf{A}|\mathbf{D}]})\boldsymbol\epsilon]_t
	 = \epsilon_t - \bar{\epsilon}_{r(t)}
	- \bar{\epsilon}_{v(t)} + \bar{\epsilon},
\]
and
\begin{equation}
\delta_t = [(\mathbf{I}_N - \mathcal{P}_{[\mathbf{A}|\mathbf{D}]})\mathbf{D}^0\boldsymbol\mu ]_t
	 = \mu_{r^0(t)} - \bar{\mu}_{r(t)}
\label{Lund!}
\end{equation}
\end{proof}

For any changepoint configuration $\boldsymbol\lambda \in 
\boldsymbol\Lambda$, as $N \rightarrow \infty$, $N^{-1}\sum_{t=h+1}^N 
\delta_t \delta_{t-h}$ converges to a constant that does not depend on 
the lag $h \in \{ 0, 1, \ldots, p \}$. This is because for any lag $h$, 
$\delta_t = \delta_{t-h}$ for all $t \in \{ 1, \ldots, N \}$, except for 
at most $(m + m^0)h \leq (m + m^0)p$ times near the changepoints in 
$\boldsymbol\lambda$ and $\boldsymbol\lambda^0$. Hence, as $N 
\rightarrow \infty$, $N^{-1}\sum_{t=h+1}^N\delta_t \delta_{t-h}$ 
converges to its limit at rate $O\left(1/N\right)$. We denote this limit 
as
\begin{equation}
\label{eq:delta_sq_average}
\delta^2 \stackrel{\text{def}}{=}
\lim_{N\rightarrow \infty} \frac{1}{N}\sum_{t=1}^N \delta_t^2
= \lim_{N\rightarrow \infty} \frac{1}{N}\sum_{t=1}^N \left(
\mu_{r^0(t)} - \bar{\mu}_{r(t)}\right)^2,
\end{equation}
which is non-negative and depends on $\boldsymbol\lambda$, but 
not on $N$. It is not hard to see that $\delta_t = 0$ for all $t \in \{ 1, 
\ldots, N \}$ if and only if $\boldsymbol{\lambda}$ contains all relative 
changepoints in $\boldsymbol{\lambda}^0$ (denoted by $\boldsymbol\lambda 
\supset \boldsymbol\lambda^0$). Therefore, $\delta^2=0$ only for models 
$\boldsymbol\lambda$ such that $\boldsymbol\lambda \supset 
\boldsymbol\lambda^0$, including $\boldsymbol{\lambda}^0$ itself.


\begin{lemma}\label{lemma:gamma_h} Under any relative changepoint 
configuration $\boldsymbol\lambda \in \boldsymbol\Lambda$ (which may or 
may not be the true changepoint configuration), for $h \in \{ 0, 1, \ldots, 
p \}$, as $N \rightarrow \infty$, the lag $h$ sample autocovariance
\begin{equation*}
\hat{\gamma}(h) = \frac{1}{N} \sum_{t = h + 1}^N 
\epsilon_t^{\text{ols}} \epsilon_{t-h}^{\text{ols}}
\end{equation*} 
obeys
\begin{equation}\label{eq:gamma_h}
\hat{\gamma}(h) 
= \gamma(h) +  
\delta^2 + O_P\left(\frac{1}{\sqrt{N}}\right),
\end{equation}
where $\gamma(h)$ is the true lag $h$ autocovariance for the AR$(p)$ 
series $\boldsymbol\epsilon$.
\end{lemma}

\begin{proof} Since the AR$(p)$ errors are assumed causal, we may write 
\[
\epsilon_t = \sum_{j = 0}^{\infty} \psi_{j} Z_{t-j}
\]
for some weights $\{ \psi_j \}_{j=0}^\infty$, where $\sum_{j = 
0}^{\infty} |\psi_j| < \infty$.  Since $W_t = \epsilon_t - 
\bar{\epsilon}_{r(t)}- \bar{\epsilon}_{v(t)} + \bar{\epsilon}$, one can 
write $W_t$ as a linear combination of all $Z_t$s up to and before time 
$N$:
\[
W_t = \sum_{j=-\infty}^{\infty} \psi_j^{(t)} Z_{t-j},
\]
where 
\begin{equation}
\label{eq:psi_j^t}
\psi_j^{(t)}= \psi_{j} - 
\frac{\sum_{k: r(k) = r(t)}\psi_{k-t+j}}{N_{r(t)}} 
- \frac{\sum_{l: v(l) = v(t)}\psi_{l-t+j}}{N/T} 
+ \frac{\sum_{u=1}^N \psi_{u-t+j}}{N}.
\end{equation}
Since $\psi_j =0$ when $j < 0$, $\psi_j^{(t)} = 0$ if $j < t-N$.

The asymptotic limit of the sample autocovariances can now be derived:
\begin{align}
\nonumber
	\hat{\gamma}(h) & = \frac{1}{N} \sum_{t = h + 1}^N \epsilon_t^{\text{ols}} \epsilon_{t-h}^{\text{ols}}
	=  \frac{1}{N} \sum_{t = h + 1}^N (W_t + \delta_t)(W_{t-h} + \delta_{t-h})\\ 
	\label{eq:autocov1}
&	=	\frac{1}{N} \sum_{t = h + 1}^N	
	(W_t W_{t-h}	+  \delta_{t-h} W_t  +  \delta_{t} W_{t-h} + \delta_t \delta_{t-h}).
\end{align}
Arguing as in Proposition 7.3.5 of \citet[pp.\ 
232]{Brockwell_Davis_1991} gives
\[
\frac{1}{N} \sum_{t = h + 1}^N W_t W_{t-h}
=	\frac{1}{N}\sum_{t=h+1}^N  \sum_{j=-\infty}^{\infty} \psi_j^{(t)}\psi_{j-h}^{(t-h)} Z_{t-j}^2
	+ O_P\left(\frac{1}{\sqrt{N}}\right).
\]
In \eqref{eq:psi_j^t}, since $\sum_{j=0}^\infty |\psi_j| < \infty$, and 
$N_{r(t)} = O(N)$ for all $t \in \{ 1, \ldots, N \}$, it is not difficult to 
show that there exists a positive finite constant $c$ such that,
\[
\sup_{t, j} \left|\psi_j^{(t)} -\psi_{j}\right| \leq \frac{c}{N}.
\]
Therefore, for each $t$ and $h$, $\left\{\psi_j^{(t)}\psi_{j-h}^{(t-h)} 
\right\}_{j=-\infty}^{\infty}$ is absolutely convergent, and
\[
\left| \sum_{j=-\infty}^{\infty} \psi_j^{(t)}\psi_{j-h}^{(t-h)} 
- \sum_{j=-\infty}^{\infty} \psi_j\psi_{j-h} \right| = O\left( \frac{1}{N} \right).
\] 
Since $\{ Z_t \}$ is iid with variance $\sigma^2$, the weak law of large numbers (WLLN) 
for linear processes \citep[pp.\ 208, Proposition 
6.3.10]{Brockwell_Davis_1991} gives
\begin{align*}
\frac{1}{N} \sum_{t = h + 1}^N W_t W_{t-h}	
&=	\frac{1}{N}\sum_{t=h+1}^N  
	\sum_{j=-\infty}^{\infty} \psi_j^{(t)}\psi_{j-h}^{(t-h)}\sigma^2 + O_P\left(\frac{1}{\sqrt{N}}\right)\\
&=	\frac{1}{N}\sum_{t=h+1}^N  
	\sum_{j=-\infty}^{\infty} \psi_{j}\psi_{j-h}\sigma^2 + O_P\left(\frac{1}{\sqrt{N}}\right).
\end{align*}
Now using that $\gamma(h) = \sigma^2 
\sum_{j=-\infty}^\infty \psi_j \psi_{j-h}$ gives
\[
\frac{1}{N} \sum_{t = h + 1}^N W_t W_{t-h}	
=\frac{N-h}{N}\gamma(h)+ O_P\left(\frac{1}{\sqrt{N}}\right)
=\gamma(h)+ O_P\left(\frac{1}{\sqrt{N}}\right).
\]
This identifies the limit of the first term in the bottom line of 
\eqref{eq:autocov1}.  By \eqref{eq:psi_j^t}, it is not hard to show that 
for each $t$, $\left\{\psi_j^{(t)} \right\}_{j = -\infty}^{\infty}$ is 
absolutely convergent. For the second and third terms in 
\eqref{eq:autocov1}, apply the WLLN again to see that these terms 
converge to zero in probability at rate $O_P(1/\sqrt{N})$. Hence, 
as $N \rightarrow \infty$,
\[
\hat{\gamma}(h) 
= 	\gamma(h) + \frac{1}{N}\sum_{t=h+1}^N \delta_t \delta_{t-h}   
	+ O_P\left(\frac{1}{\sqrt{N}}\right)
= \gamma(h) +  \delta^2 + O_P\left(\frac{1}{\sqrt{N}}\right).
\]
\end{proof}

Since the Yule-Walker estimator $\hat{\boldsymbol\phi}$ is formulated 
based on $\hat{\gamma}(h)$'s, the following asymptotic result follows from 
Lemma \ref{lemma:gamma_h}.

\begin{proposition}\label{prop:YW_phi} Under any relative changepoint 
configuration $\boldsymbol\lambda \in \boldsymbol\Lambda$, the 
Yule-Walker estimator $\hat{\boldsymbol\phi} = 
\hat{\boldsymbol\Gamma}_p^{-1} \hat{\boldsymbol\gamma}_p$ obeys
\begin{equation}
\label{eq:YW_phi_limit}
\hat{\boldsymbol\phi} = 
\left(\boldsymbol\Gamma_p + \delta^2 \mathbf{J}_p \right)^{-1}
\left(\boldsymbol\gamma_p + \delta^2 \mathbf{1}_p \right)
+ O_P\left(\frac{1}{\sqrt{N}}\right),
\end{equation}
where ${\boldsymbol\gamma}_p = ({\gamma}(1), \ldots, {\gamma}(p))'$ and 
${\boldsymbol\Gamma}_p$ is a $p \times p$ matrix with $(i,j)$th entry 
${\gamma}(|i-j|)$.
\end{proposition}


\subsection{Asymptotic behavior of estimators of ${\sigma}^2$} 
\label{subsection:YW_sigmasq}

In the BMDL and (automatic) MDL formulas, estimators for $\sigma^2$ are
\begin{align}\label{eq:sigmasq_hat_BMDL}
\hat{\sigma}^2 
&	= \frac{1}{N-p}\widehat{\mathbf{X}}' \left\{
 	\widehat{\mathbf{B}}
	- \widehat{\mathbf{B}} \widehat{\mathbf{A}}
	\left(\widehat{\mathbf{A}}' \widehat{\mathbf{B}}\widehat{\mathbf{A}}\right)^{-1}
	\widehat{\mathbf{A}}' \widehat{\mathbf{B}}
	\right\}\widehat{\mathbf{X}},\\ \label{eq:sigmasq_hat_MDL}
\hat{\sigma}^2_{\nu = \infty} 
&	= 
	\frac{1}{N-p}\widehat{\mathbf{X}}' 
	\left(\mathbf{I}_N - \mathcal{P}_{[\widehat{\mathbf{A}}|\widehat{\mathbf{D}}]}\right)
	\widehat{\mathbf{X}},
\end{align}
respectively. The following lemma shows that under any model 
$\boldsymbol\lambda$, these two estimators are asymptotically the same 
as the Yule-Walker estimator of $\sigma^2$, i.e.,
\begin{equation}
\label{eq:sigmasq_hat_YW}
\hat{\sigma}^2_{\text{YW}} = \hat\gamma(0) - \hat{\boldsymbol\gamma}_p' 
	\hat{\boldsymbol\Gamma}_p^{-1}\hat{\boldsymbol\gamma}_p.
\end{equation}

\begin{lemma}
\label{lemma:sigmasq1}
Under any changepoint configuration $\boldsymbol\lambda \in \boldsymbol\Lambda$, as $N 
\rightarrow \infty$,
\begin{align}
\label{eq:sigmasq1}
\hat{\sigma}^2
&=	\hat{\sigma}^2_{\nu = \infty} + O_P\left( \frac{1}{N} \right), \\ \label{eq:sigmasq2}
\hat{\sigma}^2_{\nu = \infty}
&=	\hat{\sigma}^2_{\text{YW}} + O_P\left( \frac{1}{N} \right).
\end{align}
\end{lemma}

\begin{proof} Under the null model $\boldsymbol\lambda_{\o}$ ($m = 0$), 
the column space of $\mathbf{D}$ is the null space and both 
$\hat{\sigma}^2$ and $\hat{\sigma}^2_{\nu = \infty}$ are 
$(N-p)^{-1}\widehat{\mathbf{X}}' \left(\mathbf{I}_N - 
\mathcal{P}_{\widehat{\mathbf{A}}}\right) \widehat{\mathbf{X}}$.  Since 
$\hat{\sigma}^2_{\text{YW}} = \frac{1}{N}\widehat{\mathbf{X}}' 
\left(\mathbf{I}_N - \mathcal{P}_{\widehat{\mathbf{A}}}\right) 
\widehat{\mathbf{X}}$, the conclusion holds.  The rest of the proof is 
for any model $\boldsymbol\lambda$ that contains $m \geq 1$ relative 
changepoints.

We first establish \eqref{eq:sigmasq1}. Since $\hat{\boldsymbol\phi}$ 
has the limit in \eqref{eq:YW_phi_limit}, it is not hard to show that as 
$N$ tends to infinity, $\widehat{\mathbf{D}}' \widehat{\mathbf{D}}/N$ 
and $\widehat{\mathbf{D}}' \widehat{\mathbf{X}}/N$ converge in 
probability to a $m \times m$ positive definite matrix and an 
$m$-dimensional vector, respectively. In the prior of $\boldsymbol\mu$, 
the parameter $\nu$ is a constant; hence,

\begin{align*}
	\frac{1}{N}\widehat{\mathbf{X}}'\widehat{\mathbf{B}}\widehat{\mathbf{X}}
&=	\frac{\widehat{\mathbf{X}}'\widehat{\mathbf{X}}}{N}
	- \frac{\widehat{\mathbf{X}}' \widehat{\mathbf{D}}}{N}
	\left(\frac{\widehat{\mathbf{D}}' \widehat{\mathbf{D}}}{N}
	+\frac{\mathbf{I}_{m}}{N \nu}\right)^{-1}
	\frac{\widehat{\mathbf{D}}' \widehat{\mathbf{X}}}{N}\\
&=	\frac{\widehat{\mathbf{X}}'\widehat{\mathbf{X}}}{N}
	- \frac{\widehat{\mathbf{X}}' \widehat{\mathbf{D}}}{N}
	\left(\frac{\widehat{\mathbf{D}}' \widehat{\mathbf{D}}}{N}\right)^{-1}
	\frac{\widehat{\mathbf{D}}' \widehat{\mathbf{X}}}{N}
	+ O_P\left( \frac{1}{N} \right) \\
&= 	\frac{1}{N}\widehat{\mathbf{X}}'
	\left(\mathbf{I}_{N-p} - \mathcal{P}_{\widehat{\mathbf{D}}}\right)
	\widehat{\mathbf{X}} + O_P\left( \frac{1}{N} \right).
\end{align*}
Similar arguments give
\begin{align*}
\frac{1}{N}\widehat{\mathbf{X}}'\widehat{\mathbf{B}}\widehat{\mathbf{A}}
&=	\frac{1}{N}\widehat{\mathbf{X}}'
	\left(\mathbf{I}_{N-p} - \mathcal{P}_{\widehat{\mathbf{D}}}\right)
	\widehat{\mathbf{A}} + O_P\left( \frac{1}{N} \right),\quad \\
\frac{1}{N}\widehat{\mathbf{A}}'\widehat{\mathbf{B}}\widehat{\mathbf{A}}
&=	\frac{1}{N}\widehat{\mathbf{A}}'
	\left(\mathbf{I}_{N-p} - \mathcal{P}_{\widehat{\mathbf{D}}}\right)
	\widehat{\mathbf{A}} + O_P\left( \frac{1}{N} \right).
\end{align*}
Hence, the left hand side of \eqref{eq:sigmasq1} has the limit
\begin{align*}
\label{eq:X_Q_Ahat_Dhat_X2}
\hat{\sigma}^2
=& \frac{1}{N-p}\widehat{\mathbf{X}}' \left\{
 	\widehat{\mathbf{B}}
	- \widehat{\mathbf{B}} \widehat{\mathbf{A}}
	\left(\widehat{\mathbf{A}}' \widehat{\mathbf{B}}\widehat{\mathbf{A}}\right)^{-1}
	\widehat{\mathbf{A}}' \widehat{\mathbf{B}}
	\right\}\widehat{\mathbf{X}}\\ \nonumber
=& \frac{1}{N-p}\widehat{\mathbf{X}}' \left\{
 	\mathbf{I}_N - \mathcal{P}_{\widehat{\mathbf{D}}}
	- \mathcal{P}_{\left(\mathbf{I}_{N} - \mathcal{P}_{\widehat{\mathbf{D}}}\right) \widehat{\mathbf{A}}}
	\right\}\widehat{\mathbf{X}} + O_P\left( \frac{1}{N} \right) \\ \nonumber
=& \frac{1}{N-p}\widehat{\mathbf{X}}' 
	\left(\mathbf{I}_{N-p} - \mathcal{P}_{[\widehat{\mathbf{A}}|\widehat{\mathbf{D}}]}\right)
	\widehat{\mathbf{X}} + O_P\left( \frac{1}{N} \right) \\
=& \hat{\sigma}^2_{\nu = \infty}+ O_P\left( \frac{1}{N} \right),
\end{align*}
where the second to last equality follows from \eqref{eq:Q_AD}.

We now show that for any $\boldsymbol\lambda$ with $m\geq 1$, \eqref{eq:sigmasq2} 
holds. For notational simplicity, for any $j \in \{ 0, 1, \ldots, p \}$, 
matrices formed from the rows of $\mathbf{A}$ and $\mathbf{D}$ are denoted 
by
\[
\mathbf{A}_j \stackrel{\text{def}}{=} \mathbf{A}_{(p+1-j):(N-j)},\quad
\mathbf{D}_j \stackrel{\text{def}}{=} \mathbf{D}_{(p+1-j):(N-j)}.
\]
Since both $\widehat{\mathbf{A}}$ and $\mathbf{A}_j$ are $(N-p)\times T$ 
matrices and each column in $\widehat{\mathbf{A}}$ can be written as a 
linear combination of the columns in $\mathbf{A}_j$, the corresponding 
column spaces agree: $C(\widehat{\mathbf{A}}) = C(\mathbf{A}_j)$. 
Therefore, $\mathcal{P}_{\widehat{\mathbf{A}}} = 
\mathcal{P}_{\mathbf{A}_j}$ for all $j$.  Now define
\begin{equation}
\label{eq:D_hat_Dj}
\boldsymbol\Delta_j =
\mathbf{D}_j - \frac{\widehat{\mathbf{D}}}
{1 - \hat{\phi}_1 - \hat{\phi}_2 - \cdots - \hat{\phi}_p}.
\end{equation}
The denominator in \eqref{eq:D_hat_Dj} cannot be zero since $1 - 
\sum_{k=1}^p \hat{\phi}_k \neq 0$ for Yule-Walker estimates when $N$ is 
large \citep{Brockwell_Davis_1991}.

Since there are at most $2m(p+h)$ non-zero entries in 
$\boldsymbol\Delta_j$, and none of these entries depend on $N$, 
$\boldsymbol\Delta_j' \boldsymbol\Delta_j = O_P(1)$. In addition, for 
any $N$-dimensional vectors $\boldsymbol\alpha$ whose entries do not 
depend on $N$, $\boldsymbol\alpha' \boldsymbol\Delta_j = O_P(1)$. Using
\eqref{eq:D_hat_Dj}, we have 
\begin{align*}
\frac{\widehat{\mathbf{D}}' \left(\mathbf{I}_N - \mathcal{P}_{\widehat{\mathbf{A}}}\right)
	\widehat{\mathbf{D}}}{N\left(1 - \sum_{k=1}^p \hat{\phi}_k\right)^2}
&=	\frac{1}{N} (\mathbf{D}_j - \boldsymbol\Delta_j)' 
	\left(\mathbf{I}_N - \mathcal{P}_{\widehat{\mathbf{A}}}\right)
	 (\mathbf{D}_j - \boldsymbol\Delta_j)\\
&= 	\frac{\mathbf{D}_j' \left(\mathbf{I}_N - \mathcal{P}_{\widehat{\mathbf{A}}}\right)
	\mathbf{D}_j}{N} + O_P\left(\frac{1}{N}\right),
\end{align*}
\begin{align*}	
\frac{\boldsymbol\alpha' \left(\mathbf{I}_N - \mathcal{P}_{\widehat{\mathbf{A}}}\right)
	\widehat{\mathbf{D}}}{N\left(1 - \sum_{k=1}^p \hat{\phi}_k\right)}
&= 	\frac{1}{N} \boldsymbol\alpha' \left(\mathbf{I}_N - \mathcal{P}_{\widehat{\mathbf{A}}}\right)
	 (\mathbf{D}_j - \boldsymbol\Delta_j)\\
& = 	\frac{\boldsymbol\alpha' \left(\mathbf{I}_N - \mathcal{P}_{\widehat{\mathbf{A}}}\right)
	\mathbf{D}_j}{N} + O_P\left(\frac{1}{N}\right).
\end{align*}
Therefore, for any $\boldsymbol\alpha, \boldsymbol\beta \in 
\mathbb{R}^N$ whose entries do not depend on $N$,
\begin{align*}
&\frac{1}{N}\boldsymbol\alpha' 
	\mathcal{P}_{\left(\mathbf{I}_N - \mathcal{P}_{\widehat{\mathbf{A}}}\right)
	\widehat{\mathbf{D}}}\boldsymbol\beta\\
=&  	\frac{\boldsymbol\alpha' \left(\mathbf{I}_N - \mathcal{P}_{\widehat{\mathbf{A}}}\right)
	\widehat{\mathbf{D}}}{N\left(1 - \sum_{k=1}^p \hat{\phi}_k\right)}
	\left\{ \frac{\widehat{\mathbf{D}}' \left(\mathbf{I}_N - \mathcal{P}_{\widehat{\mathbf{A}}}\right)
	\widehat{\mathbf{D}}}{N\left(1 - \sum_{k=1}^p \hat{\phi}_k\right)^2} \right\}^{-1}
	\frac{\widehat{\mathbf{D}}' \left(\mathbf{I}_N - \mathcal{P}_{\widehat{\mathbf{A}}}\right)
	\boldsymbol\beta}{N\left(1 - \sum_{k=1}^p \hat{\phi}_k\right)}\\
=&  	\frac{1}{N}\boldsymbol\alpha' \left\{
	\left(\mathbf{I}_N - \mathcal{P}_{\widehat{\mathbf{A}}}\right)\mathbf{D}_j
	\left( \mathbf{D}_j' \left(\mathbf{I}_N - \mathcal{P}_{\widehat{\mathbf{A}}}\right)\mathbf{D}_j \right)^{-1}
	\mathbf{D}_j' \left(\mathbf{I}_N - \mathcal{P}_{\widehat{\mathbf{A}}}\right)
	\right\}\boldsymbol\beta+ O_P\left(\frac{1}{N}\right)\\
=&  	\frac{1}{N}\boldsymbol\alpha' 
	\mathcal{P}_{\left(\mathbf{I}_N - \mathcal{P}_{\widehat{\mathbf{A}}}\right)
	\mathbf{D}_j}\boldsymbol\beta+ O_P\left(\frac{1}{N}\right).
\end{align*}
Hence, from \eqref{eq:Q_AD},
\begin{equation}
\label{eq:P_hatA_hatD}
\frac{1}{N}\boldsymbol\alpha' 
	\mathcal{P}_{[\widehat{\mathbf{A}}|\widehat{\mathbf{D}}]}\boldsymbol\beta
=	\frac{1}{N}\boldsymbol\alpha' 
	\mathcal{P}_{[\mathbf{A}_j|\mathbf{D}_j]}\boldsymbol\beta
	+ O_P\left(\frac{1}{N}\right).
\end{equation}

Since $\widehat{\mathbf{X}} = \mathbf{X}_{(p+1):N} - \sum_{j = 1}^p \hat{\phi}_j 
\mathbf{X}_{(p+1-j):(N-j)}$, for any $j, k \in \{0, 1, \ldots, p\}$, \eqref{eq:P_hatA_hatD} shows that
\begin{align*}
&	\frac{1}{N}\mathbf{X}_{(p+1-j):(N-j)}' 
	\left(\mathbf{I}_N - \mathcal{P}_{[\widehat{\mathbf{A}}|\widehat{\mathbf{D}}]}\right)
	\mathbf{X}_{(p+1-k):(N-k)}\\
=&      ~\frac{1}{N}\left\{\left(\mathbf{I}_N - \mathcal{P}_{[\mathbf{A}_j|\mathbf{D}_j]}\right)
	\mathbf{X}_{(p+1-j):(N-j)}\right\}'
	\left\{\left(\mathbf{I}_N - \mathcal{P}_{[\mathbf{A}_k|\mathbf{D}_k]}\right)
	\mathbf{X}_{(p+1-k):(N-k)}\right\} + O_P\left(\frac{1}{N}\right)\\
=&      ~\frac{1}{N}\left(\boldsymbol\epsilon_{(p+1-j):(N-j)}^{\text{ols}}\right)' \boldsymbol\epsilon_{(p+1-k):(N-k)}^{\text{ols}}
	+ O_P\left(\frac{1}{N}\right).
\end{align*}
Therefore, the left hand side of \eqref{eq:sigmasq2} is 	
\begin{align*}
&	\frac{1}{N}\widehat{\mathbf{X}}' 
	\left(\mathbf{I}_N - \mathcal{P}_{[\widehat{\mathbf{A}}|\widehat{\mathbf{D}}]}\right)
	\widehat{\mathbf{X}}  \\
=&      ~\frac{1}{N}\left\{ \boldsymbol\epsilon_{(p+1):N}^{\text{ols}} 
	- \sum_{j = 1}^p \hat{\phi}_j \boldsymbol\epsilon_{(p+1-j):(N-j)}^{\text{ols}} \right\}'
	\left\{ \boldsymbol\epsilon_{(p+1):N}^{\text{ols}} 
	- \sum_{k = 1}^p \hat{\phi}_k \boldsymbol\epsilon_{(p+1-k):(N-k)}^{\text{ols}} \right\} 
	+ O_P\left( \frac{1}{N} \right)\\
=&      ~\hat{\gamma}(0) - 2\sum_{j=1}^p \hat{\phi}_j \hat{\gamma}(j)
	+ \sum_{j=1}^p \sum_{k=1}^p \hat{\phi}_j \hat{\phi}_k \hat{\gamma}(|j-k|)
	+ O_P\left( \frac{1}{N} \right)\\
=&      ~\hat\gamma(0) - 2\hat{\boldsymbol\gamma}_p' \hat{\boldsymbol\phi} 
	+\hat{\boldsymbol\phi} ' 
	\hat{\boldsymbol\Gamma}_p\hat{\boldsymbol\phi} + O_P\left( \frac{1}{N} \right)\\
=&	~\hat\gamma(0) - \hat{\boldsymbol\gamma}_p' 
	\hat{\boldsymbol\Gamma}_p^{-1}\hat{\boldsymbol\gamma}_p
	+ O_P\left( \frac{1}{N} \right),
\end{align*}
which is the right hand side of \eqref{eq:sigmasq2}. 
\end{proof}

Under any model $\boldsymbol\lambda$, Lemma \ref{lemma:gamma_h} shows 
that the Yule-Walker estimator $\hat{\sigma}^2_{\text{YW}}$ converges to
\begin{equation}
\label{eq:f}
f(\delta^2)\stackrel{\text{def}}{=} 
~\gamma(0) + \delta^2 
	- \left(\boldsymbol\gamma_p + \delta^2 \mathbf{1}_p \right)'
	\left(\boldsymbol\Gamma_p + \delta^2 \mathbf{J}_p \right)^{-1}
	\left(\boldsymbol\gamma_p + \delta^2 \mathbf{1}_p \right),
\end{equation}
at rate $O_P(1/\sqrt{N})$. We define the limit in (\ref{eq:f}) as 
$f(\delta^2)$, emphasizing dependence on $\delta^2$. By Lemma 
\ref{lemma:sigmasq1}, the asymptotic behavior of the BMDL estimator 
$\hat{\sigma}^2$ can be summarized in the following proposition.

\begin{proposition}\label{prop:sigmasq2} Under any relative changepoint 
configuration $\boldsymbol\lambda \in \boldsymbol\Lambda$, the BMDL 
estimator of the white noise variance in $\hat{\sigma}^2$ 
\eqref{eq:sigmasq_hat_BMDL} obeys
\begin{equation}
\label{eq:sigmasq3}
\hat{\sigma}^2= f(\delta^2) + O_P\left( \frac{1}{\sqrt{N}} \right),
\end{equation}
where $f(\delta^2)$ is defined in \eqref{eq:f}. Furthermore, $f(\delta^2)$ 
strictly increases in $\delta^2$. 
\end{proposition}

\begin{proof} We show that $f(\delta^2)$ strictly increases in 
$\delta^2$. According to (2.22) in \citet[pp.\ 
428]{Harville_2008}, for any matrices $\mathbf{R}\in \mathbb{R}^{r 
\times r}, \mathbf{S}\in \mathbb{R}^{r \times l}, \mathbf{T}\in 
\mathbb{R}^{l \times l}, \mathbf{U}\in \mathbb{R}^{l \times r}$ with 
$\mathbf{R}, \mathbf{U}$ non-singular, $(\mathbf{R} + 
\mathbf{S}\mathbf{T}\mathbf{U})^{-1} = \mathbf{R}^{-1} - 
\mathbf{R}^{-1}\mathbf{S} (\mathbf{T}^{-1} + 
\mathbf{U}\mathbf{R}^{-1}\mathbf{S})^{-1} \mathbf{U}\mathbf{R}^{-1}$. 
Hence, for $\delta^2>0$, 
\begin{equation}
\label{eq:inv_matrix}
\left(\boldsymbol\Gamma_p + \delta^2 \mathbf{J}_p \right)^{-1}
=	\left(\boldsymbol\Gamma_p + \mathbf{1}_p 
	\delta^2  \mathbf{1}_p' \right)^{-1}
=	\boldsymbol\Gamma_p^{-1} - \boldsymbol\Gamma_p^{-1}\mathbf{1}_p
	\left(\frac{1}{\delta^2} + \mathbf{1}_p'
	\boldsymbol\Gamma_p^{-1}\mathbf{1}_p\right)^{-1}
	\mathbf{1}_p'\boldsymbol\Gamma_p^{-1}.
\end{equation}
For notational simplicity, denote the following scalars by
\begin{equation}
\label{eq:ab}
a \stackrel{\text{def}}{=}
	\mathbf{1}_p' \boldsymbol\Gamma_p^{-1}\mathbf{1}_p, \quad
b \stackrel{\text{def}}{=}
	\mathbf{1}_p' \boldsymbol\Gamma_p^{-1}\boldsymbol\gamma_p=\sum_{k=1}^p \phi_k.
\end{equation} 
Then $f(\delta^2)$ can be expanded as
\[
f(\delta^2) 
= \gamma(0) + \delta^2 - 
	\boldsymbol\gamma_p' \boldsymbol\Gamma_p^{-1}\boldsymbol\gamma_p
	- 2b\delta^2  - a(\delta^2)^2 + \frac{b^2}{\frac{1}{\delta^2}+a}
	+ \frac{2ab\delta^2}{\frac{1}{\delta^2}+a} + \frac{a^2(\delta^2)^2}{\frac{1}{\delta^2}+a}.
\]
Differentiation of $f(\delta^2)$ with respect to $\delta^2$ gives
\[
f'(\delta^2)
=	1-2b-2a\delta^2 + \frac{b^2 \frac{1}{(\delta^2)^2}}{\left(\frac{1}{\delta^2}+a\right)^2}
	+ \frac{2ab\left( \frac{2}{\delta^2} + a\right)}{\left(\frac{1}{\delta^2}+a\right)^2}
	+ \frac{a^2\left( 3 + 
2a\delta^2\right)}{\left(\frac{1}{\delta^2}+a\right)^2} =	
\frac{(b-1)^2}{(1 + a\delta^2)^2} > 0. 
\] 
The strict inequality follows from causality of the AR($p$) errors, 
which implies that $b = \sum_{k=1}^p \phi_k > 1$. Therefore, 
$f(\delta^2)$ is strictly increasing in $\delta^2$ and $f(0)=\sigma^2$. 
\end{proof}

\subsection{Asymptotic behavior of the BMDL in \eqref{eq:BMDL_univariate}}
\label{sec:asym_pairwise_BMDL}
Recall that under the relative changepoint model $\boldsymbol\lambda$, 
its BMDL in \eqref{eq:BMDL_univariate} is
\begin{align*}
\text{BMDL}(\boldsymbol\lambda) = ~&
	 \frac{N-p}{2}\log \left( \hat{\sigma}^2 \right) + \frac{m}{2}\log(\nu) 
	+ \frac{1}{2}\log\left( \left|\widehat{\mathbf{D}}' 
	\widehat{\mathbf{D}}+\frac{\mathbf{I}_{m}}{\nu}\right| \right)\\ \nonumber
&	-\sum_{k=1}^2\log\left\{\Gamma\left(a + m^{(k)}\right) 
\Gamma\left(b^{(k)} + N^{(k)} - m^{(k)}\right) \right\}.
\end{align*}
The next two lemmas quantify the asymptotic behavior of the third and forth 
terms in the above BMDL formula, respectively.

\begin{lemma}\label{lemma:det} Under any changepoint model 
$\boldsymbol\lambda \in \boldsymbol\Lambda$ with $m>0$,
\begin{equation}
\label{eq:det}
\frac{1}{2}\log\left( \left|\widehat{\mathbf{D}}' 
	\widehat{\mathbf{D}}+\frac{\mathbf{I}_{m}}{\nu}\right| \right)
= \frac{1}{2}\sum_{r=2}^{m+1}\log(N_r) -m \log \left( 1- \sum_{k=1}^p \hat{\phi}_k \right)
	+ O_P\left(\frac{1}{N}\right).
\end{equation}
\end{lemma}

\begin{proof} By \eqref{eq:D_hat_Dj} and the corresponding results in 
the proof of Lemma \ref{lemma:sigmasq1}, as $N \rightarrow \infty$,

\[
\frac{\widehat{\mathbf{D}}' \widehat{\mathbf{D}}}{N}+\frac{\mathbf{I}_{m}}{N\nu}
= 	\frac{\widehat{\mathbf{D}}' \widehat{\mathbf{D}}}{N} + O\left(\frac{1}{N}\right) 
=	\frac{\mathbf{D}' \mathbf{D}}{N\left( 1- \sum_{k=1}^p \hat{\phi}_k \right)^2}
	 + O_P\left(\frac{1}{N}\right).
\]
The determinant of the $m \times m$ matrix (of finite dimension) is then
\begin{align*}
\log\left( \left|\widehat{\mathbf{D}}' 
	\widehat{\mathbf{D}}+\frac{\mathbf{I}_{m}}{\nu}\right| \right)
& = m \log(N) + \log\left(\left| \frac{\widehat{\mathbf{D}}' \widehat{\mathbf{D}}}{N}
	+\frac{\mathbf{I}_{m}}{N\nu} \right| \right)\\
& = m \log(N) + \log\left(\frac{\left|\mathbf{D}' \mathbf{D}\right|}
	{N^m\left( 1- \sum_{k=1}^p \hat{\phi}_k \right)^{2m}} \right)+ O_P\left(\frac{1}{N}\right)\\
& = \log\left( \left|\mathbf{D}' \mathbf{D}\right|\right)
	-2m \log \left( 1- \sum_{k=1}^p \hat{\phi}_k \right) + O_P\left(\frac{1}{N}\right)\\
& = \log\left( \prod_{r=2}^{m+1} N_r\right)
	-2m \log \left( 1- \sum_{k=1}^p \hat{\phi}_k \right) + O_P\left(\frac{1}{N}\right),
\end{align*}
and \eqref{eq:det} follows immediately. 
\end{proof}

Since $N_r = O(N)$ for all $r \in \{ 2, \ldots, m+1\}$, Lemma \ref{lemma:det} 
implies that for any changepoint model $\boldsymbol\lambda$,
\begin{equation}
\label{eq:det2}
\frac{1}{2}\log\left( \left|\widehat{\mathbf{D}}' 
	\widehat{\mathbf{D}}+\frac{\mathbf{I}_{m}}{\nu}\right| \right)
= \frac{m}{2}\log(N) + O_P\left(1\right).
\end{equation}

\begin{lemma}\label{lemma:Gamma_pairdiff} Suppose that both the number 
of documented and undocumented times increases linearly with $N$, i.e., 
$N^{(k)} = O(N)$, for $k = 1, 2$. Then under any two changepoint models 
$\boldsymbol\lambda_1, \boldsymbol\lambda_2 \in \boldsymbol\Lambda$, 
whose total number of changepoints are $m_1, m_2$, respectively, the 
pairwise difference of the last term in the BMDL formula 
\eqref{eq:BMDL_univariate} is
\begin{align} \nonumber
& -\sum_{k=1}^2\left[\log\left\{\Gamma\left(a + m_1^{(k)}\right) 
\Gamma\left(b^{(k)} + N^{(k)} - m_1^{(k)}\right) \right\} \right.\\ \nonumber
&~~~~~~~~~\left. -\log\left\{\Gamma\left(a + m_2^{(k)}\right) 
\Gamma\left(b^{(k)} + N^{(k)} - m_2^{(k)}\right) \right\} \right]\\ \label{eq:Gamma_pairdiff} 
=~&  (m_1 - m_2) \log(N) + O_P(1).
\end{align}
\end{lemma}

\begin{proof}
The left hand side of \eqref{eq:Gamma_pairdiff} can be simplified to
\begin{equation}\label{eq:Gamma_pairdiff2}
\sum_{k=1}^2\log \left\{\frac{\Gamma\left(a + m_2^{(k)}\right)
	\Gamma\left(b^{(k)} + N^{(k)} - m_2^{(k)}\right)}
	{\Gamma\left(a + m_1^{(k)}\right)
	\Gamma\left(b^{(k)} + N^{(k)} - m_1^{(k)} \right)}\right\}.
\end{equation}
Stirling's formula quantifies the asymptotic limit of the following Gamma 
function ratio:
\begin{align*}
\frac{\Gamma\left(b^{(k)} + N^{(k)} - m_2^{(k)}\right)}
{\Gamma\left(b^{(k)} + N^{(k)} - m_1^{(k)} \right)} 
& \approx~	 
	e^{m_2^{(k)} - m_1^{(k)}}~\frac{\left(b^{(k)} + N^{(k)} - m_2^{(k)}-1\right)^{b^{(k)} + N^{(k)} - m_2^{(k)} - 1/2}}
	{\left(b^{(k)} + N^{(k)} - m_1^{(k)}-1 \right)^{b^{(k)} + N^{(k)} - m_1^{(k)} - 1/2}}\\
& \approx~	
	\left(\frac{N}{e}\right)^{m_1^{(k)} - m_2^{(k)}}.
\end{align*}
Therefore, \eqref{eq:Gamma_pairdiff2} equals $(m_1 - m_2) \log N + 
O_P(1)$. 
\end{proof}

The asymptotic behavior of the BMDL is now established in the following two 
propositions. They consider the pairwise difference of BMDLs between the 
true model $\boldsymbol\lambda^0$ and another changepoint model 
$\boldsymbol\lambda$. Proposition \ref{prop:pairwise_BMDL_OpN} considers 
the case where the model $\boldsymbol\lambda$ does not contain all 
relative changepoints in $\boldsymbol\lambda^0$, i.e., 
$\boldsymbol{\lambda} \not\supset \boldsymbol{\lambda}^0$, whereas 
Proposition \ref{prop:pairwise_BMDL_OplogN} considers the case where 
$\boldsymbol{\lambda} \supset \boldsymbol{\lambda}^0$, i.e., 
$\boldsymbol\lambda$ contains all relative changepoints in 
$\boldsymbol\lambda^0$, and also may have some redundant changepoints.

\begin{proposition}\label{prop:pairwise_BMDL_OpN} For any relative 
changepoint configuration $\boldsymbol{\lambda} \in 
\boldsymbol{\Lambda}$, if $\boldsymbol{\lambda} \not\supset 
\boldsymbol{\lambda}^0$, then as $N \rightarrow \infty$,
\[
\text{BMDL}\left(\boldsymbol\lambda\right) 
> \text{BMDL}\left(\boldsymbol\lambda^0\right), \quad
\text{BMDL}\left(\boldsymbol\lambda\right) 
-\text{BMDL}\left(\boldsymbol\lambda^0\right) = O_P(N).
\]
\end{proposition}

\begin{proof}
In this proof, when necessary, subscripts $\boldsymbol\lambda$ and 
$\boldsymbol\lambda^0$ are used to distinguish the same terms under 
different models. By \eqref{eq:det2} and \eqref{eq:Gamma_pairdiff}, the 
difference between BMDLs in the (non-true) model $\boldsymbol\lambda$ 
and the true model $\boldsymbol\lambda^0$ is asymptotically
\begin{align}
\label{eq:diff_MDL} \nonumber
&	\text{BMDL}\left(\boldsymbol\lambda\right) 
	-\text{BMDL}\left(\boldsymbol\lambda^0\right) \\
=~&      \frac{N-p}{2}\log \left(\frac{\hat{\sigma}^2_{\boldsymbol\lambda}}
	{\hat{\sigma}^2_{\boldsymbol\lambda^0}}\right)
	+ \frac{3(m-m^0)}{2}  \log (N) + O_P\left(1\right)	\\ \label{eq:diff_MDL2}
=~&      \frac{N-p}{2}\log \left\{\frac{f(\delta^2_{\boldsymbol\lambda})+ O_P\left(\frac{1}{\sqrt{N}}\right)}
	{f(0)+ O_P\left(\frac{1}{\sqrt{N}}\right)}\right\}
	+ \frac{3(m-m^0)}{2} \log (N) + O_P\left(1\right)	. 
\end{align}
Here, the last equality is justified via Proposition 
\ref{prop:sigmasq2}. For the model $\boldsymbol{\lambda} \not\supset 
\boldsymbol{\lambda}^0$, its corresponding 
$\delta^2_{\boldsymbol\lambda} > 0$. By Proposition \ref{prop:sigmasq2}, 
$f(\delta^2)$ strictly increases in $\delta^2$, which shows that the 
leftmost logarithm term in \eqref{eq:diff_MDL2} has a strictly positive 
limit. Therefore, when $N$ is large, the first term in 
\eqref{eq:diff_MDL2} is positive, of order $O_P(N)$, and dominates the 
other terms in \eqref{eq:diff_MDL2}.
\end{proof}

\begin{proposition}\label{prop:pairwise_BMDL_OplogN} For any relative 
changepoint configuration $\boldsymbol{\lambda} \in 
\boldsymbol{\Lambda}$, if $\boldsymbol{\lambda} \supset 
\boldsymbol{\lambda}^0$, then as $N \rightarrow \infty$,
\[
\text{BMDL}\left(\boldsymbol\lambda\right) 
> \text{BMDL}\left(\boldsymbol\lambda^0\right), \quad
\text{BMDL}\left(\boldsymbol\lambda\right) 
-\text{BMDL}\left(\boldsymbol\lambda^0\right) = O_P(\log N).
\]
\end{proposition}

\begin{proof} In the case where $\boldsymbol{\lambda} \supset 
\boldsymbol{\lambda}^0$, \eqref{eq:diff_MDL} still holds.  Moreover, since 
$\boldsymbol{\lambda}$ also contains redundant changepoints, $m 
> m^0$.  Hence, for large $N$, the second term in \eqref{eq:diff_MDL}
is positive and of order $O_P(\log N)$. To prove Proposition 
\ref{prop:pairwise_BMDL_OplogN}, we need to show that the first term in 
\eqref{eq:diff_MDL} is bounded in probability.  A sufficient condition 
for this simply shows that
\begin{equation}
\label{eq:sigmasq3}
\hat{\sigma}^2_{\boldsymbol\lambda} = \hat{\sigma}^2_{\boldsymbol\lambda^0} 
+ O_P\left(\frac{1}{N}\right).
\end{equation}

To establish \eqref{eq:sigmasq3}, we first focus on the model 
$\boldsymbol{\lambda}$. For notational simplicity, the subscript 
$\boldsymbol\lambda$ is omitted when there is no ambiguity. Under any 
model $\boldsymbol{\lambda} \supset \boldsymbol{\lambda}^0$, its 
corresponding $\delta_t$ in \eqref{eq:delta_t_W_t} is zero for all $t 
\in \{ 1, \ldots, N \}$; hence, by Lemma \ref{lemma:Y}, the lag-$h$ 
sample autocovariance $\hat{\gamma}(h)$ in \eqref{eq:autocov1} for all 
$h \in \{ 0, 1, \ldots, p\}$ can be written as
\begin{align} \nonumber
\hat{\gamma}(h) 
&	= \frac{1}{N} \sum_{t = h + 1}^N W_t W_{t-h} \\ \nonumber
&	= \frac{1}{N} \sum_{t = h + 1}^N 
\left(\epsilon_t - \bar{\epsilon}_{r(t)} - \bar{\epsilon}_{v(t)} + \bar{\epsilon}\right)
\left(\epsilon_{t-h} - \bar{\epsilon}_{r(t-h)} - \bar{\epsilon}_{v(t-h)} + \bar{\epsilon}\right)\\
\label{eq:autocov2}
&	= \frac{1}{N} \sum_{t = h + 1}^N \left\{
	\epsilon_t \epsilon_{t-h}
	- \epsilon_t \left( \bar{\epsilon}_{r(t-h)} + \bar{\epsilon}_{v(t-h)} - \bar{\epsilon} \right)
	- \epsilon_{t-h} \left( \bar{\epsilon}_{r(t)} + \bar{\epsilon}_{v(t)} - \bar{\epsilon} \right)
	\right.\\ \nonumber
&	\left.~~~~~~~~~~~~~~~~
	+ \left( \bar{\epsilon}_{r(t-h)} + \bar{\epsilon}_{v(t-h)} - \bar{\epsilon} \right)
	\left( \bar{\epsilon}_{r(t)} + \bar{\epsilon}_{v(t)} - \bar{\epsilon} \right)
	\right\}. 
\end{align} 

Recall that $\bar{\epsilon}_{r(\cdot)}, \bar{\epsilon}_{v(\cdot)}, 
\bar{\epsilon}$ are averages of zero-mean AR$(p)$ errors.  These 
averages are taken over error blocks whose size is proportional to $N$.  
By the central limit theorem for linear processes, these averages all 
converge to zero in probability with order $O_P(1/\sqrt{N})$. Since the 
fourth term in \eqref{eq:autocov2} is a sum of their two-way 
interactions and quadratic forms, it is also $O_P(1/N)$. The second term 
in \eqref{eq:autocov2} can be expanded as
\begin{align*}
&\frac{1}{N} \sum_{t = h + 1}^N 
	\epsilon_t \left( \bar{\epsilon}_{r(t-h)} + \bar{\epsilon}_{v(t-h)} - \bar{\epsilon} \right)\\
=~&  \frac{1}{N} \left\{ 
	\sum_{r = 1}^{m + 1} \sum_{t = 1}^{N_r}\epsilon_{r, t} \bar{\epsilon}_{r} 
	+ \sum_{v = 1}^{T} \sum_{t = 1}^{N/T}\epsilon_{v, t} \bar{\epsilon}_{v} 
	+ \sum_{t = 1}^N \epsilon_t \bar{\epsilon} + O_P(1)
	\right\}\\
=~&  \frac{1}{N} \left\{ 
	\sum_{r = 1}^{m + 1} N_r \bar{\epsilon}_{r}^2 
	+ \sum_{v = 1}^{T} \left(\frac{N}{T}\right) \bar{\epsilon}_{v}^2 
	+  N \bar{\epsilon}^2
	\right\} + O_P\left( \frac{1}{N} \right)\\
=~&  O_P\left( \frac{1}{N} \right),	
\end{align*}
where $\epsilon_{r, t}$ denotes the error during time $t$ in the $r$th 
regime, $\epsilon_{v, t}$ denotes the error during time $t$ in the $v$th 
month, and $\bar{\epsilon}_{r}$ and $\bar{\epsilon}_{v}$ are the error 
averages for the $r$th regime and $v$th month, respectively. Similarly, 
we can show that the third term in \eqref{eq:autocov2} is also 
$O_P(1/N)$. Therefore, under any model $\boldsymbol{\lambda} \supset 
\boldsymbol{\lambda}^0$, including $\boldsymbol{\lambda}^0$ itself, 
\eqref{eq:autocov2} becomes
\[
\hat{\gamma}(h) = \frac{1}{N} \sum_{t = h + 1}^N 
	\epsilon_t \epsilon_{t-h} + O_P\left( \frac{1}{N} \right),
\]
which shows that $\hat{\gamma}(h)$ under the two models 
$\boldsymbol{\lambda}$ and $\boldsymbol{\lambda}^0$ only changes by 
$O_P(1/N)$. By \eqref{eq:sigmasq_hat_YW}, $\hat{\sigma}^2_{\text{YW}}$ 
under the two models $\boldsymbol{\lambda}$ and $\boldsymbol{\lambda}^0$ 
also can only differ by $O_P(1/N)$. By Lemma \ref{lemma:sigmasq1}, the BMDL 
estimator $\hat{\sigma}^2 = \hat{\sigma}^2_{\text{YW}} + O_P(1/N)$, 
which establishes \eqref{eq:sigmasq3}.  Thus, $\hat{\sigma}^2$ under the 
two models $\boldsymbol{\lambda}$ and $\boldsymbol{\lambda}^0$ only 
differ by $O_P(1/N)$. 
\end{proof}

\subsection{A proof of Theorem \ref{thm:lambda_convergence}}
\label{PROOFthm:lambda_convergence}
To prove Theorem \ref{thm:lambda_convergence}, we first establish the 
asymptotic consistency of $\hat{\boldsymbol\lambda}_N$ in the case where 
$m^0$ is known. Here, $\boldsymbol\Lambda_m$ denotes a subset of 
$\boldsymbol\Lambda$ formed by models that have $m$ relative 
changepoints.

\begin{proposition}
\label{prop:lambda_convergence_known_m0} 
If $m^0$ is known, then as $N \rightarrow \infty$, 

\[
\hat{\boldsymbol\lambda}_N = \arg\min_{\boldsymbol\lambda 
\in \boldsymbol\Lambda_{m^0}} \text{BMDL}(\boldsymbol\lambda)
\]

\noindent satisfies $\hat{\boldsymbol\lambda}_N 
\stackrel{P}{\longrightarrow} \boldsymbol\lambda^0$. 
\end{proposition}

\begin{proof} We will show that for each subsequence $N_k$ with $N_k 
\rightarrow \infty$ as $k \rightarrow \infty$, there is a further 
subsequence $N_{k_\ell}$ with $N_{k_\ell} \rightarrow \infty$ as $\ell 
\rightarrow \infty$ such that $\hat{\boldsymbol\lambda}_{N_{k_\ell}} 
\stackrel {w} \longrightarrow \boldsymbol\lambda^0$ as $\ell \rightarrow 
\infty$, where $\stackrel {w} \longrightarrow$ denotes weak convergence, i.e.,
convergence in distribution.  
By the results in Section 25 of \cite{Billingsley_1995}, this implies 
that $\hat{\boldsymbol\lambda}_N \stackrel {w} \longrightarrow 
\boldsymbol \lambda^0$.  However, since $\boldsymbol \lambda^0$ is a 
constant configuration, one can upgrade the mode of convergence to infer 
that $\hat{\boldsymbol \lambda}_N \stackrel {P} \longrightarrow 
\boldsymbol \lambda^0$ (see again Section 25 of 
\cite{Billingsley_1995}).

Hence, let $N_k$ be an infinite sequence with $N_k \rightarrow \infty$ 
as $k \rightarrow \infty$.  By Helly's selection theorem (Theorem 25.9 
in \cite{Billingsley_1995}) and the compactness of $\Lambda_{m^0}$, 
there exists a further infinite subsequence $N_{k_\ell}$ and a possibly 
random configuration $\boldsymbol \lambda^*$ such that $ \hat{ 
\boldsymbol \lambda}_{N_{k_\ell}} \stackrel {w} {\longrightarrow} 
\boldsymbol \lambda^*$.  Here, a random configuration $\boldsymbol 
\lambda^*$ means a random variable ${\bf a}=(a_1, \ldots, 
a_{m^0})^\prime$ such that $0 \leq a_1 < a_2 < \ldots < a_{m^0} \leq 1$.  
To finish the argument, it is sufficient to show that $\boldsymbol 
\lambda^* = \boldsymbol \lambda^0$.

To show that $\boldsymbol \lambda^* = \boldsymbol \lambda^0$, we use 
proof by contradiction and suppose that $\boldsymbol \lambda^* \ne 
\boldsymbol \lambda^0$ in that $P(\boldsymbol \lambda^* \ne \boldsymbol 
\lambda^0) > 0$.  For notational simplicity, we simply replace 
$N_{k_\ell}$ by $N$ below.  Let $F_{\hat{\boldsymbol \lambda}_N}(\cdot)$ 
and $F_{\boldsymbol \lambda^*}(\cdot)$ denote the cumulative 
distribution functions of $\hat{\boldsymbol \lambda}_N$ and $\boldsymbol 
\lambda^*$, respectively, and define
\[ 
\delta^2_{\hat{\boldsymbol \lambda}_N}= \int_{{\bf a} \in \boldsymbol 
\Lambda_{m^0}} \delta^2({\bf a}) dF_{\hat{\boldsymbol \lambda}_N}({\bf a}), 
\quad 
\delta^2_{\boldsymbol \lambda^*}= \int_{{\bf a} \in 
\boldsymbol \Lambda_{m^0}} \delta^2({\bf a}) dF_{\boldsymbol \lambda^*}({\bf a}), 
\]
where the function $\delta^2(\cdot)$ is defined by \eqref{eq:delta_sq_average}.

It is easy to verify that $\delta^2({\bf a})$ is a continuous function 
in ${\bf a}$: For a fixed configuration $\mathbf{a}$ and the truth $\mathbf{a}^0
= (a_1^0, \ldots, a_{m^0}^0)$, we can rewrite their regime means as
\[
\mu_{r^0(t)} = 
\begin{cases}
\Delta_1^0, & 1 \leq t \leq \lfloor a_1^0 N \rfloor,\\
\vdots & \vdots \\
\Delta_{m^0 + 1}^0, & \lfloor a_{m^0}^0 N \rfloor + 1 \leq t \leq N,\\
\end{cases}
\]
and 
\[
\bar{\mu}_{r(t)} = 
\begin{cases}
\Delta_1, & 1 \leq t \leq \lfloor a_1 N \rfloor,\\
\vdots & \vdots \\
\Delta_{m^0 + 1}, & \lfloor a_{m^0} N \rfloor + 1 \leq t \leq N.\\
\end{cases}
\]
We then make a vector $\mathbf{b}$ of dimension at most $2m^0$ by ordering all
components in both $\mathbf{a}$ and $\mathbf{a}^0$. Thus,
\begin{equation}
\label{eq:delta_sq_average_continuous}
\delta^2(\mathbf{a}) = \lim_{N\rightarrow \infty} \frac{1}{N}\sum_{t=1}^N \left(
\mu_{r^0(t)} - \bar{\mu}_{r(t)}\right)^2
= \sum_{i=1}^{2m^0 + 1} (b_{i+1} - b_i)^2 w_i,
\end{equation}
where $b_i$ is a component in  $\mathbf{a}$ or $\mathbf{a}^0$, 
and $w_i$ has form $\pm (\Delta_k^0 - \Delta_j)$, 
$\pm (\Delta_k^0 - \Delta_j^0)$, or $\pm (\Delta_k - \Delta_j)$,  for some $k, j \in \{1, 2, \ldots, m^0\}$.

Therefore, \eqref{eq:delta_sq_average_continuous} is continuous in 
$\mathbf{a}$. We also tacitly assume that all regime mean parameters 
$\Delta_k$ are bounded.  By Part (ii) of Theorem 25.8 in 
\cite{Billingsley_1995}, if $X_N \stackrel{w}{\longrightarrow} X$ and a 
function $g(\cdot)$ is continuous and bounded, then $E[g(X_N)] 
\longrightarrow E[g(X)]$ as $N \rightarrow \infty$. Therefore, it 
follows that
\begin{equation}
\delta^2_{\hat{\boldsymbol \lambda}_N} \longrightarrow \delta^2_{\boldsymbol \lambda^*}.
\label{dump}
\end{equation}

Our work can be reduced to showing that $\text{BMDL}(\hat{\boldsymbol 
\lambda}_N) - \text{BMDL}(\boldsymbol \lambda^0)$ is bigger than a 
positive constant for all large $N$; for if this holds, then the fact 
that $\hat{\boldsymbol \lambda}_N$ minimizes the BMDL would be 
contradicted.  Hence, it suffices to show that
\[
\lim \sup_{N \rightarrow \infty} 
\frac{2}{N}
\left[ \text{BMDL}(\hat{\boldsymbol \lambda}_N ) - 
\text{BMDL}(\boldsymbol \lambda^0) \right] > 0.
\]
\noindent To do this, since $m^0$ is known, $\hat{m} = m^0$ and 
\eqref{eq:diff_MDL2} now give
\begin{align} \nonumber
~& \frac{2}{N}
\left[ 
\text{BMDL}(\hat{\boldsymbol \lambda}_N ) - \text{BMDL}(\boldsymbol \lambda^0) 
\right] \\
= ~&
\frac{2}{N}
\left[ \text{BMDL}(\hat{\boldsymbol\lambda}_N)
- \text{BMDL}(\boldsymbol\lambda^*) \right] 
\nonumber 
+ 
\frac{2}{N} 
\left[ \text{BMDL}(\boldsymbol\lambda^*)
- \text{BMDL}(\boldsymbol\lambda^0) \right]
\nonumber \\
= ~&
\frac{N-p}{N} 
\left[ 
\log \left( 
\frac{f(\delta^2_{\hat{\boldsymbol \lambda}_N}) + O_P\left(\frac{1}{\sqrt{N}}\right)}
{f(\delta^2_{\boldsymbol \lambda^*}) + O_P\left(\frac{1}{\sqrt{N}}\right)} 
\right) + 
\log \left( 
\frac{f(\delta^2_{\boldsymbol \lambda^*})+ O_P\left(\frac{1}{\sqrt{N}}\right)}
{f(\delta^2_{\boldsymbol \lambda^0})+ O_P\left(\frac{1}{\sqrt{N}}\right)} \right) 
\right].
\label{Burp}
\end{align}

Obviously, the term $N^{-1}(N-p)$ in (\ref{Burp}) converges to 
unity as $N \rightarrow \infty$.  The leftmost term in brackets in the 
bottom equation in (\ref{Burp}) converges to zero.  This follows from 
(\ref{dump}), the continuity of $f$ and the natural log function, and 
the fact that $\log(1)=0$. When $\boldsymbol \lambda^* \ne \boldsymbol 
\lambda^0$, since the number of changepoints in these two models are the 
same, $\boldsymbol\lambda^* \not\supset \boldsymbol\lambda^0$. 
Therefore, by \eqref{eq:delta_sq_average}, we have 
$\delta^2_{\boldsymbol \lambda^0} = 0$ and $\delta^2_{\boldsymbol 
\lambda^*} > 0$. The limit of the rightmost bracketed term in 
(\ref{Burp}) must be positive.  Positivity follows from 
$f(\delta^2_{\boldsymbol \lambda^*}) > f(\delta^2_{\boldsymbol 
\lambda^0})=\sigma^2$, which can be verified by an argument akin to that 
proving Proposition \ref{prop:sigmasq2}, the nondecreasing and 
continuous nature of $f$, that $f(0)= \sigma^2>0$, and that 
$P(\boldsymbol \lambda^* \ne \boldsymbol \lambda^0)> 0$.  The details 
are omitted; this said, one can get a flavor for the argument in the 
proof of the next result, which quantifies how much 
$\delta^2_{\boldsymbol \lambda}$ varies when elements of it are changed. 
This finishes our work. \end{proof}


Next, under the assumption that $m^0$ is unknown, we first establish the 
following convergence rate lemma on estimated changepoint locations 
$\hat{\lambda}_j$.

\begin{lemma}\label{lemma:lambda_convergence} Suppose that $m^0$ is 
unknown. Then for each $\lambda_r^0$, $r \in \{ 1, \ldots, m^0 \}$, 
there exists a $\hat{\lambda}_j$ in $\hat{\boldsymbol\lambda}_N$ such 
that
\begin{equation}
\label{eq:lambda_convergence_rate2}
\left| \hat{\lambda}_j - \lambda^0_r \right| = O_P(N^{- 1}).
\end{equation}
\end{lemma}

\begin{proof} By the spacing assumptions made on the changepoint 
configuration, there can be at most a finite number of changepoints.  
Using this and repeating the argument in the proof of Proposition 
\ref{prop:lambda_convergence_known_m0}, one can argue that the estimated 
changepoint model $\hat{\boldsymbol\lambda}_N$ in \eqref{eq:lambda_hat} 
converges to a limit $\boldsymbol\lambda^*$ that contains all 
changepoints in $\boldsymbol\lambda^0$; that is, $P(\boldsymbol\lambda^* 
\supset \boldsymbol\lambda^0 )=1$.  This means that for each 
$\lambda_r^0$, $r = 1, \ldots, m^0$, there exists a 
$\hat{\lambda}_{j(r),N}$ in $\hat{\boldsymbol\lambda}_N$ such that 
$\hat{\lambda}_{j(r),N} \stackrel {P} {\longrightarrow} \lambda_r^0$; 
that is, $| \hat{\lambda}_{j(r),N} - \lambda^0_r| = o_P(1)$. For 
notation simplicity, we rewrite $\hat{\lambda}_{j(r),N}$ as 
$\hat{\lambda}_{j}$ when there is no ambiguity.

The above shows that for all $r \in \{ 1, \ldots m^0 \}$, 
$|\hat{\lambda}_{j} - \lambda_r^0|= O_P(N^{\alpha_r-1})$ for some 
finite $\alpha_r$; in fact, we know that $\alpha_r \leq 1$.
Now let 
\begin{equation}
\label{eq:omega_r}
\omega_r = \inf \{ \alpha_r: |\hat{\lambda}_j- \lambda_r^0| =
O_P(N^{\alpha_r-1}) \}.
\end{equation}  
To prove the Lemma, we need to show that $\omega_r \leq 0$ for all $r$, 
or that $\omega \leq 0$ where
\begin{equation}
\label{eq:omega_def}
\omega \stackrel{\text{def}}{=} 
\max_{1\leq r \leq m^0} \omega_r.
\end{equation}
This will be done by contradiction.   Hence, suppose that $\omega > 0$, 
then there exist an $r$ such that 
\begin{equation}\label{eq:lambda_conv1}
\omega_r = \omega >0,\quad \text{and } |\hat{\lambda}_j - \lambda_r^0|= O_P(N^{\omega-1}).
\end{equation}
This will now be used to draw a contradiction.

For a sufficiently large $N$, a new model $\tilde{\boldsymbol\lambda}_N$ 
is created from $\hat{\boldsymbol\lambda}_N$ by replacing the 
changepoint $\hat{\lambda}_j$ in $\hat{\boldsymbol\lambda}_N$ with 
$\lambda^0_r$:
\[
\tilde{\boldsymbol\lambda}_N
= \left(\hat{\lambda}_1, \ldots, \hat{\lambda}_{j-1}, \lambda^0_r, \hat{\lambda}_{j+1}, \ldots, 
\hat{\lambda}_{\hat{m}}\right)'.
\] 
A contradiction occurs if $\text{BMDL}(\tilde{\boldsymbol\lambda}_N) < 
\text{BMDL}(\hat{\boldsymbol\lambda}_N)$ for all large $N$ since 
$\hat{\boldsymbol\lambda}_N$ minimizes the BMDL.

We first investigate the difference in $\hat{\gamma}(h)$ 
in \eqref{eq:autocov1} under the models $\hat{\boldsymbol\lambda}_N$ 
and $\tilde{\boldsymbol\lambda}_N$, for each $h \in \{ 0, 1, \ldots, 
p\}$. Following the argument in Proposition 
\ref{prop:pairwise_BMDL_OplogN},
\begin{equation}
\label{eq:autocov3}
\frac{1}{N} \sum_{t = h + 1}^N W_t W_{t-h} = 
\frac{1}{N} \sum_{t = h + 1}^N 
	\epsilon_t \epsilon_{t-h} + O_P\left( \frac{1}{N} \right)
\end{equation}
only depends on the observed data up to an $O_P(1/N)$ error.  Hence, its 
difference under the models $\hat{\boldsymbol\lambda}_N$ and 
$\tilde{\boldsymbol\lambda}_N$ is $O_P(1/N)$.

For the other terms in \eqref{eq:autocov1}, we need only focus on the 
summation over $t$ satisfying $\lfloor\hat{\lambda}_{j-1} N \rfloor \leq 
t \leq \lfloor \hat{\lambda}_{j+1} N \rfloor - 1$, depicted in Figure 
\ref{fig:diagram1}.  This is because $(W_t, \delta_t)$ for all $t$ 
elsewhere are identical in the models $\hat{\boldsymbol\lambda}_N$ and 
$\tilde{\boldsymbol\lambda}_N$. For notational simplicity, lengths of 
time intervals on the rescaled timeline are denoted by
\[
l_{r} = \lambda^0_{r} - \lambda^0_{r-1}, \quad 
l_{r+1} = \lambda^0_{r+1} - \lambda^0_{r}.
\]
We first consider the case where $\hat{\lambda}_{j-1}$ is to the left of 
$\lambda^0_{r-1}$ and $\hat{\lambda}_{j+1}$ is to the right of 
$\lambda^0_{r+1}$. Without loss of generality, we assume that 
$\hat{\lambda}_{j}$ is to the left of $\lambda^0_{r}$. The length 
between these estimated changepoints and their limits are denoted by
\begin{equation}
\label{eq:lambda_conv2}
\Delta l_{r-1} = \lambda^0_{r-1} - \hat{\lambda}_{j-1}, \quad
\Delta l_r = \lambda^0_r - \hat{\lambda}_j, \quad
\Delta l_{r+1} = \hat{\lambda}_{j+1} - \lambda^0_{r+1},
\end{equation}
all of which converge to zero at rates no slower than $O_P(N^{\omega - 1})$.

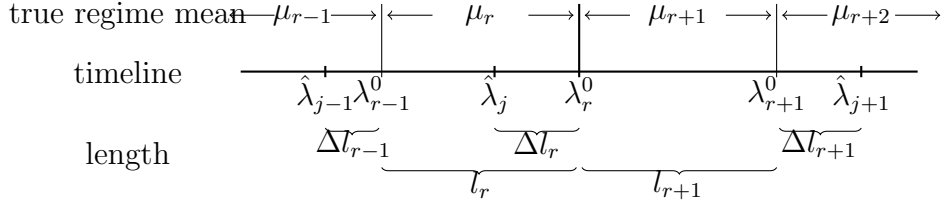
\begin{figure}
\centering
\begin{tikzpicture}[scale = 0.75]
\draw[thick]  (-1, 0) -- (11, 0);
\node at (-3, 1) {true regime mean};
\node at (-3, 0) {timeline};
\node at (-3, -1.50) {length};
\draw[thick] (0.5, -0.1) -- (0.5, 0.1); \node at (0.5, -0.4) {$\hat{\lambda}_{j-1}$};
\draw (1.5, -0.1) -- (1.5, 1.2); \node at (1.5, -0.4) {$\lambda_{r-1}^0$};
\draw[thick] (3.5, -0.1) -- (3.5, 0.1); \node  at (3.5, -0.4) {$\hat{\lambda}_{j}$};
\draw[thick] (5, -0.1) -- (5, 1.2); \node at (5, -0.4) {$\lambda^0_r$};
\draw (8.5, -0.1) -- (8.5, 1.2); \node at (8.5, -0.4) {$\lambda_{r+1}^0$};
\draw[thick] (10, -0.1) -- (10, 0.1); \node at (10, -0.4) {$\hat{\lambda}_{j+1}$};
\draw [->] (0.6, 1) -- (1.4, 1); \draw [<-] (-1.2, 1) -- (-0.4, 1); \node at (0.1, 1) {$\mu_{r-1}$};
\draw [<-] (1.6, 1) -- (2.4, 1); \draw [->] (4.1, 1) -- (4.9, 1); \node at (3.25, 1) {$\mu_{r}$};
\draw [<-] (5.1, 1) -- (5.9, 1); \draw [->] (7.6, 1) -- (8.4, 1); \node at (6.75, 1) {$\mu_{r+1}$};
\draw [->] (10.6, 1) -- (11.4, 1); \draw [<-] (8.6, 1) -- (9.4, 1); 
\node at (10, 1) {$\mu_{r + 2}$};
\draw [decoration={brace, mirror}, decorate] (0.5, -1.0) -- (1.45, -1.0);  \node at (1, -1.3) {$\Delta l_{r-1}$};
\draw [decoration={brace, mirror}, decorate] (3.5, -1.0) -- (5, -1.0);  \node at (4.25, -1.3) {$\Delta l_r$};
\draw [decoration={brace, mirror}, decorate] (8.55, -1.0) -- (10, -1.0);  \node at (9.25, -1.3) {$\Delta l_{r+1}$};
\draw [decoration={brace, mirror}, decorate] (1.5, -1.7) -- (4.95, -1.7);  \node at (3.25, -2) {$l_{r}$};
\draw [decoration={brace, mirror}, decorate] (5.05, -1.7) -- (8.5, -1.7);  \node at (6.75, -2) {$l_{r+1}$};
\draw[white] (-1, 2) -- (11, 2);
\end{tikzpicture}
\caption{Changepoints locations around time $\lambda^0_r$ for the proof of Lemma \ref{lemma:lambda_convergence}.
}\label{fig:diagram1}
\end{figure}

Under the model $\hat{\boldsymbol\lambda}_N$, $\delta_t$ in 
\eqref{eq:delta_t_W_t} can be written as
\begin{equation}
\label{eq:delta_t_lambdahat_proof1}
\delta_{\hat{\boldsymbol\lambda}_N, t} = 
\begin{cases}
\mu_{r-1} - \frac{\mu_{r-1} \Delta l_{r-1} + \mu_r (l_r - \Delta l_r)}{\Delta l_{r-1} + l_r -  \Delta l_r},	
	& \text{ if } \lfloor \hat{\lambda}_{j-1} N  \rfloor \leq t \leq \lfloor \lambda^0_{r-1} N \rfloor - 1,\\
\mu_r - \frac{\mu_{r-1} \Delta l_{r-1} + \mu_r (l_r - \Delta l_r)}{\Delta l_{r-1} + l_r -  \Delta l_r},	
	& \text{ if } \lfloor \lambda^0_{r-1} N  \rfloor \leq t \leq \lfloor \hat{\lambda}_j N \rfloor - 1,\\
\mu_r - \frac{\mu_r \Delta l_r + \mu_{r+1} l_{r+1} + \mu_{r+2} \Delta l_{r+1} }
	{\Delta l_r + l_{r+1} +  \Delta l_{r+1}},	
	& \text{ if } \lfloor \hat{\lambda}_j N  \rfloor \leq t \leq \lfloor \lambda^0_r N \rfloor - 1,\\
\mu_{r+1} - \frac{\mu_r \Delta l_r + \mu_{r+1} l_{r+1} + \mu_{r+2} \Delta l_{r+1} }
	{\Delta l_r + l_{r+1} +  \Delta l_{r+1}},	
	& \text{ if } \lfloor \lambda^0_r N  \rfloor \leq t \leq \lfloor \lambda^0_{r+1} N \rfloor - 1,\\
\mu_{r+2} - \frac{\mu_r \Delta l_r + \mu_{r+1} l_{r+1} + \mu_{r+2} \Delta l_{r+1} }
	{\Delta l_r + l_{r+1} +  \Delta l_{r+1}},	
	& \text{ if } \lfloor \lambda^0_{r+1} N  \rfloor \leq t \leq \lfloor \hat{\lambda}_{j+1} N \rfloor - 1;
\end{cases}
\end{equation}
whereas, under the model $\tilde{\boldsymbol\lambda}_N$,
\begin{equation}
\label{eq:delta_t_lambdatilde_proof1}
\delta_{\tilde{\boldsymbol\lambda}_N, t} = 
\begin{cases}
\mu_{r-1} - \frac{\mu_{r-1} \Delta l_{r-1} + \mu_r l_r }{\Delta l_{r-1} + l_r},	
	& \text{ if } \lfloor \hat{\lambda}_{j-1} N  \rfloor \leq t \leq \lfloor \lambda^0_{r-1} N \rfloor - 1,\\
\mu_r - \frac{\mu_{r-1} \Delta l_{r-1} + \mu_r l_r }{\Delta l_{r-1} + l_r},	
	& \text{ if } \lfloor \lambda^0_{r-1} N  \rfloor \leq t \leq  \lfloor \lambda^0_r N \rfloor - 1,\\
\mu_{r+1} - \frac{\mu_{r+1} l_{r+1} + \mu_{r+2} \Delta l_{r+1} }
	{ l_{r+1} +  \Delta l_{r+1}},	
	& \text{ if } \lfloor \lambda^0_r N  \rfloor \leq t \leq \lfloor \lambda^0_{r+1} N \rfloor - 1,\\
\mu_{r+2} - \frac{ \mu_{r+1} l_{r+1} + \mu_{r+2} \Delta l_{r+1} }
	{l_{r+1} +  \Delta l_{r+1}},	
	& \text{ if } \lfloor \lambda^0_{r+1} N  \rfloor \leq t \leq \lfloor \hat{\lambda}_{j+1} N \rfloor - 1.\\
\end{cases}
\end{equation}
When $N$ is large, $\delta_t = \delta_{t-h}$ for all but a finite number 
of times $t$; hence, for the second term (a similar argument applies to 
the third term) in \eqref{eq:autocov1},
\begin{align}
\label{eq:autocov4}
&	\frac{1}{N} \sum_{t = \lfloor \hat{\lambda}_{j-1} N \rfloor}^{\lfloor \hat{\lambda}_{j+1} N \rfloor  -1} 
	\delta_{t-h} W_t  \\ \nonumber
=~&	 
 	\frac{1}{N} \sum_{t = \lfloor \hat{\lambda}_{j-1} N \rfloor}^{\lfloor \hat{\lambda}_{j} N \rfloor - 1}
	\delta_{t} W_t 
	+ \frac{1}{N} \sum_{t = \lfloor \hat{\lambda}_{j} N \rfloor}^{\lfloor \lambda^0_{r} N \rfloor - 1} 
	\delta_{t} W_t 
	+ \frac{1}{N} \sum_{t = \lfloor \lambda^0_{r} N \rfloor}^{\lfloor \hat{\lambda}_{j+1} N \rfloor - 1}
	\delta_{t} W_t 
	+O_P\left( \frac{1}{N} \right).
\end{align}
By \eqref{eq:delta_t_lambdahat_proof1} and 
\eqref{eq:delta_t_lambdatilde_proof1}, under the two models 
$\hat{\boldsymbol\lambda}_N$ and $\tilde{\boldsymbol\lambda}_N$, the 
difference of $\delta_t$ is piecewise constant:
\begin{align}
\label{eq:delta_t_diff_proof1}
&\delta_{\hat{\boldsymbol\lambda}_N, t} - \delta_{\tilde{\boldsymbol\lambda}_N, t}\\ \nonumber
=~& \begin{cases}
\frac{(\mu_r - \mu_{r-1}) \Delta l_{r-1} \Delta l_r}
	{(\Delta l_{r-1} + l_r)(\Delta l_{r-1} + l_r -  \Delta l_r)}
	= O_P\left( N^{2\omega - 2}\right),\\
~~~~~~~~~~~~~~~~~~~~~~~~~~~~~~~~~~~~~~~~~~~~~~~~~~~~~~~~~~~
	\text{if } \lfloor \hat{\lambda}_{j-1} N  \rfloor \leq t \leq \lfloor \hat{\lambda}_j N \rfloor - 1,\\
\frac{(\mu_r - \mu_{r+1}) l_r  l_{r+1} + O_P(\Delta l)}
	{(\Delta l_{r-1} + l_r)(\Delta l_r + l_{r+1} +  \Delta l_{r+1})}
	= O_P\left( 1\right),\\
~~~~~~~~~~~~~~~~~~~~~~~~~~~~~~~~~~~~~~~~~~~~~~~~~~~~~~~~~~~
	\text{if } \lfloor \hat{\lambda}_j N  \rfloor \leq t \leq \lfloor \lambda^0_{r} N \rfloor - 1,\\
\frac{(\mu_{r+1} - \mu_r)\Delta l_r  l_{r+1} + (\mu_{r+2} - \mu_r)\Delta l_r \Delta l_{r+1}}
	{(l_{r+1} +  \Delta l_{r+1})(\Delta l_r + l_{r+1} +  \Delta l_{r+1})}
	= O_P\left( N^{\omega - 1}\right),\\
~~~~~~~~~~~~~~~~~~~~~~~~~~~~~~~~~~~~~~~~~~~~~~~~~~~~~~~~~~~
	\text{if } \lfloor \lambda^0_{r} N  \rfloor \leq t \leq \lfloor \hat{\lambda}_{j+1} N \rfloor - 1.
\end{cases}
\end{align}

To study the sum of $W_t$ in \eqref{eq:delta_t_W_t} over the above 
intervals, apply the central limit theorem for linear processes to see 
that $\bar{\epsilon}_{r(t)}, \bar{\epsilon}_{v(t)}, \bar{\epsilon}$ all 
converge to zero at the rate $O_P(1/\sqrt{N})$ for any $t$. Hence, for a 
$t \in [a,b]$ whose length $b-a$ depends on $N$ and is $O_P(N^{\xi})$ 
with $\xi \in (0, 1]$, the sums of $\epsilon_t$ and $W_t$ over this 
interval satisfy
\begin{equation*}
\sum_{t = a}^{b} \epsilon_t 
	= (b-a) \left(\frac{\sum_{t = a}^{b} \epsilon_t}{b-a}\right)
= O_P(N^{\xi})O_P\left( \frac{1}{\sqrt{N^{\xi}}} \right) = O_P(N^{\frac{\xi}{2}})
\end{equation*}
and
\begin{align}
\label{eq:sum_Wt}
\sum_{t = a}^{b} W_t 
=~& 	
\sum_{t = a}^{b} 
\left( \epsilon_t - \bar{\epsilon}_{r(t)} - \bar{\epsilon}_{v(t)} + \bar{\epsilon} \right)
	= 	\sum_{t = a}^{b} \epsilon_t 
	+ (b-a) O_P \left(\frac{1}{\sqrt{N}}\right) \\ \nonumber
=~& 	O_P(N^{\frac{\xi}{2}}) + O_P(N^{\xi - \frac{1}{2}}) \\ \nonumber
=~&     O_P(N^{\frac{\xi}{2}}), 
\end{align}
where the last equality follows from $\xi \leq 1$. For the three 
interval sums in \eqref{eq:delta_t_diff_proof1}, the corresponding 
convergence rates $\xi$ of their lengths are $1, \omega$, and $ 1$, 
respectively. Hence, in \eqref{eq:autocov4}, when decomposed as three 
sums in these intervals, differences under the models 
$\hat{\boldsymbol\lambda}_N$ and $\tilde{\boldsymbol\lambda}_N$ are thus
\begin{align} \nonumber
&	\frac{1}{N} \sum_{t = \lfloor \hat{\lambda}_{j-1} N \rfloor}^{\lfloor \hat{\lambda}_{j+1} N \rfloor - 1}  
	\left(\delta_{\hat{\boldsymbol\lambda}_N, t} - \delta_{\tilde{\boldsymbol\lambda}_N, t}\right) W_t 
	\\ \nonumber
=~& 	\frac{1}{N} \left\{
	O_P\left( N^{2\omega - 2} \right)~ O_P\left( N^{\frac{1}{2}} \right)	
	+ O_P\left( 1 \right)~ O_P\left( N^{\frac{\omega}{2}} \right)
	+ O_P\left( N^{\omega - 1} \right)~ O_P\left( N^{\frac{1}{2}} \right)
	\right\} \\ \nonumber
&+ O_P\left( N^{-1}\right)\\ \label{eq:autocov7}
=~& O_P\left( N^{\frac{\omega}{2}-1}\right),
\end{align}
where the last equality follows from $\omega \leq 1$. Therefore, the 
second and third term differences in \eqref{eq:autocov1} under the two 
models $\hat{\boldsymbol\lambda}_N$ and $\tilde{\boldsymbol\lambda}_N$ 
is $O_P\left( N^{\frac{\omega}{2}-1}\right)$.

For the last term in \eqref{eq:autocov1}, we similarly have
\[
\frac{1}{N} \sum_{t = \lfloor \hat{\lambda}_{j-1} N \rfloor}^{\lfloor \hat{\lambda}_{j+1} N \rfloor  -1} 
	\delta_{t-h} \delta_t 
= \frac{1}{N} \sum_{t = \lfloor \hat{\lambda}_{j-1} N \rfloor}^{\lfloor \hat{\lambda}_{j+1} N \rfloor  -1} 
	\delta_t^2 + O_P\left( \frac{1}{N} \right).
\]
Under the model $\hat{\boldsymbol\lambda}_N$,
\begin{align*}
&\frac{1}{N} \sum_{t = \lfloor \hat{\lambda}_{j-1} N \rfloor}^{\lfloor \hat{\lambda}_{j+1} N \rfloor  -1} 
	\delta_{\hat{\boldsymbol\lambda}_N, t}^2  \\
=~& 	\left(\mu_{r-1} - \frac{\mu_{r-1} \Delta l_{r-1} + \mu_r (l_r - \Delta l_r)}
	{\Delta l_{r-1} + l_r -  \Delta l_r}\right)^2 \Delta l_{r-1}\\ 
&	+\left(\mu_r - \frac{\mu_{r-1} \Delta l_{r-1} + \mu_r (l_r - \Delta l_r)}
	{\Delta l_{r-1} + l_r -  \Delta l_r}\right)^2 (l_r - \Delta l_r)\\ 
&	+ \left( \mu_r - \frac{\mu_r \Delta l_r + \mu_{r+1} l_{r+1} + \mu_{r+2} \Delta l_{r+1} }
	{\Delta l_r + l_{r+1} +  \Delta l_{r+1}} \right)^2 \Delta l_r \\ 
&	+ \left( \mu_{r+1} - \frac{\mu_r \Delta l_r + \mu_{r+1} l_{r+1} + \mu_{r+2} \Delta l_{r+1} }
	{\Delta l_r + l_{r+1} +  \Delta l_{r+1}} \right)^2 l_{r+1} \\ 
&	+ \left( \mu_{r+2} - \frac{\mu_r \Delta l_r + \mu_{r+1} l_{r+1} + \mu_{r+2} \Delta l_{r+1} }
	{\Delta l_r + l_{r+1} +  \Delta l_{r+1}} \right)^2 \Delta l_{r+1} \\ 	
=~&	 \frac{\left(\mu_r - \mu_{r-1}\right)^2 (l_r - \Delta l_r)  \Delta l_{r-1}}
	{\Delta l_{r-1} + l_r -  \Delta l_r}\\ 
&	+ \frac{(\mu_{r+1} - \mu_r)^2 \Delta l_r l_{r+1} + (\mu_{r+2} - \mu_r)^2 \Delta l_r \Delta l_{r+1}
	+  (\mu_{r+2} - \mu_{r+1})^2 \Delta l_{r+1} l_{r+1}}
	{\Delta l_r + l_{r+1} +  \Delta l_{r+1}}.
\end{align*}
On the other hand, under the model $\tilde{\boldsymbol\lambda}_N$, 
\begin{align*}
&\frac{1}{N} \sum_{t = \lfloor \hat{\lambda}_{j-1} N \rfloor}^{\lfloor \hat{\lambda}_{j+1} N \rfloor  -1} 
	\delta_{\tilde{\boldsymbol\lambda}_N, t}^2  \\
=~& \left(\mu_{r-1} - \frac{\mu_{r-1} \Delta l_{r-1} + \mu_r  l_r}
	{\Delta l_{r-1} + l_r}\right)^2 \Delta l_{r-1}
	+\left(\mu_r - \frac{\mu_{r-1} \Delta l_{r-1} + \mu_r l_r }
	{\Delta l_{r-1} + l_r}\right)^2 l_r \\ \nonumber
&	+ \left( \mu_{r+1} - \frac{\mu_{r+1} l_{r+1} + \mu_{r+2} \Delta l_{r+1} }
	{ l_{r+1} +  \Delta l_{r+1}} \right)^2  l_{r+1} \\ \nonumber
&	+ \left( \mu_{r+2} - \frac{\mu_{r+1} l_{r+1} + \mu_{r+2} \Delta l_{r+1} }
	{ l_{r+1} +  \Delta l_{r+1}} \right)^2  \Delta l_{r+1}\\ 	
=~&	\frac{\left(\mu_r - \mu_{r-1}\right)^2 l_r   \Delta l_{r-1}}
	{\Delta l_{r-1} + l_r }
	+ \frac{\left(\mu_{r+2} - \mu_{r+1}\right)^2 l_{r+1}   \Delta l_{r+1}}
	{\Delta l_{r+1} + l_{r+1} }.	
\end{align*}
The difference of the last term in \eqref{eq:autocov1} under the two models,
up to an $O_P(1/N)$ error, is thus
\begin{align} \nonumber
&\frac{1}{N} \sum_{t = \lfloor \hat{\lambda}_{j-1} N \rfloor}^{\lfloor \hat{\lambda}_{j+1} N \rfloor  -1} 
	\left( \delta_{\hat{\boldsymbol\lambda}_N, t}^2 -
	 \delta_{\tilde{\boldsymbol\lambda}_N, t}^2 \right)\\ \nonumber
=~&	-\frac{(\mu_r - \mu_{r-1})^2 \Delta l_{r-1}^2 \Delta l_r}
	{(\Delta l_{r-1} + l_r - \Delta l_r)(\Delta l_{r-1} + l_r)}
	- \frac{(\mu_{r+2} - \mu_{r+1})^2 \Delta l_r l_{r+1}   \Delta l_{r+1}}
	{(\Delta l_r + l_{r+1} +  \Delta l_{r+1}) (\Delta l_{r+1} + l_{r+1})} \\ \nonumber
&	+ \frac{(\mu_{r+1} - \mu_r)^2 \Delta l_r l_{r+1} }
	{\Delta l_r + l_{r+1} +  \Delta l_{r+1}}
	+ \frac{(\mu_{r+2} - \mu_r)^2 \Delta l_r \Delta l_{r+1}}
	{\Delta l_r + l_{r+1} +  \Delta l_{r+1}}\\ \label{eq:autocov8}
=~& 	(\mu_{r+1} - \mu_r)^2 \Delta l_r + o_P(\Delta l_r) = O_P\left( N^{\omega -1} \right).
\end{align}

Therefore, the difference of $\hat{\gamma}(h)$ 
\eqref{eq:autocov1} under the models $\hat{\boldsymbol\lambda}_N$ and 
$\tilde{\boldsymbol\lambda}_N$ is
\begin{align} 
\label{eq:autocov5}
\hat{\gamma}(h)_{\hat{\boldsymbol\lambda}_N} - \hat{\gamma}(h)_{\tilde{\boldsymbol\lambda}_N} 
&	= O_P(N^{-1}) + O_P(N^{\frac{\omega}{2}-1}) + O_P(N^{\omega -1}) 
	= O_P(N^{\omega -1}).
\end{align}
Here, the convergence rates of the three terms in the summation are 
given by the results shown in \eqref{eq:autocov3}, \eqref{eq:autocov7}, 
and \eqref{eq:autocov8}, respectively. Since $\omega > 0$, the third 
term in \eqref{eq:autocov5} dominates the overall convergence rate. Note 
that by \eqref{eq:autocov8}, this term has the same limit as $(\mu_{r+1} 
- \mu_r)^2 \Delta l_r$. Therefore, the limit of \eqref{eq:autocov5} 
remains the same across different value of $h \in \{ 0, 1, \ldots, p\}$.

By \eqref{eq:omega_r}, \eqref{eq:lambda_conv1}, and 
\eqref{eq:lambda_conv2}, $\Delta l_r$ is positive, and converges to zero 
in probability on the order of $O_P(N^{\omega-1})$, but not at any 
faster polynomial rate.  Since $\mu_{r+1} \neq \mu_r$, by 
\eqref{eq:autocov5}, for large $N$, 
$\hat{\gamma}(h)_{\hat{\boldsymbol\lambda}_N} - 
\hat{\gamma}(h)_{\tilde{\boldsymbol\lambda}_N}$ is also positive, 
converging to zero in probability on the order of $O_P(N^{\omega-1})$, 
but not any faster.

Following similar reasoning, if $\hat{\lambda}_{j}$ is to the right of 
$\lambda^0_{r}$, the result in \eqref{eq:autocov5} still holds. This 
conclusion does not change if $\hat{\lambda}_{j-1}$ is to the right of 
$\lambda^0_{r-1}$ (or $\hat{\lambda}_{j+1}$ is to the left of 
$\lambda^0_{r+1}$): we can simply take $\Delta l_{r-1} = 0$ (or $\Delta 
l_{r+1} = 0$) and all above derivations hold unaltered.
 
Next, we will show that for sufficiently large $N$, the model 
$\tilde{\boldsymbol\lambda}_N$ has a smaller BMDL than model 
$\hat{\boldsymbol\lambda}_N$.  Proposition \ref{prop:sigmasq2} shows 
that $f(\delta^2)$ in \eqref{eq:f} is strictly increasing in $\delta^2$. 
A similar argument applies here after replacing $\delta^2$ by the limit 
of \eqref{eq:autocov5}, which is $(\mu_{r+1} - \mu_r)^2 \Delta l_r$.  
This implies that the difference of the Yule-Walker estimators 
$\hat{\sigma}^2_{\text{YW}}$ in \eqref{eq:sigmasq_hat_YW} under the 
models $\hat{\boldsymbol\lambda}_N$ and $\tilde{\boldsymbol\lambda}_N$ 
obeys
\[
\hat{\sigma}^2_{\hat{\boldsymbol\lambda}_N, YW}
- \hat{\sigma}^2_{\tilde{\boldsymbol\lambda}_N, YW}
=  O_P(N^{\omega - 1}).
\]
Furthermore, this difference is positive and converges to zero in 
probability on the order of $O_P(N^{\omega-1})$, but not at any faster 
polynomial rate.  By Lemma \ref{lemma:sigmasq1}, the BMDL estimator 
$\hat{\sigma}^2 = \hat{\sigma}^2_{\text{YW}} + O_P(1/N)$, thus, the 
difference of the BMDL estimator $\hat{\sigma}^2$ under the two models 
satisfies
\begin{equation}
\label{eq:sigmasq_diff}
\hat{\sigma}^2_{\hat{\boldsymbol\lambda}_N}
- \hat{\sigma}^2_{\tilde{\boldsymbol\lambda}_N}
=  O_P(N^{\omega - 1}) + O_P(1/N) = O_P(N^{\omega - 1}), 
\end{equation}
the last equality stemming from $\omega > 0$. This shows that 
\eqref{eq:sigmasq_diff} is dominated by 
$\hat{\sigma}^2_{\hat{\boldsymbol\lambda}_N, YW} - 
\hat{\sigma}^2_{\tilde{\boldsymbol\lambda}_N, YW}$, and thus is positive 
and converges to zero in probability on the order of $O_P(N^{\omega-1})$ 
(but not at any faster polynomial rate). Since $\omega > 0 $, $\left( 
\hat{\sigma}^2_{\hat{\boldsymbol\lambda}_N} - 
\hat{\sigma}^2_{\tilde{\boldsymbol\lambda}_N} \right) / 
N^{\frac{\omega}{2} -1}$ diverges in probability, i.e., for a strictly 
positive constant $C$, when $N$ is large enough,
\[
\frac{\hat{\sigma}^2_{\hat{\boldsymbol\lambda}_N}
- \hat{\sigma}^2_{\tilde{\boldsymbol\lambda}_N}}{N^{\frac{\omega}{2} -1}}
\geq C.
\] 

Recall that the model $\tilde{\boldsymbol\lambda}_N$ contains the 
same number of changepoints as the model $\hat{\boldsymbol\lambda}_N$; 
therefore,

\begin{align*} \nonumber
	\text{BMDL}(\hat{\boldsymbol\lambda}_N) 
	-\text{BMDL}(\tilde{\boldsymbol\lambda}_N) 
&=      ~\frac{N-p}{2}\log \left(\frac{\hat{\sigma}^2_{\hat{\boldsymbol\lambda}_N}}
	{\hat{\sigma}^2_{\tilde{\boldsymbol\lambda}_N}}\right)
	+ O_P\left(1\right)	\\ \nonumber
&=      ~\frac{N}{2}\log \left(\frac{\hat{\sigma}^2_{\hat{\boldsymbol\lambda}_N}}
	{\hat{\sigma}^2_{\tilde{\boldsymbol\lambda}_N}}\right)
	+ O_P\left(1\right)	\\ \nonumber
&=      ~\frac{N}{2}\log \left(1 + \frac{\hat{\sigma}^2_{\hat{\boldsymbol\lambda}_N}
	-\hat{\sigma}^2_{\tilde{\boldsymbol\lambda}_N}}
	{\hat{\sigma}^2_{\tilde{\boldsymbol\lambda}_N}}\right)
	+ O_P\left(1\right)	\\ \nonumber
& \geq  ~\frac{N}{2}\log \left(1 + \frac{C}
	{\hat{\sigma}^2_{\tilde{\boldsymbol\lambda}_N} N^{1 - \frac{\omega}{2}} }\right)
	+ O_P\left(1\right)	\\ \nonumber
&=	\frac{N^{\frac{\omega}{2}}}{2}	\log \left(1 + \frac{C}
	{\hat{\sigma}^2_{\tilde{\boldsymbol\lambda}_N} N^{1 - \frac{\omega}{2}} }\right)^{N^{1-\frac{\omega}{2}}}
	+ O_P\left(1\right) \\ \nonumber
&=	\frac{N^{\frac{\omega}{2}}}{2}	\frac{C}
	{\hat{\sigma}^2_{\tilde{\boldsymbol\lambda}_N}} + O_P\left(1\right),
\end{align*}
where the last equality follows from $\lim_{ N \rightarrow \infty}(1 + 
\frac{x}{N})^N \rightarrow e^x$ and $\omega \leq 1$. Hence, 
$\text{BMDL}(\hat{\boldsymbol\lambda}_N) - 
\text{BMDL}(\tilde{\boldsymbol\lambda}_N)$ diverges to infinity at rate 
$O_P(N^{\frac{\omega}{2}})$ or faster, should $\omega > 0$. Here, a 
contradiction arises since $\hat{\boldsymbol\lambda}_N$ minimizes the 
BMDL. \end{proof}


In Theorem \ref{thm:lambda_convergence}, the convergence rate in 
\eqref{eq:lambda_convergence_rate} comes from Lemma 
\ref{lemma:lambda_convergence}. Now the proof of 
\eqref{eq:lambda_convergence} is given.

\begin{proof} [A proof of \eqref{eq:lambda_convergence} in Theorem 
\ref{thm:lambda_convergence}] In the proof of Lemma 
\ref{lemma:lambda_convergence}, $\boldsymbol\lambda^* \supset 
\boldsymbol\lambda^0$. To verify \eqref{eq:lambda_convergence}, we need 
only show that $\boldsymbol\lambda^* = \boldsymbol\lambda^0$; in other 
words, there are no changepoints in $\boldsymbol\lambda^*$ that are not 
in $\boldsymbol\lambda^0$.

Proof by contradiction will again be used. Suppose that for a large $N$, 
the BMDL estimator $\hat{\boldsymbol\lambda}_N$ contains more than $m^0$ 
changepoints. More specifically, suppose that during the $(r+1)$th 
regime in the true model $\boldsymbol\lambda^0$, there are redundant 
changepoints estimated in $\hat{\boldsymbol\lambda}_N$, i.e., for some 
integer $d > 1$, 
\[ 
\hat{\lambda}_j \stackrel {P} {\longrightarrow} \lambda_r^0, \quad 
\hat{\lambda}_{j+d} \stackrel {P} {\longrightarrow} \lambda_{r+1}^0, 
\] 
where $\hat{\lambda}_j$ can be to the left or right of $\lambda_r^0$, 
and $\hat{\lambda}_{j+d}$ can be to the left or right of 
$\lambda_{r+1}^0$. Since the estimated changepoints 
$\hat{\lambda}_{j+1}, \ldots, \hat{\lambda}_{j+d-1}$ are redundant, a 
new relative multiple changepoint model
\[ 
\tilde{\boldsymbol\lambda}_N 
= \left(\hat{\lambda}_1, \ldots, \hat{\lambda}_j, \hat{\lambda}_{j+d}, 
\ldots, \hat{\lambda}_{\hat{m}}\right)' 
\] 
is created by removing the redundant changepoints $\hat{\lambda}_{j+1}, 
\ldots, \hat{\lambda}_{j+d-1}$ from $\hat{\boldsymbol\lambda}_N$. A 
contradiction would arise if $\text{BMDL}(\hat{\boldsymbol\lambda}_N)> 
\text{BMDL}( \tilde{\boldsymbol\lambda}_N)$ for large $N$ since 
$\hat{\boldsymbol\lambda}_N$ minimizes the BMDL.

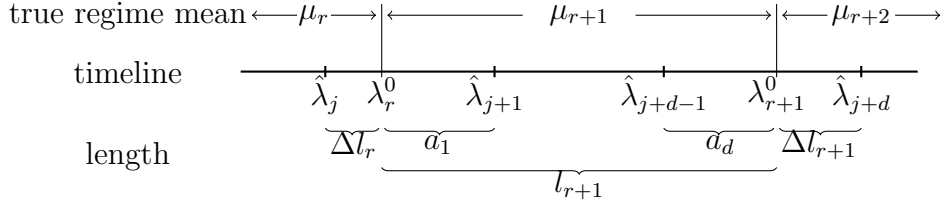
\begin{figure}
\centering
\begin{tikzpicture}[scale = 0.75]
\draw[thick]  (-1, 0) -- (11, 0);
\node at (-3, 1) {true regime mean};
\node at (-3, 0) {timeline};
\node at (-3, -1.50) {length};
\draw[thick] (0.5, -0.1) -- (0.5, 0.1); \node at (0.5, -0.4) {$\hat{\lambda}_j$};
\draw (1.5, -0.1) -- (1.5, 1.2); \node at (1.5, -0.4) {$\lambda_r^0$};
\draw[thick] (3.5, -0.1) -- (3.5, 0.1); \node  at (3.5, -0.4) {$\hat{\lambda}_{j + 1}$};
\draw[thick] (6.5, -0.1) -- (6.5, 0.1); \node at (6.5, -0.4) {$\hat{\lambda}_{j + d - 1}$};
\draw (8.5, -0.1) -- (8.5, 1.2); \node at (8.5, -0.4) {$\lambda_{r+1}^0$};
\draw[thick] (10, -0.1) -- (10, 0.1); \node at (10, -0.4) {$\hat{\lambda}_{j+d}$};
\draw [->] (0.6, 1) -- (1.4, 1); \draw [<-] (-0.8, 1) -- (0, 1); \node at (0.3, 1) {$\mu_{r}$};
\draw [<-] (1.6, 1) -- (4, 1); \draw [->] (6, 1) -- (8.4, 1); \node at (5, 1) {$\mu_{r+1}$};
\draw [->] (10.6, 1) -- (11.4, 1); \draw [<-] (8.6, 1) -- (9.4, 1); 
\node at (10, 1) {$\mu_{r + 2}$};
\draw [decoration={brace, mirror}, decorate] (0.5, -1.0) -- (1.45, -1.0);  \node at (1, -1.3) {$\Delta l_r$};
\draw [decoration={brace, mirror}, decorate] (1.55, -1.0) -- (3.5, -1.0);  \node at (2.5, -1.3) {$a_1$};
\draw [decoration={brace, mirror}, decorate] (6.5, -1.0) -- (8.45, -1.0);  \node at (7.5, -1.3) {$a_d$};
\draw [decoration={brace, mirror}, decorate] (8.55, -1.0) -- (10, -1.0);  \node at (9.25, -1.3) {$\Delta l_{r+1}$};
\draw [decoration={brace, mirror}, decorate] (1.5, -1.7) -- (8.5, -1.7);  \node at (5, -2) {$l_{r+1}$};
\draw[white] (-1, 2) -- (11, 2);
\end{tikzpicture}
\caption{Changepoint locations around the $(r+1)$th regime in the true changepoint model 
for the proof of \eqref{eq:lambda_convergence} in Theorem \ref{thm:lambda_convergence}.}\label{fig:diagram2}
\end{figure}

Similar to the proof of Lemma \ref{lemma:lambda_convergence}, the 
difference of $\hat{\gamma}(h)$ \eqref{eq:autocov1} under the two models 
$\hat{\boldsymbol\lambda}_N$ and $\tilde{\boldsymbol\lambda}_N$ will be 
investigated for each $h \in \{ 0, 1, \ldots, p \}$. By 
\eqref{eq:autocov3}, the first term in \eqref{eq:autocov1} is the same 
under $\hat{\boldsymbol\lambda}_N$ and $\tilde{\boldsymbol\lambda}_N$, 
up to a $O_P(1/N)$ difference.
 
For the other terms in \eqref{eq:autocov1}, we need only focus on the 
summation over $t$ in the interval $\lfloor \hat{\lambda}_j N \rfloor 
\leq t \leq \lfloor \hat{\lambda}_{j+d} N \rfloor - 1$, illustrated in 
Figure \ref{fig:diagram2}, since $(W_t, \delta_t)$ are the same for all 
other $t$ in $\hat{\boldsymbol\lambda}_N$ and 
$\tilde{\boldsymbol\lambda}_N$. For simplicity, lengths of time 
intervals on the rescaled timeline are denoted by
\[
l_{r+1} = \lambda^0_{r+1} - \lambda^0_{r}, \quad 
a_1 = \hat{\lambda}_{j+1} - \lambda^0_{r}, \quad
a_d = \lambda^0_{r} - \hat{\lambda}_{j+d-1}.
\]
If $\hat{\lambda}_j$ is to the left of $\lambda^0_r$ and 
$\hat{\lambda}_{j+d}$ is to the right of $\lambda^0_{r+1}$ (see Figure 
\ref{fig:diagram2}), then the vanishing length between them and their 
limits are denoted by
\[
\Delta l_r = \lambda^0_r - \hat{\lambda}_j, \quad
\Delta l_{r+1} = \hat{\lambda}_{j+d} - \lambda^0_{r+1},
\]
both of which converge to zero at rates no slower than $O_P(N^{\omega - 
1})$, where $\omega$ is defined in \eqref{eq:omega_def}.

Under the model $\hat{\boldsymbol\lambda}_N$, $\delta_t$ in 
\eqref{eq:delta_t_W_t} can be written as
\begin{equation}
\label{eq:delta_t_lambdahat}
\delta_{\hat{\boldsymbol\lambda}_N, t} = 
\begin{cases}
\mu_r - \frac{\mu_r \Delta l_r + \mu_{r+1} a_1}{\Delta l_r + a_1},	
	& \text{ if } \lfloor \hat{\lambda}_j N  \rfloor \leq t \leq \lfloor \lambda^0_r N \rfloor - 1,\\
\mu_{r+1} - \frac{\mu_r \Delta l_r + \mu_{r+1} a_1}{\Delta l_r + a_1},	
	& \text{ if } \lfloor \lambda^0_r N  \rfloor \leq t \leq \lfloor \hat{\lambda}_{j+1} N \rfloor - 1,\\
0,	
	& \text{ if } \lfloor \hat{\lambda}_{j+1} N  \rfloor \leq t \leq \lfloor \hat{\lambda}_{j+d-1} N \rfloor - 1,\\
 \mu_{r+1}-\frac{\mu_{r+2} \Delta l_{r+1} + \mu_{r+1} a_d}{\Delta l_{r+1} + a_d},	
	& \text{ if } \lfloor \hat{\lambda}_{j+d-1} N  \rfloor \leq t \leq \lfloor \lambda^0_{r+1} N \rfloor - 1,\\
\mu_{r+2}-\frac{\mu_{r+2} \Delta l_{r+1} + \mu_{r+1} a_d}{\Delta l_{r+1} + a_d},	
	& \text{ if } \lfloor \lambda^0_{r+1} N  \rfloor \leq t \leq \lfloor \hat{\lambda}_{j+d} N \rfloor - 1.
\end{cases}
\end{equation}
On the other hand, under the model $\tilde{\boldsymbol\lambda}_N$, 
\begin{equation}\label{eq:delta_t_lambdatilde}
\delta_{\tilde{\boldsymbol\lambda}_N, t} = 
\begin{cases}
\mu_r - \frac{\mu_r \Delta l_r + \mu_{r+1} l_{r+1} + \mu_{r+2} \Delta l_{r+1}}
	{\Delta l_r + l_{r+1} + \Delta l_{r+1}},	
& \text{ if } \lfloor \hat{\lambda}_j N  \rfloor \leq t \leq \lfloor \lambda^0_r N \rfloor - 1,\\
\mu_{r+1} - \frac{\mu_r \Delta l_r + \mu_{r+1} l_{r+1} + \mu_{r+2} \Delta l_{r+1}}
	{\Delta l_r + l_{r+1} + \Delta l_{r+1}},	
& \text{ if } \lfloor \lambda^0_r N  \rfloor \leq t \leq \lfloor \lambda^0_{r+1} N \rfloor - 1,\\
\mu_{r+2} - \frac{\mu_r \Delta l_r + \mu_{r+1} l_{r+1} + \mu_{r+2} \Delta l_{r+1}}
	{\Delta l_r + l_{r+1} + \Delta l_{r+1}},	
& \text{ if } \lfloor \lambda^0_{r+1} N  \rfloor \leq t \leq \lfloor \hat{\lambda}_{j+d} N \rfloor - 1.
\end{cases}
\end{equation}

When $N$ is large, $\delta_t = \delta_{t-h}$ for all but a finite number 
of times $t$; hence, for the second term (and similarly, the third term) 
in \eqref{eq:autocov1},
\begin{align}
\label{eq:autocov6}
&	\frac{1}{N} \sum_{t = \lfloor \hat{\lambda}_j N \rfloor}^{\lfloor \hat{\lambda}_{j+d} N \rfloor  -1} 
	\delta_{t-h} W_t  \\ \nonumber
=~&	\frac{1}{N} \sum_{t = \lfloor \hat{\lambda}_j N \rfloor}^{\lfloor \hat{\lambda}_{j+1} N \rfloor - 1} 
	\delta_{t} W_t 
+ \frac{1}{N} \sum_{t = \lfloor \hat{\lambda}_{j+1} N \rfloor}^{\lfloor \hat{\lambda}_{j+d-1} N \rfloor - 1} 
	\delta_{t} W_t 
+ \frac{1}{N} \sum_{t = \lfloor \hat{\lambda}_{j+d-1} N \rfloor}^{\lfloor \hat{\lambda}_{j+d} N \rfloor - 1} 
	\delta_{t} W_t 
	+O_P\left( \frac{1}{N} \right).
\end{align}
By \eqref{eq:delta_t_lambdahat} and \eqref{eq:delta_t_lambdatilde}, 
under the models $\hat{\boldsymbol\lambda}_N$ and 
$\tilde{\boldsymbol\lambda}_N$, the difference of $\delta_t$ is 
piecewise constant, i.e.,
\begin{align}
\label{eq:delta_t_diff} 
&	\delta_{\hat{\boldsymbol\lambda}_N, t} - \delta_{\tilde{\boldsymbol\lambda}_N, t}\\ \nonumber
=~&	
\begin{cases}
\frac{\left(\mu_{r+1} -  \mu_r\right) \Delta l_r (l_{r+1} - a_1)
	+ (\mu_{r+2} - \mu_{r+1})\Delta l_{r+1} a_1 + O_P(\Delta l^2) }
	{\left(\Delta l_r + a_1\right) \left( \Delta l_r + l_{r+1} + \Delta l_{r+1} \right)} 
= O_P\left( N^{\omega - 1}\right),\\
~~~~~~~~~~~~~~~~~~~~~~~~~~~~~~~~~~~~~~~~~~~~~~~~~~~~~~~~~~~
	\text{ if } \lfloor \hat{\lambda}_{j} N  \rfloor \leq t \leq \lfloor \hat{\lambda}_{j+1} N \rfloor - 1,\\
\frac{\left(\mu_{r+1} -  \mu_r\right) \Delta l_r 
	+ (\mu_{r+2} - \mu_{r+1})\Delta l_{r+1} }
	{\Delta l_r + l_{r+1} + \Delta l_{r+1}} 
= O_P\left( N^{\omega - 1}\right),\\
~~~~~~~~~~~~~~~~~~~~~~~~~~~~~~~~~~~~~~~~~~~~~~~~~~~~~~~~~~~
	\text{ if } \lfloor \hat{\lambda}_{j+1} N  \rfloor \leq t \leq \lfloor \hat{\lambda}_{j+d-1} N \rfloor - 1,\\
\frac{\left(\mu_{r+1} -  \mu_{r+2}\right) \Delta l_{r+1} (l_{r+1} - a_d)
	+ (\mu_r - \mu_{r+1})\Delta l_r a_d + O_P(\Delta l^2)}
	{(\Delta l_{r+1} + a_d)(\Delta l_r + l_{r+1} +  \Delta l_{r+1})}
	= O_P\left( N^{\omega - 1}\right),\\
~~~~~~~~~~~~~~~~~~~~~~~~~~~~~~~~~~~~~~~~~~~~~~~~~~~~~~~~~~~
	\text{ if } \lfloor \hat{\lambda}_{j+d-1} N  \rfloor \leq t \leq \lfloor \hat{\lambda}_{j+d} N \rfloor - 1.
\end{cases}
\end{align}
For the three time intervals in \eqref{eq:delta_t_diff}, their lengths 
are $\Delta l_r + a_1 = O_P(N^{\xi_1})$, $l_{r+1} - a_1 - a_d  = O_P(1)$, 
and $a_d + \Delta l_{r+1} = O_P(N^{\xi_d})$, 
respectively, with $\xi_1, \xi_d \in [\omega, 1]$. For 
\eqref{eq:autocov6}, when decomposed as three sums in these intervals, 
by \eqref{eq:sum_Wt}, its difference under the models 
$\hat{\boldsymbol\lambda}_N$ and $\tilde{\boldsymbol\lambda}_N$ is
\begin{align*}
&\frac{1}{N} \sum_{t = \lfloor \hat{\lambda}_{j} N \rfloor}^{\lfloor \hat{\lambda}_{j+d} N \rfloor - 1}  
	\left(\delta_{\hat{\boldsymbol\lambda}_N, t} - \delta_{\tilde{\boldsymbol\lambda}_N, t}\right) W_t \\
=~& 	\frac{1}{N} O_P\left( N^{\omega - 1} \right) \left\{
	O_P\left( N^{\frac{\xi_1}{2}} \right)	
	+ O_P\left( N^{\frac{1}{2}} \right)	+ O_P\left( N^{\frac{\xi_d}{2}} \right)	
	\right\} + O_P\left( N^{-1}\right)\\
=~& 	O_P\left( N^{\omega -\frac{3}{2}}\right) + O_P\left( N^{-1}\right). 
\end{align*}
By Lemma \ref{lemma:lambda_convergence}, $\omega \leq 0$; hence, for the 
second term (and similarly for the third term) in \eqref{eq:autocov1}, its 
difference under the two models converges to zero at rate $O_P(1/N)$.

For the fourth term in \eqref{eq:autocov1}, since
\[
\frac{1}{N} \sum_{t = \lfloor \hat{\lambda}_{j} N \rfloor}^{\lfloor \hat{\lambda}_{j+d} N \rfloor - 1} 
	\delta_{t-h} \delta_t 
= \frac{1}{N} \sum_{t = \lfloor \hat{\lambda}_{j} N \rfloor}^{\lfloor \hat{\lambda}_{j+d} N \rfloor - 1} 
	\delta_t^2 + O_P\left( \frac{1}{N} \right),
\]
under the model $\hat{\boldsymbol\lambda}_N$, it can be written as
\begin{align} \nonumber
&	\frac{1}{N} \sum_{t = \lfloor \hat{\lambda}_{j} N \rfloor}^{\lfloor \hat{\lambda}_{j+d} N \rfloor - 1}
	\delta_{\hat{\boldsymbol\lambda}_N, t}^2  \\ \nonumber
=~& \left(\mu_r - \frac{\mu_r \Delta l_r + \mu_{r+1} a_1}{\Delta l_r + a_1}\right)^2 \Delta l_r
	+ \left( \mu_{r+1} - \frac{\mu_r \Delta l_r + \mu_{r+1} a_1}{\Delta l_r + a_1} \right)^2 a_1\\ \nonumber
&	+ \left( \mu_{r+1} - \frac{\mu_{r+2} \Delta l_{r+1} + \mu_{r+1} a_d}{\Delta l_{r+2} + a_d} 
	\right)^2 a_d
+ \left( \mu_{r+2} - \frac{\mu_{r+2} \Delta l_{r+1} + \mu_{r+1} a_d}{\Delta l_{r+2} + a_d} 
	 \right)^2 \Delta l_{r+1}\\ \label{eq:deltasq_lambdahat}
=~&	\frac{(\mu_{r+1} - \mu_r)^2 a_1 \Delta l_r}{a_1 + \Delta l_r}
	+ \frac{(\mu_{r+2} - \mu_{r+1})^2 a_d \Delta l_{r+1}}{a_d + \Delta l_{r+1}},
\end{align}
whereas under the model $\tilde{\boldsymbol\lambda}_N$, 
\begin{align}\nonumber
&	\frac{1}{N} \sum_{t = \lfloor \hat{\lambda}_{j} N \rfloor}^{\lfloor \hat{\lambda}_{j+d} N \rfloor - 1}
	\delta_{\tilde{\boldsymbol\lambda}_N, t}^2  \\ \nonumber
=~& 	 \left( \mu_r  - \frac{\mu_r \Delta l_r + \mu_{r+1} l_{r+1} + \mu_{r+2} \Delta l_{r+1}}
	{\Delta l_r + l_{r+1} + \Delta l_{r+1}}\right)^2 \Delta l_r\\ \nonumber
&	+ 	\left( \mu_{r+1} - \frac{\mu_r \Delta l_r + \mu_{r+1} l_{r+1} + \mu_{r+2} \Delta l_{r+1}}
	{\Delta l_r + l_{r+1} + \Delta l_{r+1}} \right)^2 l_{r+1}\\ \nonumber
&	+ 	\left( \mu_{r+2} - \frac{\mu_r \Delta l_r + \mu_{r+1} l_{r+1} + \mu_{r+2} \Delta l_{r+1}}
	{\Delta l_r + l_{r+1} + \Delta l_{r+1}} \right)^2 \Delta l_{r+1}\\ \label{eq:deltasq_lambdatilde}
=~& 	 \frac{(\mu_{r+1} - \mu_r)^2 l_{r+1} \Delta l_r
	+(\mu_{r+2} - \mu_r)^2 \Delta l_r  \Delta l_{r+1}
	+(\mu_{r+2} - \mu_{r+1})^2 l_{r+1} \Delta l_{r+1} }{\Delta l_r + l_{r+1} + \Delta l_{r+1}}.
\end{align}
Since $\Delta l_r = O_P\left(N^{\omega-1}\right)$ and $\Delta l_{r+1} = 
O_P\left(N^{\omega-1}\right)$, where $\omega \leq 0$, both 
\eqref{eq:deltasq_lambdahat} and \eqref{eq:deltasq_lambdatilde} converge 
to zero at rate $O_P(N^{\omega-1})$. Hence, the difference of the fourth 
term in \eqref{eq:autocov1} converges to zero at rate $O_P(1/N)$.

The difference in $\hat{\gamma}(h)$ in \eqref{eq:autocov1} under the two 
models $\hat{\boldsymbol\lambda}_N$ and $\tilde{\boldsymbol\lambda}_N$ 
thus satisfies
\[
\hat{\gamma}_{\hat{\boldsymbol\lambda}_N}(h) = \hat{\gamma}_{\tilde{\boldsymbol\lambda}_N}(h)
+ O_P\left( \frac{1}{N} \right),
\]
which holds for all $h \in \{ 0, 1, \ldots, p \}$. By Lemma 
\ref{lemma:sigmasq1}, a similar result holds for the BMDL estimators of 
$\sigma^2$ under the two models $\hat{\boldsymbol\lambda}_N$ and 
$\tilde{\boldsymbol\lambda}_N$:
\begin{equation}
\label{eq:diff_sigmasq_lambdahat_lambdatilde}
\hat{\sigma}^2_{\hat{\boldsymbol\lambda}_N} = 
\hat{\sigma}^2_{\tilde{\boldsymbol\lambda}_N} + O_P\left( \frac{1}{N} \right).
\end{equation}
Note that if $\hat{\lambda}_j$ is to the right of $\lambda^0_r$ (or 
$\hat{\lambda}_{j+d}$ is to the left of $\lambda^0_{r+1}$), then we 
simply let $\Delta l_r = 0$ (or $\Delta l_{r+1} = 0$), so that all above 
derivations, including \eqref{eq:delta_t_lambdahat} and 
\eqref{eq:delta_t_lambdatilde}, and more importantly, 
\eqref{eq:diff_sigmasq_lambdahat_lambdatilde} hold as stated.

The difference between $\text{BMDL}(\hat{\boldsymbol\lambda}_N)$ and 
$\text{BMDL}(\tilde{\boldsymbol\lambda}_N)$ will now be studied. Recall 
that $\hat{\boldsymbol\lambda}_N$ has $d - 1$ more changepoints than 
$\tilde{\boldsymbol\lambda}_N$. By \eqref{eq:det2} and 
\eqref{eq:Gamma_pairdiff}, the BMDL difference is
\begin{align*}
	\text{BMDL}(\hat{\boldsymbol\lambda}_N) 
	-\text{BMDL}(\tilde{\boldsymbol\lambda}_N)
=~&   \frac{N-p}{2}\log \left(\frac{\hat{\sigma}^2_{\hat{\boldsymbol\lambda}_N} }
	{\hat{\sigma}^2_{\tilde{\boldsymbol\lambda}_N} }\right)
	+ \frac{3(d - 1)}{2}  \log (N) + O_P\left(1\right)	\\ \label{eq:diff_MDL2}
=~&    O_P\left(1\right)
	+ \frac{3(d-1)}{2} \log (N) + O_P\left(1\right) \\
=~&      O_P(\log N), 
\end{align*} 
and is positive. Here, the second equality follows from 
\eqref{eq:diff_sigmasq_lambdahat_lambdatilde}. This contradicts that 
$\hat{\boldsymbol\lambda}_N$ minimizes the BMDL. 
\end{proof}

\subsection{Proof of Theorem \ref{thm:parameter_convergence}}
\label{PROOFthm:parameter_convergence}

\begin{proof} [A proof of Theorem \ref{thm:parameter_convergence}] By Theorem 
\ref{thm:lambda_convergence}, as $N$ tends to infinity, 
$\hat{\boldsymbol\lambda}_N \stackrel{P}{\longrightarrow} 
\boldsymbol\lambda^0$, and hence $\delta^2_{\hat{\boldsymbol\lambda}_N} 
\stackrel{P}{\longrightarrow} 0$. Therefore, by Proposition 
\ref{prop:YW_phi}, the BMDL estimator
\[
\hat{\boldsymbol\phi}_N = 
\left(\boldsymbol\Gamma_p + \delta^2_{\hat{\boldsymbol\lambda}_N} \mathbf{J}_p \right)^{-1}
\left(\boldsymbol\gamma_p + \delta^2_{\hat{\boldsymbol\lambda}_N} \mathbf{1}_p \right)
+ O_P\left(\frac{1}{\sqrt{N}}\right)
\stackrel{P}{\longrightarrow} 
\boldsymbol\Gamma_p^{-1} \boldsymbol\gamma_p = \boldsymbol\phi^0.
\]
By \eqref{eq:f}, when $\delta = 0$, $f(0) = \gamma(0) - 
\boldsymbol\gamma_p'\boldsymbol\Gamma_p^{-1} \boldsymbol\gamma_p = 
\left(\sigma^2\right)^0$, i.e., the true value of $\sigma^2$. Since 
$f(\delta^2)$ is continuous in $\delta^2$, Proposition 
\ref{prop:sigmasq2} shows that as $N \rightarrow \infty$, the BMDL 
estimator
\[
\hat{\sigma}^2_N \stackrel{P}{\longrightarrow} f(0) = \left(\sigma^2\right)^0.
\]

For sufficiently large $N$, since $\hat{\boldsymbol\lambda}_N$ is close 
to the true model $\boldsymbol\lambda^0$, the regime indicator matrix 
$\mathbf{D}$ under $\hat{\boldsymbol\lambda}_N$ is close to its 
counterpart $\mathbf{D}^0$ under the true model. Therefore, 
\eqref{eq:likelihood3} implies that
\begin{equation}
\label{eq:likelihood4}
\widehat{\mathbf{X}}= \widehat{\mathbf{A}} \mathbf{s} + 
\widehat{\mathbf{D}}\boldsymbol\mu + \hat{\mathbf{z}},
\end{equation}
where $\hat{\mathbf{z}} = (\hat{z}_{p+1}, \ldots, 
\hat{z}_N)'$, and $\hat{z}_t = \epsilon_t - \sum_{j=1}^p 
\hat{\phi}_j \epsilon_{t-j}$. Since  
$\hat{\mathbf{z}}$ is a series of white
noises \citep[pp.\ 240]{Brockwell_Davis_1991}, \eqref{eq:likelihood4} can be 
viewed as a linear model with unknown coefficients $(\mathbf{s}, 
\boldsymbol\mu)$.

Following the proof of Lemma \ref{lemma:sigmasq1}, the BMDL estimators 
for $\mathbf{s}$ and $\boldsymbol\mu$ have the following limits:
\begin{align*}
\hat{\mathbf{s}}_N 
&	= (\widehat{\mathbf{A}}'\widehat{\mathbf{B}}
	\widehat{\mathbf{A}})^{-1}
	(\widehat{\mathbf{A}}'\widehat{\mathbf{B}}\widehat{\mathbf{X}})\\
& 	= \left\{ \widehat{\mathbf{A}}'
	\left(\mathbf{I}_{N-p} - \mathcal{P}_{\widehat{\mathbf{D}}}\right)
	\widehat{\mathbf{A}}  \right\}^{-1}
	\left\{ \widehat{\mathbf{A}}'
	\left(\mathbf{I}_{N-p} - \mathcal{P}_{\widehat{\mathbf{D}}}\right)
	\widehat{\mathbf{X}}  \right\} + O_P\left(\frac{1}{N}\right),\\
\hat{\boldsymbol\mu}_N
& 	= \left(\widehat{\mathbf{D}}' 
	\widehat{\mathbf{D}}+\frac{\mathbf{I}_{m}}{\nu}\right)^{-1}
	\widehat{\mathbf{D}}' \left( \widehat{\mathbf{X}} - \widehat{\mathbf{A}}\hat{\mathbf{s}}_N \right)\\
&	= \left(\widehat{\mathbf{D}}' 
	\widehat{\mathbf{D}}\right)^{-1}
	\widehat{\mathbf{D}}' \left( \widehat{\mathbf{X}} - \widehat{\mathbf{A}}\hat{\mathbf{s}}_N \right)
	+ O_P\left(\frac{1}{N}\right).
\end{align*}
After rewriting \eqref{eq:likelihood4} as
\begin{align*}
\widehat{\mathbf{X}}
=	&\left\{\left(\mathbf{I}_{N-p} - \mathcal{P}_{\widehat{\mathbf{D}}}\right)
	\widehat{\mathbf{A}}\right\} \mathbf{s} 
	+ \left\{\widehat{\mathbf{D}}\boldsymbol\mu 
	+ \mathcal{P}_{\widehat{\mathbf{D}}} \widehat{\mathbf{A}}\mathbf{s} \right\}
	+ \hat{\mathbf{z}}\\
=	& 	\left\{\left(\mathbf{I}_{N-p} - \mathcal{P}_{\widehat{\mathbf{D}}}\right)
	\widehat{\mathbf{A}}\right\} \mathbf{s} 
	+ \widehat{\mathbf{D}}\left\{ \boldsymbol\mu +  
	\left(\widehat{\mathbf{D}}' \widehat{\mathbf{D}}\right)^{-1}\widehat{\mathbf{D}}' 
	\widehat{\mathbf{A}} \mathbf{s}\right\} + \hat{\mathbf{z}},
\end{align*}
it is not hard to see that $\hat{\mathbf{s}}_N$ and 
$\hat{\boldsymbol\mu}_N$ are the least square estimators of this linear 
model. Since least square estimators are asymptotically consistent, 
$\hat{\mathbf{s}}_N \stackrel{P}{\longrightarrow} \mathbf{s}^0$ and 
$\hat{\boldsymbol\mu}_N \stackrel{P}{\longrightarrow} \boldsymbol\mu^0$. 
\end{proof}

\section{Additional Simulations and Real Examples}
\label{sec:appendix_examples}

\subsection{Simulation Examples} Additional figures related to our 
simulation examples in Section \ref{sec:simulation} are included here.

\begin{figure}
\centering
\includegraphics[width = 0.6\textwidth, angle = 270]
{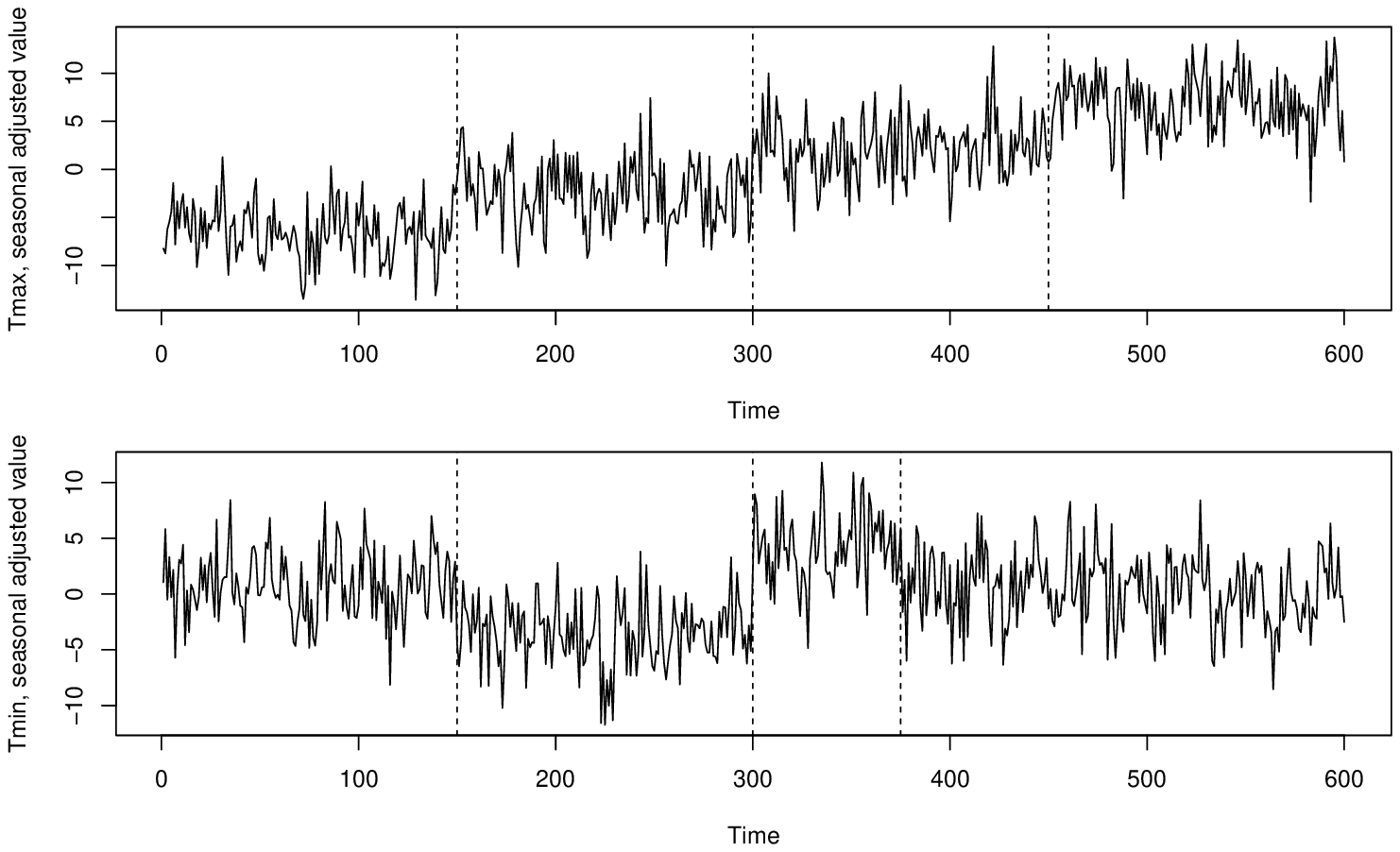}
\caption{\label{fg:simulation_sample2}
The Figure \ref{fg:simulation_sample1} series after subtracting 
sample monthly means. Vertical dashed lines mark true changepoint times.}
\end{figure}

\begin{figure}
\centering
\includegraphics[width = 0.5\textwidth, angle = 270]
{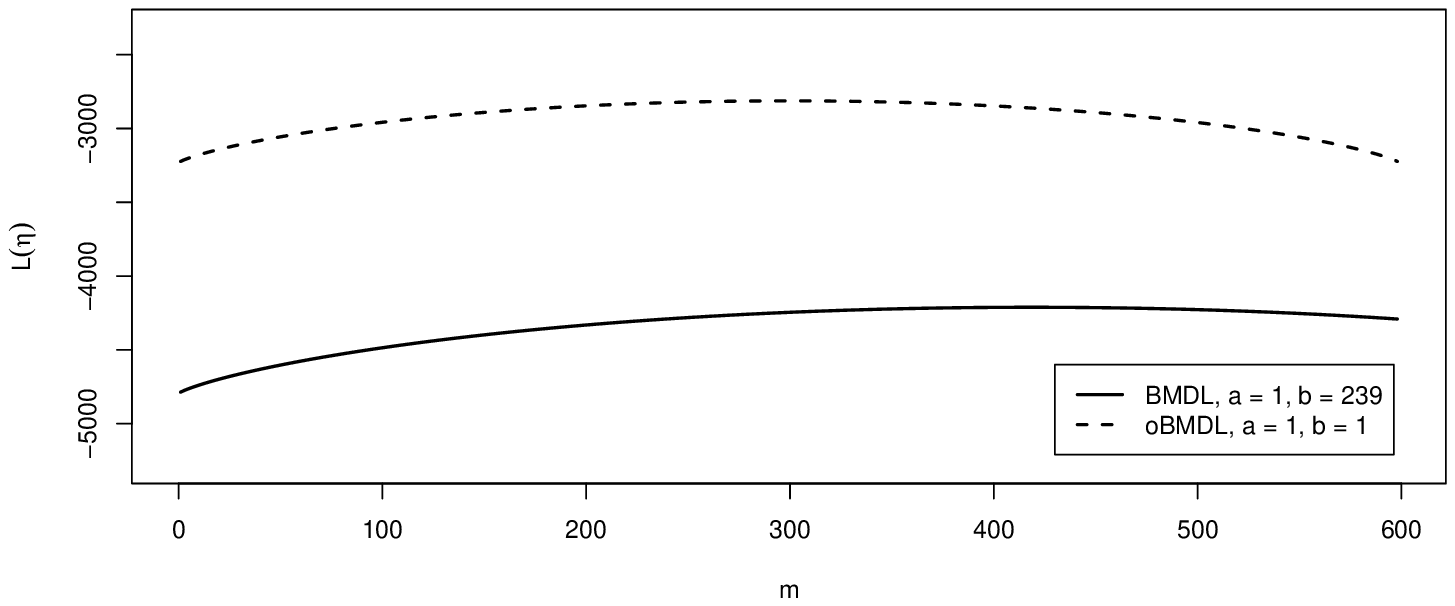}
\caption{\label{fg:BMDL_penalties}
Model code lengths $\mathcal{L}(\boldsymbol\eta) = -\log\Gamma\left(a + m\right) 
-\log\Gamma\left(b + N -p - m\right)$ between the BMDL and the oBMDL.}
\end{figure}

\subsection{Tuscaloosa Data Analysis:  Target Minus Reference}
\label{subsec:target_minus_reference}

\begin{figure}
\centering
\includegraphics[width = 0.6\textwidth, angle = 270]
{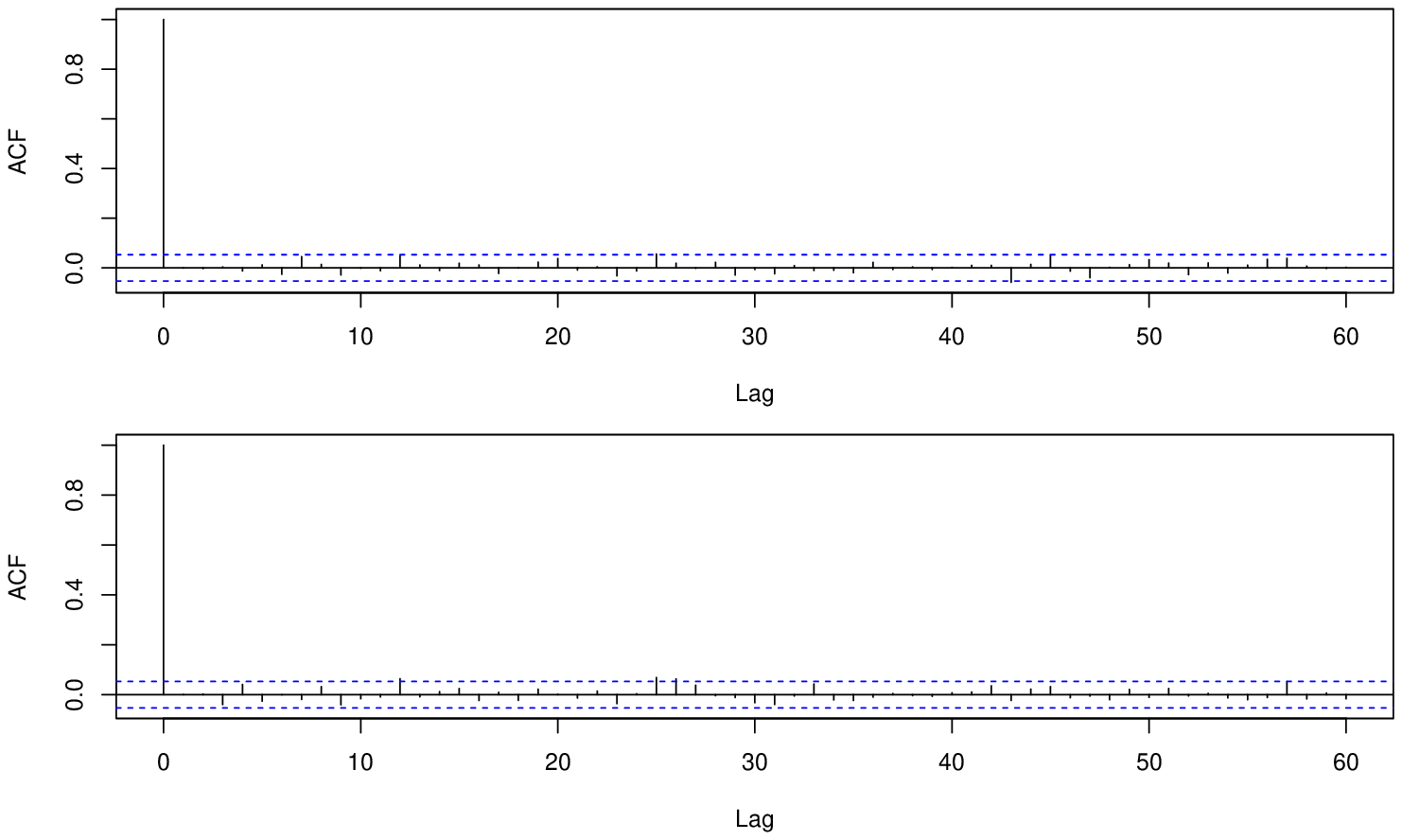}
\caption{\label{fg:Tuscaloosa_acf}
Sample model residual autocorrelations for Tmax (top panel) and Tmin 
(bottom panel), fitted using the univariate BMDL with metadata and $p=2$.}
\end{figure}

A reference series is a record from a station near the target station 
that is subtracted from the target series.  The idea is that two nearby 
stations should experience similar weather; hence, any trends or 
seasonal cycles should be lessened (if not altogether removed) in the 
target minus reference subtraction. Changepoints caused by artificial 
reasons, rather than by real climate changes, are easier to detect 
(visually) in target minus reference comparisons. Following 
\citet{Lu_etal_2010}, our reference series is obtained by averaging 
three nearby stations: Aberdeen, MS; Greensboro, AL; and Selma, AL.  By 
averaging multiple reference series (this is called a composite 
reference), impacts of mean shifts in any of the individual stations in 
the composite reference are lessened.

\begin{figure}
\centering
\includegraphics[width = 0.6\textwidth, angle = 270]
{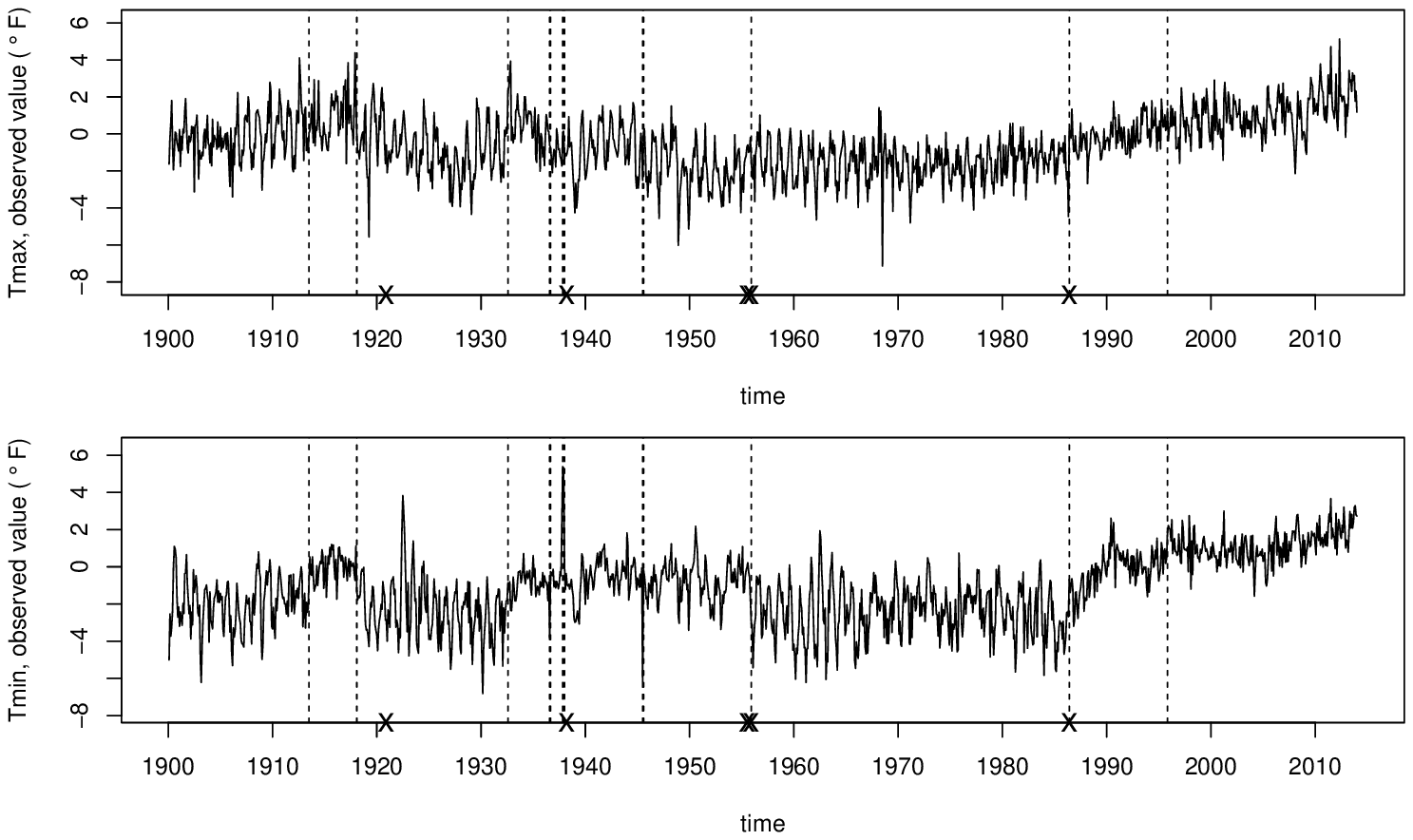}
\caption{\label{fg:Tuscaloosa3} 
Target minus reference Tmax (top panel) and Tmin (bottom panel) 
series. Metadata times for Tuscaloosa are marked with crosses on the 
axis. Vertical dashed lines show estimated changepoint times from our 
methods. } 
\end{figure}

Figure \ref{fg:Tuscaloosa3} shows the optimal changepoint configuration 
for the target minus reference series and contains 12 concurrent 
changes: June 1914, January 1919, July 1933, July 1937, August 1937, 
October 1938, December 1938, June 1946, July 1946, November 1956, May 
1987, and October 1996.  Among them, the 1956 and 1987 changepoints are 
in the metadata; the two changepoints in 1938 are close to the 1939 
station relocation. The changepoints in 1919, 1933, and 1990 are also 
flagged by \citet{Lu_etal_2010}. One of the shifts, November 1956, moves 
the Tmax series warmer and the Tmin series colder.

The October and December 1938 changepoints are likely due to typos in 
the data record. Specifically, the October and November 1938 Tmin values 
in the target minus reference series appear to be abnormally high. While 
the data have been quality checked, some errors persist.  This 
conjecture is made because the three reference stations lie in various 
directions from Tuscaloosa; climatologically, series to the north and 
west of Tuscaloosa should be cooler and those to the south and east 
should be warmer. In this case, Tuscaloosa was significantly warmer than 
all three references.  Similar statements apply to the two ``outlier'' 
changepoints in 1937, and the two changepoints in 1946, where the Tmin 
records for Tuscaloosa are lower than those for all three reference 
stations.  It is interesting that our method picked up outliers.

It is natural to flag more changepoints in the target minus reference 
series than the target series alone.  An ideal reference series should 
have the same trend and seasonal cycles as the target series and be free 
of artificial mean shifts.  This said, we do not assume that the target 
minus reference comparison completely removes the monthly mean cycle; 
indeed, \cite{Liu_etal_2016} shows that this is seldom the case. 
Reference series selection is a problem currently studied by 
climatologists.  As our reference series averages three neighbor 
stations, mean shifts in any of the reference records may induce shifts 
in the target minus reference series. For example, the estimated 
changepoint in 1914 is close to the 1915 metadata time listed in the 
Aberdeen reference.  This said, averaging three neighbors should help 
mitigate the effects of changepoints in any individual reference series.



\end{document}